\documentclass[11pt]{article}
\usepackage[letterpaper,margin=1in]{geometry}
\usepackage[T1]{fontenc}
\usepackage{lmodern}
\usepackage{amsmath,amssymb,amsthm,mathtools,mathrsfs}
\usepackage{booktabs,longtable,tabularx,microtype}
\usepackage{graphicx}
\usepackage{tikz,pgfplots}
\pgfplotsset{compat=1.18}
\usetikzlibrary{arrows.meta,positioning,fit,calc}
\usepackage{xcolor}
\usepackage{colortbl,multirow,hhline}
\usepackage[inline]{enumitem}
\usepackage{cite}
\usepackage{algorithm}
\usepackage{algpseudocode}
\algrenewcommand{\algorithmicrequire}{\textbf{Input:}}
\algrenewcommand{\algorithmicensure}{\textbf{Output:}}
\algrenewcommand{\algorithmicreturn}{\textbf{return}}
\usepackage{caption}
\usepackage{placeins}
\usepackage{afterpage}
\usepackage[hidelinks]{hyperref}
\usepackage{aliascnt}
\usepackage[capitalise,nameinlink,noabbrev]{cleveref}

\crefformat{equation}{#2(#1)#3}
\Crefformat{equation}{#2(#1)#3}
\crefrangeformat{equation}{#3(#1)#4--#5(#2)#6}
\Crefrangeformat{equation}{#3(#1)#4--#5(#2)#6}
\crefmultiformat{equation}{#2(#1)#3}{ and #2(#1)#3}{, #2(#1)#3}{, and #2(#1)#3}
\Crefmultiformat{equation}{#2(#1)#3}{ and #2(#1)#3}{, #2(#1)#3}{, and #2(#1)#3}
\crefrangemultiformat{equation}
  {#3(#1)#4--#5(#2)#6}{ and #3(#1)#4--#5(#2)#6}
  {, #3(#1)#4--#5(#2)#6}{, and #3(#1)#4--#5(#2)#6}
\Crefrangemultiformat{equation}
  {#3(#1)#4--#5(#2)#6}{ and #3(#1)#4--#5(#2)#6}
  {, #3(#1)#4--#5(#2)#6}{, and #3(#1)#4--#5(#2)#6}

\newtheorem{theorem}{Theorem}[section]
\newtheorem*{restatedtheorembase}{\restatedtheoremname}
\newenvironment{restatedtheorem}[1]
 {\begingroup\def\restatedtheoremname{\Cref{#1}}\begin{restatedtheorembase}}
 {\end{restatedtheorembase}\endgroup}
\newaliascnt{proposition}{theorem}
\newtheorem{proposition}[proposition]{Proposition}
\aliascntresetthe{proposition}
\newaliascnt{corollary}{theorem}
\newtheorem{corollary}[corollary]{Corollary}
\aliascntresetthe{corollary}
\newaliascnt{lemma}{theorem}
\newtheorem{lemma}[lemma]{Lemma}
\aliascntresetthe{lemma}
\theoremstyle{definition}
\newaliascnt{definition}{theorem}
\newtheorem{definition}[definition]{Definition}
\aliascntresetthe{definition}
\theoremstyle{remark}
\newaliascnt{remark}{theorem}

\aliascntresetthe{remark}
\crefname{proposition}{Proposition}{Propositions}
\Crefname{proposition}{Proposition}{Propositions}
\crefname{lemma}{Lemma}{Lemmas}
\Crefname{lemma}{Lemma}{Lemmas}
\crefname{corollary}{Corollary}{Corollaries}
\Crefname{corollary}{Corollary}{Corollaries}
\newcommand{\R}{\mathbb R}
\newcommand{\B}{\mathcal B}
\newcommand{\ones}{\boldsymbol 1}
\DeclareMathOperator{\OPT}{OPT}
\allowdisplaybreaks[2]

\newcolumntype{Y}{>{\raggedright\arraybackslash}X}
\hypersetup{pdftitle={On the Exact Complexity of Hylland--Zeckhauser Equilibria},pdfauthor={Zhengyang Liu, Yixin Tao, Yonglei Yan}}
\newcommand{\HZUpperStatement}{Given rational data $v\in[0,1]^{n\times n}$ and $\delta>0$, a deterministic
algorithm computes a rational pair $(X,\mathbf p)$ forming a
$(0.0755211+\delta)$-approximate HZ equilibrium.
The running time is polynomial in the input length and $1/\delta$.
}

\newcommand{\HZExactStatement}{Computing an exact Hylland--Zeckhauser equilibrium for a market with
rational utilities is FIXP-complete under the separable linear reduction.
The hardness already holds when only the equilibrium allocation is requested.
}

\begin{document}
\hypersetup{pageanchor=false}
\begin{titlepage}
\title{Settling the Complexity of Hylland--Zeckhauser Equilibria}
\author{Zhengyang Liu\thanks{Beijing Institute of Technology. \texttt{zhengyang@bit.edu.cn}}
\and Yixin Tao\thanks{Shanghai University of Finance and Economics. \texttt{taoyixin@mail.shufe.edu.cn}}
\and Yonglei Yan\thanks{Beijing Institute of Technology. \texttt{yanyonglei@bit.edu.cn}}}
\date{}
\maketitle
\begin{abstract}
The Hylland--Zeckhauser (HZ) mechanism is a foundational framework for allocating indivisible goods in one-sided matching markets via lotteries supported by competitive prices and equal artificial budgets. We resolve the exact computational complexity of the HZ scheme by proving that computing an exact HZ equilibrium is FIXP-complete, thereby settling a conjecture of Vazirani and Yannakakis. Notably, this hardness persists even when computing the equilibrium allocation alone. We establish our hardness result via a reduction that encodes quadratic fixed-point equations into a reciprocal variational inequality. On the algorithmic side, we present a deterministic polynomial-time approximation algorithm that achieves an additive error of $0.0755211 + \delta$ for any $\delta > 0$, significantly improving upon the previous $1/e$ guarantee.
\end{abstract}
\end{titlepage}
\pagenumbering{roman}
\hypersetup{pageanchor=true}
\begingroup
\small
\tableofcontents
\endgroup
\clearpage
\pagenumbering{arabic}
\section{Introduction}\label{sec:intro}
\afterpage{\clearpage\begin{table}[H]
\centering
\newsavebox{\HZResultsBox}
\begin{lrbox}{\HZResultsBox}
\begin{minipage}{\linewidth}
\normalsize
\color{black}
\captionsetup{font=small,labelfont=bf,skip=7pt,
  justification=centering,singlelinecheck=true}
\caption{Complexity and approximation guarantees for general HZ markets.}
\label{tab:hz-results}
\setlength{\tabcolsep}{7pt}
\setlength{\arrayrulewidth}{.35pt}
\renewcommand{\arraystretch}{1.15}
\renewcommand{\multirowsetup}{\centering}
\newcommand{\HZCurrentCell}{\cellcolor{black!8}}
\newcommand{\HZSource}[1]{{\footnotesize #1}}
\newcommand{\HZOwnSource}[1]{{\footnotesize\textbf{This work} (\Cref{#1})}}
\begin{tabularx}{\linewidth}{>{\centering\arraybackslash}p{.17\linewidth}>{\centering\arraybackslash}p{.18\linewidth}>{\centering\arraybackslash}X}
\arrayrulecolor{black}
\specialrule{.6pt}{0pt}{0pt}
\multirow{4}{=}{\textbf{Exact}}
  & \multirow{2}{=}{Membership}
  & FIXP-membership \\[-2pt]
  & & \HZSource{\cite{VaziraniYannakakis2025}} \\[2pt]
\hhline{~--}
  & \multirow{2}{=}{Hardness}
  & \HZCurrentCell\textbf{FIXP-hard} \\[-2pt]
  & & \HZCurrentCell\HZOwnSource{thm:exact-completeness} \\[2pt]
\specialrule{.6pt}{0pt}{0pt}
\multirow{8}{=}{\textbf{Approximate}}
  & \multirow{4}{=}{Algorithms}
  & Polynomial time with $\epsilon=1/e$ \\[-2pt]
  & & \HZSource{Our earlier version \cite{YanLiu2026}} \\[2pt]
\hhline{~~-}
  & & \HZCurrentCell{\bfseries\boldmath Polynomial time with $\epsilon=0.0755211+\delta$} \\[-2pt]
  & & \HZCurrentCell\HZOwnSource{thm:upper} \\[2pt]
\hhline{~--}
  & \multirow{4}{=}{Hardness}
  & PPAD-hard for $\epsilon=n^{-c}$ \\[-2pt]
  & & \HZSource{\cite{ChenChenPengYannakakis2022}} \\[2pt]
\hhline{~~-}
  & & PPAD-hard for some fixed $\epsilon>0$ \\[-2pt]
  & & \HZSource{\cite{BravermanLiuXueZhou2026,OpenAI2026PCP}} \\[2pt]
\specialrule{.6pt}{0pt}{0pt}
\end{tabularx}
\par\vspace{5pt}
\end{minipage}
\end{lrbox}
\usebox{\HZResultsBox}
\end{table}}
In a {\em one-sided matching market}, $n$ agents must be matched to 
$n$ indivisible goods based on their preferences, with each 
agent receiving exactly one good. Such settings frequently arise 
in practice, such as assigning students to dormitory rooms or 
workers to job positions, where monetary transfers are unavailable. 
The central objective is to find an assignment that respects agents' 
preferences while simultaneously ensuring fairness and efficiency 
\cite{HyllandZeckhauser1979,BogomolnaiaMoulin2001}.

When multiple agents covet the same goods, randomization offers 
a natural mechanism for sharing access to desirable outcomes. 
A random assignment specifies the probability that each agent 
receives each good, forming a fractional matching. By the 
Birkhoff--von Neumann theorem, any such fractional matching 
can be implemented as a lottery over deterministic matchings. 
Assuming agents evaluate these lotteries through von 
Neumann--Morgenstern utilities, their expected utility is 
the probability-weighted sum of their valuations. These cardinal 
preferences capture nuances that ordinal rankings alone cannot 
express: two agents may rank the goods identically, yet differ 
substantially in how willing they are to accept a second choice 
in place of their first. The allocation problem is thus to select 
feasible lotteries that accommodate these intensity differences 
while treating all agents fairly.

Hylland and Zeckhauser \cite{HyllandZeckhauser1979} addressed 
this challenge by introducing a market for probability shares. 
In the Hylland--Zeckhauser (HZ) mechanism, each good is viewed 
as one unit of divisible probability shares, and every agent is 
endowed with an artificial budget of one. Agents face common 
nonnegative prices and choose lotteries that maximize their 
expected utility subject to their budget and the requirement 
that their probabilities sum to one. An HZ equilibrium consists 
of prices and a fractional matching such that each agent 
receives an optimal affordable lottery and every good is fully 
allocated. Here, market prices coordinate the agents' choices, 
while equal budgets grant everyone access to an identical set 
of affordable lotteries.

Hylland and Zeckhauser~\cite{HyllandZeckhauser1979} proved that an equilibrium always exists. The resulting allocation is ex ante 
envy-free, as every agent can afford every other agent's lottery 
and weakly prefers their own. Furthermore, by selecting 
minimum-cost bundles among utility maximizers, the allocation 
is also ex ante Pareto efficient. This framework has since 
been extended to accommodate priorities 
\cite{HeMirallesPyciaYan2018} and broader feasibility constraints 
\cite{BudishCheKojimaMilgrom2013,EcheniqueMirallesZhang2021}. 
These strong normative guarantees make the HZ mechanism an 
attractive paradigm for fair allocation, though realizing its 
benefits relies on efficient algorithms for computing equilibrium 
prices and allocations.

\subsection{Hardness of Exact Computation}\label{sec:intro-exact}

At fixed prices, each agent's optimal lottery can be found by linear programming efficiently.
The core technical difficulty lies in finding prices and optimal lotteries that jointly satisfy the matching market-clearing constraints.
To understand the complexity of computing such a price-supported allocation, we first review the known background results (summarized alongside our contributions in \Cref{tab:hz-results}):
\begin{itemize}[label=-,leftmargin=*,itemsep=0.5em]
\item Vazirani and Yannakakis \cite{VaziraniYannakakis2025} prove that exact HZ equilibrium computation belongs to the class FIXP.
Introduced by Etessami and Yannakakis \cite{EtessamiYannakakis2010}, FIXP captures real-valued search problems polynomial-time reducible to computing exact fixed points of continuous self-maps represented by algebraic circuits.
Complete problems for FIXP include computing a three-player Nash equilibrium and finding equilibrium prices in an exchange economy with algebraic excess demand \cite{EtessamiYannakakis2010}.
\item Additionally, computing an approximate HZ equilibrium is PPAD-hard \cite{ChenChenPengYannakakis2022,BravermanLiuXueZhou2026,OpenAI2026PCP}.
Recall that PPAD corresponds to LinearFIXP, the restriction of FIXP to circuits using only linear arithmetic and $\max$/$\min$ operators; in particular, $\mathrm{PPAD}\subseteq\mathrm{FIXP}$ \cite{papadimitriou1994complexity,EtessamiYannakakis2010,ruta14}.
\end{itemize}
These results leave a different reduction to be constructed.
FIXP membership expresses HZ equilibria as algebraic fixed points;
FIXP-hardness requires encoding exact nonlinear fixed-point equations
inside an HZ market. The PPAD reduction of Chen et al.\ 
\cite{ChenChenPengYannakakis2022} uses first-order approximations of
demand and discrete comparisons to represent a threshold game.
Those approximations allow small errors, whereas recovering an exact
three-player Nash equilibrium requires the encoded arithmetic
equations to hold exactly.

Vazirani and Yannakakis state the following conjecture \cite{VaziraniYannakakis2025}:
\begin{quote}
``We conjecture that this problem is FIXP-complete; resolving this remains a challenging open problem.''
\end{quote}
The question of FIXP-completeness has also been repeatedly highlighted as an open problem in the literature \cite{ChenChenPengYannakakis2022,FilosRatsikasHansenHoghHollender2023}.
In this paper, we resolve this open problem in the affirmative:
\begin{theorem}\label{thm:exact-completeness}
\HZExactStatement
\end{theorem}
We reduce three-player Nash equilibrium to quadratic fixed points,
then to a reciprocal box variational inequality.
Small groups of reciprocal terms implement arithmetic, and feedback
imposes the fixed-point equations. Using an output changes its node
balance, so we bound each group's error and input contributions by
its output contribution. Summing the bounds forces all errors to
vanish. Within the graph construction, we replace the temporarily
fixed node by an interval variable whose box condition forces the
required value. Finally, node agents
determine price ratios in an HZ market, edge agents implement the
reciprocal terms through their budget equations, and market clearing
gives the box conditions. Paired graph variables have a fixed sum,
allowing each Nash probability to be recovered as an affine function
of one allocation entry. Every stage therefore preserves the
separable linear form of the reduction.

Under FIXP reductions, exact HZ equilibrium computation is equivalent to the three-player Nash equilibrium problem and the exchange-equilibrium problem with algebraic excess-demand circuits \cite{EtessamiYannakakis2010}.
Thus, linear utilities, identical budgets, and unit demand already suffice for this assignment model to express the full complexity of algebraic fixed-point computation.
To the best of our knowledge, this is the first FIXP-hardness result for one-sided assignment markets with exogenously fixed budgets.

\subsection{Approximation Guarantees}\label{sec:intro-approximation}

We also give a polynomial-time algorithm for approximate HZ equilibria.
Having classified the exact problem, we ask how accurately the equilibrium conditions can be satisfied in polynomial time while keeping the matching constraints exact.
As noted above, computing an approximate HZ equilibrium is PPAD-hard even at a fixed positive error tolerance.
Chen, Chen, Peng, and Yannakakis \cite{ChenChenPengYannakakis2022} prove PPAD-hardness for inverse-polynomial error, even when each agent has at most four utility values.
Subsequently, Braverman, Liu, Xue, and Zhou \cite{BravermanLiuXueZhou2026} establish hardness at some fixed positive error under the PCP-for-PPAD conjecture.
Combining their reduction with the subsequent proof of that conjecture in \cite{OpenAI2026PCP} yields unconditional PPAD-hardness at constant error.
Hardness at one error tolerance does not rule out efficient algorithms at larger tolerances; the question is how small an error such an algorithm can guarantee.

In this paper, we consider the same notion of additive approximation as Chen et al.\ \cite{ChenChenPengYannakakis2022}, for utilities in $[0,1]$.
For $\epsilon\ge0$, an $\epsilon$-approximate HZ equilibrium consists of a fractional matching and nonnegative prices, with the minimum price normalized to zero.
Each agent still receives total probability one, and every good is allocated in full, so the output can be implemented as a lottery over perfect matchings.
Each agent's lottery costs at most $1+\epsilon$ and gives expected utility no more than $\epsilon$ below the maximum attainable with budget one at the same prices.
The budget-one benchmark gives every agent the same set of comparison lotteries, and the utility bound limits each agent's loss relative to its own preferred choice in this set.
\Cref{def:approx-hz} gives the formal definition.
We ask:
\begin{quote}
How small an additive error can be achieved in polynomial time for general HZ markets?
\end{quote}

An earlier version of this work \cite{YanLiu2026} gives a polynomial-time
algorithm with additive error $1/e$ by replacing each agent with
copies whose binary utilities indicate whether a good meets a
specified utility threshold. The copy proportions depend only on
the utility levels. An exact equilibrium of the expanded binary
market can be aggregated into a feasible allocation, but the
aggregation can lose utility approaching $1/e$. This loss occurs
even without rounding, so increasing the number of copies alone
does not improve the worst-case guarantee. We obtain the following
smaller error by choosing assignment weights that depend on the
allocated lotteries.
\begin{theorem}\label{thm:upper}
\HZUpperStatement
\end{theorem}
We design assignment weights that depend on the probability assigned
above each utility level. The challenge is to strengthen the utility guarantee while
keeping the assigned mean weight fixed and the allocation computable.
We add pairwise correction terms that cancel in the assigned mean,
preserve convexity in utility, and keep the negative weight map
monotone. Monotonicity permits polynomial-time approximation of a
variational inequality; assignment duality then converts its solution
into common prices. Convexity reduces the utility comparison with all
affordable lotteries to a coefficient gap. Finally, tail averages
reduce the required gap, for every market size, to one inequality in
three real variables. The interval bounds for that inequality give
the stated approximation constant.

At the computed prices, no agent can gain more than $0.0755211+\delta$ in expected utility by switching to a lottery of cost at most one.
This bound applies to every agent and every input utility matrix, with an error that does not grow with the market.
For illustration, with independent $U[0,1]$ utilities and small $\epsilon$, every $\epsilon$-approximate HZ equilibrium also has high mean utility with probability tending to one as the market grows.\footnote{For fixed $\epsilon>0$ and independent $U[0,1]$ utilities, every $\epsilon$-approximate HZ equilibrium has mean utility at least $\frac12+\frac{1}{2(1+\epsilon)}-\epsilon-o(1)$ with high probability as $n\to\infty$; see \Cref{prop:random-mean-utility}.
At $\epsilon=0.0755211$, the asymptotic lower bound is approximately $0.88937$.
For our algorithm, use $\epsilon=0.0755211+\delta$; this is a lower bound over approximate equilibria.}
For every input, the returned allocation remains a lottery over valid matchings, and each agent's artificial spending is at most $1+0.0755211+\delta$.

\section{Preliminaries}\label{sec:model}

\paragraph{Notation.}
Write $[n]=\{1,\ldots,n\}$. For a real number $a$, write
$a_+=\max\{a,0\}$. For vectors,
$\langle x,y\rangle=x\cdot y=\sum_jx_jy_j$; for matrices,
$\langle A,B\rangle=\sum_{i,j}A_{ij}B_{ij}$.
We use $\|x\|_1=\sum_j|x_j|$,
$\|x\|_2=(\sum_jx_j^2)^{1/2}$, and
$\|x\|_\infty=\max_j|x_j|$ for vector norms.
The Euclidean matrix norm is $(\sum_{i,j}A_{ij}^2)^{1/2}$;
an unsubscripted norm is Euclidean. When different norms are specified, the two sides use the output and input norms, respectively.

\paragraph{Agents and Utilities.}
We consider $n$ agents and $n$ unit-supply goods, both indexed by $[n]$. The input utility matrix is $v\in[0,1]^{n\times n}$. Agent $i$ has utility vector $\mathbf{v}_i=(v_{i1},\ldots,v_{in})$ and evaluates a bundle $\mathbf{x}$ by $\mathbf{v}_i\cdot\mathbf{x}$, where $v_{ij}$ represent the utility that agent $i$ derives from one unit of good $j$. When considering a single agent, we write $\mathbf v$ for that agent's utility vector. Unless stated otherwise, we assume $\min_j v_{ij}=0$ for every agent $i$ and $\max_j v_{ij}=1$ for every nonconstant utility row. Constant rows are therefore zero.

\paragraph{Allocations.}
Each agent receives a total quantity of one. Thus a bundle belongs to
\[
 \Delta_n=\left\{\mathbf{x}\in\R_{\ge0}^n:
                     \sum_{j=1}^n x_j=1\right\}.
\]
An allocation is a nonnegative matrix $X$, where $X_{ij}$ is the
quantity of good $j$ given to agent $i$ and $\mathbf{X}_i$ is that
agent's bundle. The feasible allocations form
\[
 \B_n=\left\{X\in\R_{\ge0}^{n\times n}:
     \sum_{j=1}^n X_{ij}=1\ (i\in[n]),\quad
     \sum_{i=1}^n X_{ij}=1\ (j\in[n])\right\}.
\]
The row constraints impose unit demand; the column constraints clear
the market exactly. By the Birkhoff--von Neumann theorem, every $X\in\B_n$ is a convex combination of
permutation matrices. It therefore describes a lottery over matchings
in which every agent receives one good and every good is assigned. For a single agent, $\mathbf{x}$ denotes the
assigned bundle and $\mathbf{y}$ a possible alternative. The support
of $\mathbf{x}$ is the set of goods with $x_j>0$.

\paragraph{Prices and Demand.}
All agents face the same price vector $\mathbf{p}\ge0$ and have budget
one. When $\min_jp_j\le1$, the largest utility agent $i$ can afford is
\begin{equation}\label{eq:demand-value}
 \OPT_i(\mathbf{p})=
 \max\{\mathbf{v}_i\cdot\mathbf{y}:
       \mathbf{y}\in\Delta_n,\ \mathbf{p}\cdot\mathbf{y}\le1\}.
\end{equation}
The condition $\min_jp_j\le1$ makes the affordable set nonempty. For any utility vector
$\mathbf{u}$, we also write $\OPT_{\mathbf{u}}(\mathbf{p})$ for the
same maximum with objective $\mathbf{u}\cdot\mathbf{y}$.

\paragraph{Equilibrium and Approximation.}
An equilibrium assigns each agent a best affordable bundle while
clearing every good. Any affordable utility-maximizing bundle is
admissible, as in
\cite{ChenChenPengYannakakis2022,BravermanLiuXueZhou2026}.

\begin{definition}[HZ Equilibrium {\cite{HyllandZeckhauser1979,VaziraniYannakakis2025,ChenChenPengYannakakis2022}}]\label{def:hz-equilibrium}
Let $X\in\B_n$ and $\mathbf p\ge0$. The pair $(X,\mathbf p)$
is an HZ equilibrium if it satisfies:
\begin{enumerate}[label=(\roman*),leftmargin=*]
\item \emph{Budget balance:}
$\mathbf p\cdot\mathbf X_i\le1$ for every agent $i$.
\item \emph{Utility maximization:}
$\mathbf v_i\cdot\mathbf X_i=\OPT_i(\mathbf p)$ for every agent $i$. \footnote{Following the footnote in~\cite{ChenChenPengYannakakis2022}, we do not require the minimum cost condition on allocations in~\cite{HyllandZeckhauser1979} for our hardness result.}
\end{enumerate}
\end{definition}

We relax the spending and utility conditions by a common additive
error. The comparison bundle has budget one, while unit demand
and market clearing remain exact. We also normalize the minimum price
to zero, following the approximation convention in
\cite{ChenChenPengYannakakis2022,BravermanLiuXueZhou2026}.

\begin{definition}[Approximate HZ Equilibrium {\cite{VaziraniYannakakis2025,ChenChenPengYannakakis2022}}]\label{def:approx-hz}
Let $\epsilon\ge0$, let $X\in\B_n$, and let $\mathbf p\ge0$ satisfy
$\min_jp_j=0$. The pair $(X,\mathbf p)$ is an
$\epsilon$-approximate HZ equilibrium if it satisfies:
\begin{enumerate}[label=(\roman*),leftmargin=*]
\item \emph{Budget balance:}
$\mathbf p\cdot\mathbf X_i\le1+\epsilon$ for every agent $i$.
\item \emph{Utility maximization:}
$\OPT_i(\mathbf p)-\mathbf v_i\cdot\mathbf X_i\le\epsilon$
for every agent $i$.
\end{enumerate}
\end{definition}

\paragraph{Nash Equilibrium and FIXP.}
Our hardness reduction starts from a finite game, so we recall the
equilibrium condition that the constructed market must enforce.
Player $i$ has pure strategies $[m_i]$ and a mixed strategy
$\xi_i\in\Delta_{m_i}$. For a profile $\xi=(\xi_i)_i$, let
$U_{ij}(\xi_{-i})$ be player $i$'s expected payoff from pure strategy
$j$ when the other players use $\xi_{-i}$. Thus its payoff at the
profile is $\sum_{t=1}^{m_i}\xi_{it}U_{it}(\xi_{-i})$.

\begin{definition}[Nash Equilibrium {\cite{Nash1951}}]\label{def:nash}
A mixed-strategy profile $\xi$ is a Nash equilibrium if
$U_{ij}(\xi_{-i})\le\sum_{t=1}^{m_i}\xi_{it}U_{it}(\xi_{-i})$
for every player $i$ and pure strategy $j\in[m_i]$.
\end{definition}

This condition rules out a profitable unilateral deviation. Exact
equilibria can require irrational probabilities even when payoffs are
rational, which motivates a complexity class for real-valued fixed-point
computation \cite{EtessamiYannakakis2010}.

\begin{definition}[FIXP {\cite{EtessamiYannakakis2010}}]\label{def:fixp}
FIXP consists of search problems reducible in polynomial time, under
separable linear reductions, to finding a fixed point of a continuous
self-map of a cube given by an algebraic circuit over rational constants
and $+,-,\times,/,\max,\min$. Every division is defined throughout the cube.
\end{definition}

In such a reduction, each source-solution coordinate is an affine
function of one target-solution coordinate; the instance map, selected
coordinates, and affine coefficients are computable in polynomial time.
Computing a three-player Nash equilibrium is FIXP-complete under these
reductions \cite{EtessamiYannakakis2010}.

\paragraph{Variational Inequalities.}
Both proofs use variational inequalities to express equilibrium
conditions on a convex feasible set. In the hardness construction they
describe the balance of reciprocal terms on intervals; in the algorithm
they express optimality under assignment weights evaluated at the
allocation itself.

\begin{definition}[Variational Inequality {\cite{HartmanStampacchia1966,FacchineiPang2003}}]\label{def:vi}
Let $P\subseteq\mathbb R^d$ be a nonempty closed convex set and
$F:P\to\mathbb R^d$ a map. The variational inequality for $(P,F)$
asks for $x\in P$ such that $\langle F(x),y-x\rangle\ge0$ for every
$y\in P$.
\end{definition}

\section{FIXP-Hardness}\label{sec:fixp}
We prove FIXP-hardness through three SL reductions, starting from
three-player Nash equilibrium and ending with an HZ equilibrium.
Combined with the FIXP membership result of Vazirani and Yannakakis
\cite{VaziraniYannakakis2025}, the reduction proves FIXP-completeness.

\begin{restatedtheorem}{thm:exact-completeness}
\HZExactStatement
\end{restatedtheorem}

\newcommand{\QFPProblem}{\textnormal{\textsc{Quadratic Fixed Point}}}
\newcommand{\ReciprocalVIProblem}{\textnormal{\textsc{Reciprocal VI}}}
\newcommand{\HZEquilibriumProblem}{\textnormal{\textsc{HZ Equilibrium}}}
\newcommand{\QFPHardnessStatement}{\QFPProblem{} is FIXP-hard under SL reductions.}
\newcommand{\ReciprocalVIHardnessStatement}{\ReciprocalVIProblem{} is FIXP-hard under SL reductions.}
\newcommand{\HZEquilibriumHardnessStatement}{\HZEquilibriumProblem{} is FIXP-hard under SL reductions, even when only the equilibrium allocation is requested.}

\tikzset{
  fx node/.style={circle,draw,minimum size=8mm,inner sep=2pt,font=\small},
  fx output/.style={fx node,line width=.9pt},
  fx private/.style={fx node,fill=black!12},
  fx constant/.style={draw,rounded corners=1pt,minimum height=8mm,
    inner sep=4pt,font=\small,fill=black!5},
  fx edge/.style={-{Latex[length=2mm]},line width=.65pt},
  fx flow/.style={-{Stealth[length=2mm]},line width=.65pt},
  fx label/.style={font=\small,inner sep=2pt},
  fx role/.style={font=\small,align=center},
  fx gate/.style={draw,rounded corners=1pt,fill=black!3,
    minimum height=10mm,text width=2.4cm,align=center,inner sep=4pt,font=\small},
  fx group/.style={fx gate,line width=.9pt},
  fx panel/.style={font=\small\bfseries,anchor=west}
}

\subsection{Overview of the Reduction}\label{sec:fixp-overview}
We define three search problems and prove their FIXP-hardness in the order shown in \Cref{fig:fixp-reduction}. Each reduction is separable linear (SL): each coordinate of a solution to the source problem is an affine function of a single coordinate of a solution to the target problem, with indices and coefficients computable in polynomial time. The following three subsections present the constructions and theorem proofs in this sequence. Supporting lemmas state intermediate claims; their proofs, along with derivative, interval, and parameter bounds, appear in \Cref{app:fixp-estimates}.

\begin{figure}[!htb]
\centering
\begin{tikzpicture}[
  reduction stage/.style={draw,rounded corners=1pt,fill=black!3,
    minimum height=11mm,inner sep=3pt,font=\small,align=center}
]
  \node[reduction stage,text width=2.5cm] (rn) at (0,0)
    {\textsc{Three-Player}\\\textsc{Nash}};
  \node[reduction stage,text width=2.3cm,right=3mm of rn] (rx)
    {\textsc{Quadratic}\\\textsc{Fixed Point}};
  \node[reduction stage,text width=2.6cm,right=3mm of rx] (rr)
    {\textsc{Reciprocal VI}};
  \node[reduction stage,text width=2.9cm,right=3mm of rr] (rm)
    {\textsc{HZ Equilibrium}};
  \draw[fx flow] (rn) -- (rx);
  \draw[fx flow] (rx) -- (rr);
  \draw[fx flow] (rr) -- (rm);
\end{tikzpicture}
\caption{The Reduction from Nash Equilibrium to HZ Equilibrium}\label{fig:fixp-reduction}
\end{figure}

\paragraph{From \textnormal{\textsc{Three-Player Nash Equilibrium}} to \QFPProblem{} \textnormal{\mbox{(\Cref{sec:fixp-quadratic})}}.}
The first stage formulates the Nash equilibrium conditions as fixed-point equations over a hypercube.

\begin{definition}[\QFPProblem]\label{def:fixp-quadratic-problem}
\leavevmode\par
\noindent\textbf{Input:} A map $\mathcal F:[0,1]^d\to[0,1]^d$, where each coordinate is specified as a list of monomials of degree at most two, with rational coefficients encoded in binary and of absolute value at most one.

\noindent\textbf{Output:} A vector $x\in[0,1]^d$ satisfying $x=\mathcal F(x)$.
\end{definition}

\begin{theorem}\label{thm:fixp-quadratic}
\QFPHardnessStatement
\end{theorem}

\Cref{sec:fixp-quadratic} proves this theorem by constructing $\mathcal F$ from a given three-player game. The fixed-point equations enforce that designated coordinates sum to one for each player while ruling out any profitable deviation. Consequently, these coordinates directly yield a Nash equilibrium.

\paragraph{From \QFPProblem{} to \ReciprocalVIProblem{} \textnormal{\mbox{(\Cref{sec:fixp-network})}}.}
We employ a directed multigraph with positive real node values. Each node $i$ is associated with a variable $r_i$ constrained to a rational interval $[a_i,b_i]\subset(0,\infty)$. Each edge $e:i\to j$ carries a positive rational ratio $\lambda_e$; parallel edges contribute independently.
For a vector $r$ of node values, edge $e:i\to j$ contributes $I_e(r)$ at $i$ and $-I_e(r)$ at $j$. The total contribution at node $i$ is $\mathcal G_i(r)$, defined by
\begin{equation}\label{eq:fixp-edge-field}
 I_e(r)=\frac{\lambda_e}{r_i+\lambda_er_j},\qquad
 \mathcal G_i(r)=\sum_{e:i\to j}I_e(r)-\sum_{e:j\to i}I_e(r).
\end{equation}
The node values must satisfy the following \emph{box conditions}:
\begin{equation}\label{eq:fixp-box}
 \mathcal G_i(r)\le0\ (r_i=a_i),\qquad
 \mathcal G_i(r)=0\ (a_i<r_i<b_i),\qquad
 \mathcal G_i(r)\ge0\ (r_i=b_i).
\end{equation}
The input explicitly lists all edges. All rational interval endpoints and edge ratios have their numerators and denominators encoded in unary.

\begin{definition}[\ReciprocalVIProblem]\label{def:fixp-vi}
\leavevmode\par
\noindent\textbf{Input:} A directed multigraph with intervals $[a_i,b_i]$ satisfying $1\le a_i<b_i$ for all nodes, ratios $0<\lambda_e\le a_i$ for every edge $e:i\to j$, a unary integer $M\ge1$, and disjoint designated pairs $(A_j,D_j)$, $j\in[d]$, each with a designated ratio-one edge $A_j\to D_j$. The intervals of $A_j$ and $D_j$ are $[M,2M]$ and $[2M,3M]$, respectively. Every solution to \Cref{eq:fixp-box} satisfies $r_{A_j}+r_{D_j}=4M$ for each $j$.

\noindent\textbf{Output:} A vector $z\in[1,2]^d$ that extends to a solution of \Cref{eq:fixp-box} with $r_{A_j}=Mz_j$.
\end{definition}

\begin{theorem}\label{thm:fixp-vi}
\ReciprocalVIHardnessStatement
\end{theorem}

We carry out the reduction in three steps, corresponding to \Cref{sec:fixp-arithmetic,sec:fixp-gates,sec:fixp-integer}:
\begin{enumerate}[label=(\arabic*),leftmargin=*,itemsep=0.25em]
\item We construct a circuit comprising three gate types: harmonic, reciprocal, and feedback. After shifting $x$ to $z=\ones+x$, harmonic and reciprocal gates compute $\Phi(z)=\ones+\mathcal F(z-\ones)$ and $4-z_j$, while feedback gates link each computed coordinate back to its source to enforce $z=\Phi(z)$.
\item We implement each gate via a reciprocal subgraph, initially employing a shared node fixed at one. This intermediate graph enforces the box conditions at all other nodes. Linking subgraphs adds contributions to shared nodes, potentially perturbing gate outputs and necessitating a precise error analysis. We select parameters such that any nonzero arithmetic error would demand more balancing flow than the gates can supply. Consequently, every solution to the box conditions evaluates the circuit exactly.
\item We replace the fixed node with an interval variable whose incoming edges constrain its value to one. We then insert opposite ratio-one edges across each designated pair and scale all intervals by a common integer $M$, yielding the instance in \Cref{def:fixp-vi} with solution recovery $x_j=z_j-1$.
\end{enumerate}
Because the box conditions form a continuous variational inequality on a compact box, a solution is guaranteed to exist \cite{HartmanStampacchia1966,FacchineiPang2003}.

\paragraph{From \ReciprocalVIProblem{} to \HZEquilibriumProblem{} \textnormal{\mbox{(\Cref{sec:fixp-market})}}.}
The final stage constructs a market whose equilibrium allocations encode the designated graph coordinates. While supporting prices establish the underlying graph equations, the allocation alone suffices to recover the required coordinates.

\begin{definition}[\HZEquilibriumProblem]\label{def:fixp-hz-equilibrium}
\leavevmode\par
\noindent\textbf{Input:} A rational utility matrix $v\in[0,1]^{n\times n}$, encoded in binary.

\noindent\textbf{Output:} An allocation $X\in\B_n$ and prices $\mathbf p\ge0$ such that $(X,\mathbf p)$ satisfies the HZ equilibrium conditions in \Cref{def:hz-equilibrium}.
\end{definition}

\begin{theorem}\label{thm:fixp-market}
\HZEquilibriumHardnessStatement
\end{theorem}

\begin{proof}[Proof of \Cref{thm:exact-completeness}]
\Cref{thm:fixp-quadratic,thm:fixp-vi,thm:fixp-market} establish FIXP-hardness, even when restricted to finding the equilibrium allocation. Combined with the FIXP membership result of Vazirani and Yannakakis \cite{VaziraniYannakakis2025}, this completes the proof of both completeness claims, establishing \Cref{thm:exact-completeness}.
\end{proof}

\subsection{\texorpdfstring{From \textnormal{\textsc{Three-Player Nash Equilibrium}} to \QFPProblem}{From Three-Player Nash Equilibrium to Quadratic Fixed Point}}\label{sec:fixp-quadratic}
We reduce three-player Nash equilibrium to \QFPProblem{} by encoding probability normalization and payoff comparisons into a quadratic self-map of a hypercube. Finding a three-player Nash equilibrium is FIXP-complete under SL reductions \cite{EtessamiYannakakis2010}.

\begin{restatedtheorem}{thm:fixp-quadratic}
\QFPHardnessStatement
\end{restatedtheorem}

We first present the construction and subsequently establish its correctness.

\subsubsection{Construction}
We construct the quadratic map from the game's payoff tables in three steps: normalizing the payoffs, defining the cube coordinates, and specifying their simultaneous update rules.

\paragraph{1. Normalize the Payoffs.}
Let player $i$ have $m_i$ pure strategies, and set $m=\max_i m_i$. Normalize each nonconstant payoff table to $[0,1]$ via a positive affine transformation; replace any constant payoff table with zero. Denote by $u_i(j,s,t)$ the resulting payoff when player $i$ chooses strategy $j$ while the other two players choose $s$ and $t$, respectively. Unless noted otherwise, all summations range over the players' actual strategy sets.

\paragraph{2. Create the Coordinates.}
For each player--strategy pair $(i,j)$, we introduce four coordinates $\xi_{ij}, A_{ij}, \chi_{ij}, V_{ij} \in [0,1]$. For each player $i$, we also introduce a coordinate $B_i \in [0,1]$. The domain hypercube thus has dimension $d = 4(m_1+m_2+m_3) + 3$. The coordinates $\xi_{ij}$ represent strategy probabilities, while the remaining variables act as auxiliary coordinates in the update equations below.

\paragraph{3. Specify the Map and the Output Coordinates.}
Given a point in the current hypercube, we first evaluate the auxiliary expressions:
\[
  t_{ij} = A_{ij}/m - B_i, \qquad
  s_{ij} = (t_{ij} + \chi_{ij})/2, \qquad
  d_i = \frac{1}{m} \sum_j V_{ij}.
\]
Using these expressions, we define the updated coordinates simultaneously by:
\begin{equation}\label{eq:fixp-quadratic-map}
\begin{aligned}
  A'_{ij} &= \frac{1 + \sum_{s,t} u_i(j,s,t) \xi_{hs} \xi_{kt}}{m^2 + 1},
        && \{i,h,k\} = \{1,2,3\}, \\
  B'_i &= \frac{1}{m} \sum_j \xi_{ij} A_{ij}, &
  \chi'_{ij} &= \chi_{ij} + \frac{t_{ij}^2 - \chi_{ij}^2}{2}, \\
  V'_{ij} &= \frac{\xi_{ij} + s_{ij}^2}{2}, &
  \xi'_{ij} &= \xi_{ij}(1 - d_i) + V_{ij}/m.
\end{aligned}
\end{equation}
Expanding the right-hand sides into monomials yields an instance of \QFPProblem{} specified by its rational coefficients. The recovery map simply projects onto the strategy coordinates $\xi_{ij}$.

\subsubsection{Proof of Correctness}
We verify that the map self-maps the hypercube, that every fixed point induces a valid probability distribution over strategies, and that no player can profitably deviate.

\begin{proof}[Proof of \Cref{thm:fixp-quadratic}]
We first confirm that the map preserves the hypercube. The numerator of $A'_{ij}$ lies in $[1, m^2+1]$ because it contains at most $m^2$ payoff terms, each bounded in $[0,1]$; hence $A'_{ij} \in [0,1]$. Similarly, $B'_i, d_i \in [0,1]$. The bounds $|t_{ij}| \le 1$ and $|s_{ij}| \le 1$ imply $V'_{ij} \in [0,1]$. Rewriting the absolute-value update as $\chi'_{ij} = (\chi_{ij} - \chi_{ij}^2/2) + t_{ij}^2/2$ shows that both components lie in $[0, 1/2]$, guaranteeing $\chi'_{ij} \in [0,1]$. Finally, since $0 \le V_{ij}/m \le d_i$, we have $0 \le \xi'_{ij} \le \xi_{ij}(1-d_i) + d_i \le 1$. Thus, every coordinate remains within $[0,1]$. Furthermore, each right-hand side has degree at most two in the current coordinates. Expanding the squared terms yields polynomially many rational terms whose coefficients have polynomial bit lengths and absolute values at most one.

Next, we show that any fixed point yields normalized strategy probabilities. The fixed-point equation for $\chi_{ij}$ implies $\chi_{ij}^2 = t_{ij}^2$, so $\chi_{ij} = |t_{ij}|$ and $s_{ij} = (t_{ij})_+$. We now claim that $d_i > 0$. If $d_i = 0$, then all $V_{ij} = 0$, and the fixed-point relation $2V_{ij} = \xi_{ij} + s_{ij}^2$ forces $\xi_{ij} = s_{ij} = 0$ for all $j$. This implies $B_i = 0$, whereas $A_{ij} > 0$, yielding $t_{ij} > 0$ and $s_{ij} > 0$, a contradiction. Thus $d_i > 0$. Summing $\xi_{ij} d_i = V_{ij}/m$ over $j$ and dividing by $d_i$ then yields $\sum_j \xi_{ij} = 1$.

Finally, we rule out profitable deviations. Define $\mathcal H_i = \sum_j s_{ij}^2$. Summing the $V_{ij}$ fixed-point equations over $j$ and using normalization gives $2m d_i = 1 + \mathcal H_i$. Substituting $V_{ij} = m d_i \xi_{ij}$ back into each equation yields $\mathcal H_i \xi_{ij} = s_{ij}^2$. If $\mathcal H_i > 0$, every strategy played with positive probability ($\xi_{ij} > 0$) must have $s_{ij} > 0$ (and thus $t_{ij} > 0$). However, these strictly positive advantages contradict their zero expected value under $\xi_i$:
\[
  \sum_j \xi_{ij} t_{ij}
  = \frac{1}{m} \sum_j \xi_{ij} A_{ij} - B_i \sum_j \xi_{ij} = 0.
\]
Consequently, $\mathcal H_i = 0$, which implies $t_{ij} \le 0$ for all $j$, or equivalently $A_{ij} \le \sum_t \xi_{it} A_{it}$. Substituting the explicit formula for $A_{ij}$ from \Cref{eq:fixp-quadratic-map} cancels the common positive scale factor and additive constant, recovering precisely the Nash equilibrium conditions in \Cref{def:nash}.

The construction requires only polynomially many rational operations and runs in polynomial time. Each strategy probability corresponds to a coordinate $\xi_{ij}$ of the fixed point, making recovery a coordinate projection. This establishes an SL reduction from the FIXP-complete three-player Nash problem \cite{EtessamiYannakakis2010}.
\end{proof}
\subsection{\texorpdfstring{From \QFPProblem{} to \ReciprocalVIProblem}{From Quadratic Fixed Point to Reciprocal VI}}\label{sec:fixp-network}
We reduce \QFPProblem{} to \ReciprocalVIProblem{} by implementing quadratic arithmetic via reciprocal edge contributions and leveraging the box conditions to enforce exact circuit evaluation.

\begin{restatedtheorem}{thm:fixp-vi}
\ReciprocalVIHardnessStatement
\end{restatedtheorem}

The reduction proceeds in three sequential stages, structured across the subsequent subsections:
\begin{enumerate}[label=(\arabic*),leftmargin=*,itemsep=0.25em]
  \item \textbf{Circuit Formulation (\Cref{sec:fixp-arithmetic}):} We shift coordinates to $z = \ones + x$ and formulate the quadratic map $\mathcal F$ as an explicit arithmetic circuit built from three elementary gate types: harmonic, reciprocal, and feedback.
  \item \textbf{Reciprocal Subgraph Implementation (\Cref{sec:fixp-gates}):} We replace each gate with a reciprocal graph gadget anchored to a shared reference node fixed at one. We establish parameter bounds to ensure that the box conditions force exact gate evaluations despite network interaction errors.
  \item \textbf{Variable Unfixing and Instance Scaling (\Cref{sec:fixp-integer}):} We replace the fixed node with a variable constrained to one, introduce designated node pairs with ratio-one edges, and scale all graph intervals by a common integer $M$ to complete the target \ReciprocalVIProblem{} instance.
\end{enumerate}

\subsubsection{Implementing the Quadratic Map by Gates}\label{sec:fixp-arithmetic}
We construct an arithmetic circuit whose gate equations capture the targeted fixed-point system. We introduce source variables $z_1,\dots,z_d \in [1,2]$ alongside a fixed reference constant of value one. Setting $x = z - \ones$, the fixed-point condition $x = \mathcal F(x)$ translates to $z = \Phi(z)$, where $\Phi(z) = \ones + \mathcal F(z - \ones)$. A solution to the circuit assigns a value within its declared interval to each variable while satisfying every gate equation.

\paragraph{The Harmonic Gate.}
A harmonic gate receives inputs $A \in [a_1,b_1]$ and $B \in [a_2,b_2]$ with $0 < a_1 < b_1 < a_2 < b_2$, and is parametrized by $\rho \in (0,1)$. Its output is given by
\[
  H_\rho(A,B) = \left( \frac{\rho}{1+A} + \frac{1-\rho}{1+B} \right)^{-1} - 1.
\]
The declared output interval is set to $[H_\rho(a_1,a_2), H_\rho(b_1,b_2)]$, or any larger positive interval containing it.

\paragraph{The Reciprocal Gate.}
A reciprocal gate takes an input $A \in [a,b]$ with $a > 0$ and evaluates $J(A) = \frac{\kappa}{A+c} + \beta$, where $\kappa > 0$ and $c \ge 0$. Its declared output interval $[\ell,u] \subset (0,\infty)$ is positive and contains the range $[J(b), J(a)]$.

\paragraph{The Feedback Gate.}
A feedback gate receives a computed value $Y$ and ties it to an existing source variable $z \in [1,2]$. Given a parameter $\beta \ge 2$, the gate enforces $z = Y + 1 - \beta$. Once a coordinate evaluation yields $Y_j = \mathcal F_j(z - \ones) + \beta_j$, this gate imposes the fixed-point condition $z_j = \Phi_j(z)$ by routing $Y_j$ back to $z_j$.
\paragraph{Arithmetic Gadgets.}
We combine harmonic and reciprocal gates to construct four primary arithmetic gadgets: affine transformations, addition, complementation, and multiplication. Each gadget operates on designated input variables, internal auxiliary variables, and a single terminal output variable. Connections between gates are established by assigning the output variable of one gate as the input variable to another; only a gadget's terminal output variable is accessible to downstream gadgets. Together, these gadgets evaluate the coefficients and monomials of the quadratic map, after which feedback gates enforce the corresponding fixed-point system.

The constructions below specify the intended values of output variables in the sequence their constituent gates are added. A real quantity $v$ (which may be negative) is stored within a strictly positive variable $X = v + b_X \in [1,U]$, where $b_X \in [0,U]$ is a known integer offset and $U \ge 2$ is an integer upper bound. For constant inputs, subsequent output intervals are chosen using a positive, nondegenerate interval containing the constant value; in particular, the scalar constant one is assigned the interval $[1,2]$.

\paragraph{Affine Functions and Addition.}
An affine gadget with target output $X = \alpha A + \beta$ ($\alpha > 0$) requires $X$ to remain positive over the entire domain $A \in [a,b]$. The gadget consists of a chain of two reciprocal gates: the first produces an intermediate output $A_1 = 1/A \in [1/b, 1/a]$, which is then fed into the second reciprocal gate to yield the final output $X = \alpha / A_1 + \beta \in [\alpha a + \beta, \alpha b + \beta]$.

Addition of positive inputs $A, B \in [1,U]$ is implemented using four gates. Set $K_1 = 2U$, $K_2 = 8U^2$, and parameter $\rho = K_1 / (K_1 + K_2)$. The gadget first creates intermediate reciprocal outputs $A_1$ and $B_1$, feeds both into a harmonic gate to produce $Z$, and finally routes $Z$ into a reciprocal gate yielding $X$:
\begin{equation}\label{eq:fixp-addition}
\begin{aligned}
 A_1 &= K_1/A - 1, & A_1 &\in [1, 2U - 1], \\
 B_1 &= K_2/B - 1, & B_1 &\in [8U - 1, 8U^2 - 1], \\
 Z &= H_\rho(A_1, B_1), & Z &\in \left[ \frac{K_1+K_2}{2U} - 1, \frac{K_1+K_2}{2} - 1 \right], \\
 X &= (K_1 + K_2) / (Z + 1), & X &\in [2, 2U].
\end{aligned}
\end{equation}
The evaluated result satisfies $X = A + B$. To obtain a scaled and shifted sum $\alpha(A + B) + \beta$, only the final reciprocal gate is modified: its numerator is set to $\alpha(K_1 + K_2)$ and its additive constant to $\beta$, with its declared output interval defined as the corresponding positive affine image of the sum interval.

For inputs with tighter bounds, the four-gate construction uses the same integer upper bound $U$, but computes intermediate intervals from the exact input endpoints. Specifically, inputs within $[a_-, a_+]$ and $[b_-, b_+]$ yield the sum interval $[a_- + b_-, a_+ + b_+]$. When adding stored quantities, their respective integer offsets are added directly.

\paragraph{Complement and Negation.}
The complement gadget computes $D - X$ for $X \in [a,b]$ and constant $D > b$ via a three-gate reciprocal chain:
\begin{equation}\label{eq:fixp-complement}
 t_1 = D^2/X - D, \qquad t_2 = D^2/t_1 + D, \qquad N_D(X) = D^2/t_2.
\end{equation}
Here, $N_D(X)$ denotes the terminal output variable. The intermediate and final declared output intervals are:
\[
 \left[ \frac{D(D-b)}{b}, \frac{D(D-a)}{a} \right], \qquad
 \left[ \frac{D^2}{D-a}, \frac{D^2}{D-b} \right], \qquad 
 [D-b, D-a],
\]
respectively. To output $\alpha(D - X) + \beta$ ($\alpha > 0$), the numerator of the final gate is scaled by $\alpha$ and the gate output is shifted by $\beta$, adjusting the declared interval accordingly.

Negation of a stored quantity $X = v + b_X \in [1,U]$ is realized using this complement gadget by selecting $D = 2U + 2$ and recording the updated offset $D - b_X$. The output variable $N_D(X) \in [D - U, D - 1]$ then represents $-v$. The parameter $D$ serves as the updated upper bound for the stored variable. General subtraction is executed by cascading a negation gadget into an addition gadget, accumulating offsets accordingly.

\Cref{fig:fixp-arithmetic} summarizes the structural diagrams for affine, addition, and complement gadgets. Each rectangle represents an individual gate labeled by type and output variable; directed edges indicate data dependencies. \Cref{sec:fixp-gates} details the graph-level realization of these interconnections via shared nodes in reciprocal graphs.

\begin{figure}[!htb]
\centering
\begin{tikzpicture}
  \node[fx panel] at (-.4,1.05) {(a) Affine Function: $X=\alpha A+\beta$};
  \node[fx node] (aa) at (0,0) {$A$};
  \node[fx gate] (ai) at (3,0) {Reciprocal\\$1/A$};
  \node[fx gate] (af) at (8,0) {Reciprocal\\$\alpha/A_1+\beta$};
  \node[fx output] (ax) at (12.5,0) {$X$};
  \draw[fx flow] (aa) -- (ai);
  \draw[fx flow] (ai) -- node[fx label,above] {$A_1$} (af);
  \draw[fx flow] (af) -- (ax);

  \node[fx panel] at (-.4,-1.4) {(b) Addition: $X=A+B$};
  \node[fx node] (sa) at (0,-2.4) {$A$};
  \node[fx node] (sb) at (0,-4) {$B$};
  \node[fx gate] (sa1) at (2.6,-2.4) {Reciprocal\\$K_1/A-1$};
  \node[fx gate] (sb1) at (2.6,-4) {Reciprocal\\$K_2/B-1$};
  \node[fx gate] (sh) at (6.3,-3.2) {Harmonic\\$H_\rho(A_1,B_1)$};
  \node[fx gate] (sf) at (10,-3.2) {Reciprocal\\$\dfrac{K_1+K_2}{Z+1}$};
  \node[fx output] (sx) at (12.5,-3.2) {$X$};
  \draw[fx flow] (sa) -- (sa1);
  \draw[fx flow] (sb) -- (sb1);
  \draw[fx flow] (sa1) -- node[fx label,above=3pt] {$A_1$} (sh.160);
  \draw[fx flow] (sb1) -- node[fx label,below=3pt] {$B_1$} (sh.200);
  \draw[fx flow] (sh) -- node[fx label,above] {$Z$} (sf);
  \draw[fx flow] (sf) -- (sx);

  \node[fx panel] at (-.4,-5.25) {(c) Complement: $N_D(X)=D-X$};
  \node[fx node] (cx) at (0,-6.25) {$X$};
  \node[fx gate] (ct1) at (2.6,-6.25) {Reciprocal\\$D^2/X-D$};
  \node[fx gate] (ct2) at (6.3,-6.25) {Reciprocal\\$D^2/t_1+D$};
  \node[fx gate] (cn) at (10,-6.25) {Reciprocal\\$D^2/t_2$};
  \node[fx output] (cout) at (12.5,-6.25) {$N_D(X)$};
  \draw[fx flow] (cx) -- (ct1);
  \draw[fx flow] (ct1) -- node[fx label,above] {$t_1$} (ct2);
  \draw[fx flow] (ct2) -- node[fx label,above] {$t_2$} (cn);
  \draw[fx flow] (cn) -- (cout);
\end{tikzpicture}
\caption{Connecting Basic Gates into Arithmetic Gadgets}\label{fig:fixp-arithmetic}
\end{figure}

\paragraph{Multiplication.}
Multiplication forms the fourth arithmetic gadget alongside those in \Cref{fig:fixp-arithmetic}. It combines two parallel harmonic branches with a final addition gadget. Given inputs $X = v + b_X \in [1,U]$ and $Y = w + b_Y \in [1,U]$, the terminal output stores the product $vw$ with offset $b_W$. We set the structural constants to:
\[
 c = 12U + 5, \qquad a = (c + 1)/6 = 2U + 1, \qquad D = 6c + 2, \qquad b_W = aD + 1.
\]
The gadget is constructed in four sequential steps:
\begin{enumerate}[label=(\arabic*),leftmargin=*,itemsep=0.35em,beginpenalty=10000]
\item Compute $N_{2U+2}(Y)$ using a complement gadget. Apply two shifted addition gadgets to construct the branch inputs:
\[
 u_+ = X + Y + c - b_X - b_Y, \qquad
 u_- = X + N_{2U+2}(Y) + c - b_X + b_Y - (2U + 2).
\]
Declare both outputs in $\mathcal J = [c - 4U, c + 4U]$. Their intended values are $c + v + w$ and $c + v - w$, respectively.
\item For each variable $u_\pm$, attach a complement gadget with constant $2c + 1$. Rescale and shift the final gate to generate $V_\pm = 2N_{2c+1}(u_\pm) - 1$ in $[2c + 1 - 8U, 2c + 1 + 8U]$.
\item Feed $u_+$ and $V_+$ into a harmonic gate with parameter $\rho = 1/3$ to produce $R_+$. Analogously, feed $u_-$ and $V_-$ into an identical harmonic gate to produce $R_-$. Both outputs are declared in $[1, 3c]$.
\item Generate $N_D(R_+) \in [3c + 2, 6c + 1]$ via a complement gadget. Pass $N_D(R_+)$ and $R_-$ into an addition gadget configured such that its final gate evaluates $W = a(R_- + N_D(R_+)) + 1$. Declare $W \in [1, U_W]$ where $U_W = a(9c + 1) + 1$, and assign offset $b_W$. The terminal output variable $W$ represents the product $vw$.
\end{enumerate}
Only $W$ is passed to subsequent gadgets. Within each branch, $u_\pm$ feeds concurrently into its complement gadget and directly into its harmonic gate. Beyond this, the two branches contain disjoint sets of internal variables.

\Cref{fig:fixp-multiplication} illustrates the internal topology across panels (d-1)–(d-3): panel~(d-1) generates the branch inputs, panel~(d-2) details one harmonic branch, and panel~(d-3) combines their outputs.
\begin{figure}[!htb]
\centering
\begin{tikzpicture}
  \node[fx panel] at (-.4,7.9)
    {(d) Multiplication: $W=(X-b_X)(Y-b_Y)+b_W$};
  \begin{scope}[yshift=6.1cm]
  \node[fx panel] at (-.4,1.1) {(d-1) Forming the Two Branch Inputs};
  \node[fx node] (px) at (0,0) {$X$};
  \node[fx node] (py) at (2,-1.05) {$Y$};
  \node[fx group,text width=3cm] (pplus) at (8,0)
    {Shifted Addition\\shift $c-b_X-b_Y$};
  \node[fx group,text width=2.6cm] (pcomp) at (4.3,-2.5)
    {Complement\\$N_{2U+2}(Y)$};
  \node[fx group,text width=3.8cm] (pminus) at (9,-2.5)
    {Shifted Addition\\shift $c-b_X+b_Y-(2U+2)$};
  \node[fx output] (pup) at (13.6,0) {$u_+$};
  \node[fx output] (pum) at (13.6,-2.5) {$u_-$};
  \draw[fx flow] (px) -- (pplus);
  \draw[fx flow] (py.east) -- (6.3,-1.05) -- (pplus.south);
  \draw[fx flow] (py.south) -- (2,-2.5) -- (pcomp);
  \draw[fx flow] (px.south) -- (0,-3.8)
    -- node[fx label,below] {$X$} (9,-3.8) -- (pminus.south);
  \draw[fx flow] (pcomp) -- (pminus);
  \draw[fx flow] (pplus) -- (pup);
  \draw[fx flow] (pminus) -- (pum);
  \end{scope}

  \node[fx panel] at (-.4,1.65) {(d-2) One Harmonic Branch};
  \node[fx node] (mu) at (0,0) {$u$};
  \node[fx group,text width=3.4cm] (mv) at (3.5,0)
    {Rescaled Complement\\$2N_{2c+1}(u)-1$};
  \node[fx gate,text width=3cm] (mh) at (9,0)
    {Harmonic\\$H_{1/3}(u,V)$};
  \node[fx output] (mr) at (13.6,0) {$R$};
  \draw[fx flow] (mu) -- (mv);
  \draw[fx flow] (mv) -- node[fx label,above] {$V$} (mh);
  \draw[fx flow] (mu.north) -- (0,1.05)
    -- node[fx label,above] {$u$} (9,1.05) -- (mh.north);
  \draw[fx flow] (mh) -- (mr);

  \node[fx panel] at (-.4,-1.55) {(d-3) Combining Two Copies of the Branch};
  \node[fx node] (mp) at (0,-2.75) {$u_+$};
  \node[fx node] (mm) at (0,-4.85) {$u_-$};
  \node[fx group,text width=2cm] (mpb) at (2.6,-2.75) {Harmonic\\Branch (d-2)};
  \node[fx group,text width=2cm] (mmb) at (2.6,-4.85) {Harmonic\\Branch (d-2)};
  \node[fx group,text width=2.5cm] (mnc) at (6.4,-2.75) {Complement\\$N_D(R_+)$};
  \node[fx group,text width=3.6cm] (madd) at (10.8,-3.8)
    {Scaled Addition\\$a\bigl(R_-+N_D(R_+)\bigr)+1$};
  \node[fx output] (mw) at (13.6,-3.8) {$W$};
  \draw[fx flow] (mp) -- (mpb);
  \draw[fx flow] (mm) -- (mmb);
  \draw[fx flow] (mpb) -- node[fx label,above] {$R_+$} (mnc);
  \draw[fx flow] (mnc.east) -- (8.25,-2.75) -- (madd.160);
  \draw[fx flow] (mmb.east) -- node[fx label,below] {$R_-$} (8.25,-4.85) -- (madd.200);
  \draw[fx flow] (madd) -- (mw);
\end{tikzpicture}
\caption{The Multiplication Gadget}\label{fig:fixp-multiplication}
\end{figure}

\paragraph{Generating Coefficients from Their Bits.}
Arbitrary rational coefficients $p/q \in [0,1]$ are constructed by synthesizing bit representations using affine, complement, and multiplication gadgets alongside reciprocal gates. Let $b$ denote the binary bit length of $q$. Starting from the base constant $A_0 = 1$, process the bits $\varepsilon_0, \dots, \varepsilon_{b-1}$ of $p$ from least to most significant (padding with leading zeros as necessary). For each bit $\varepsilon_j$, append an affine gadget yielding output:
\[
 A_{j+1} = (A_j + 1 + \varepsilon_j)/2 \in [1,2].
\]
Let $A_p$ denote the final output variable of this chain. Replicate this construction independently for $q$ to produce output $A_q$; declare its final interval to be $[3/2, 2]$, as the most significant bit of $q$ is non-zero.

To invert $A_q - 1$, construct a four-stage reciprocal pipeline using positive intermediate variables: a complement gadget, an affine gadget, a second complement gadget, and a reciprocal gate, yielding outputs:
\[
 C = N_3(A_q) \in [1, 3/2], \quad E = 2C \in [2,3], \quad V = N_4(E) \in [1,2], \quad B = 2/V \in [1,2].
\]
This nine-gate cascade functions as a unified gadget with input $A_q$ and output $B$. Finally, multiply $A_p - 1$ (stored at $A_p$ with offset one) by $B$ (stored with offset zero) using a multiplication gadget. The output variable and offset of this multiplication gadget represent the coefficient $p/q$. For negative coefficients, attach an additional negation gadget.

For instance, the coefficient $3/5$ utilizes bit sequences $(1,1,0)$ and $(1,0,1)$. The resulting chains yield $A_p = 11/8$ and $A_q = 13/8$. The nine-gate inversion gadget computes $B = 8/5$, and the final multiplication step evaluates $(11/8 - 1)(8/5) = 3/5$.

\paragraph{Assembling the Coordinate Formulas.}
We construct an independent formula tree for each coordinate of $\mathcal F$. Source coordinates $x_j$ are represented by variables $z_j$ with offset one. For each constant, linear, or quadratic term, we construct its coefficient as described above, form the requisite source variable products, and multiply by the coefficient. The resulting stored terms are combined using addition gadgets, accumulating offsets at each stage. Intermediate gadget outputs are duplicated whenever multiple uses are required, ensuring that every intermediate gadget output serves as input to at most one downstream gadget (tree structure). Source variables and the constant 1 may be shared freely.

To prepare coordinate outputs for feedback, let $\widetilde\beta_j$ denote the accumulated offset for coordinate $j$. Attach an affine gadget that adds two and label its output $Y_j$. Set $\beta_j = \widetilde\beta_j + 2 \ge 2$, establishing the target relation $Y_j = \mathcal F_j(z - \ones) + \beta_j$. Identically zero coordinates use the constant one with offset one as their stored value before adding two. Additionally, construct an auxiliary complement gadget $D_j = N_4(z_j) \in [2,3]$ for each source variable $z_j$, forming the designated pair $(A_j, D_j)$ with $A_j = z_j$. The complement gadget shares its input $z_j$ with the coordinate trees, but its output $D_j$ is restricted exclusively to the designated pair.

\paragraph{Closing the Circuit with Feedback.}
For each coordinate $j$, the formula output $Y_j$ is connected back to its source variable $z_j$ through a feedback gate with parameter $\beta_j$. The gate equation $z_j = Y_j + 1 - \beta_j$ enforces the fixed-point system $z = \Phi(z)$, while the complement gadget sets $D_j = 4 - z_j$. \Cref{fig:fixp-feedback} depicts the coordinate evaluation tree, its feedback loop, and the auxiliary complement gadget.

\begin{figure}[!htb]
\centering
\begin{tikzpicture}
  \node[fx group,text width=3cm,minimum height=1.7cm] (fs) at (0,0)
    {Source Variables\\$z_1,\ldots,z_j,\ldots,z_d$};
  \node[fx group,text width=4cm] (ff) at (6.2,0)
    {Coordinate Formula\\$\mathcal F_j(z-\ones)+\beta_j$};
  \node[fx output] (fy) at (11.5,0) {$Y_j$};
  \draw[fx flow] (fs) -- node[fx label,above,align=center]
    {inputs used\\by $\mathcal F_j$} (ff);
  \draw[fx flow] (ff) -- (fy);
  \node[fx gate,text width=3.4cm] (fb) at (6.2,-2.2)
    {Feedback Gate\\$Y_j+1-\beta_j$};
  \draw[fx flow] (fy.south) -- (11.5,-2.2) -- (fb.east);
  \draw[fx flow] (fb.west) -- node[fx label,above] {$z_j$}
    (0,-2.2) -- (fs.south);
  \node[fx role,text width=3.6cm,anchor=north] at (2,-2.55)
    {Reuse source variable $z_j$\\as the output.};
  \node[fx group,text width=3.4cm] (fc) at (0,-4.45)
    {Complement\\$N_4(z_j)$};
  \node[fx output] (fd) at (6.2,-4.45) {$D_j$};
  \draw[fx flow] (fs.west) -- (-2.4,0)
    -- node[fx label,left] {$z_j$} (-2.4,-4.45) -- (fc.west);
  \draw[fx flow] (fc) -- (fd);
  \node[fx role,anchor=west] at (8,-4.45)
    {Designated pair:\\$(A_j,D_j)=(z_j,D_j)$};
\end{tikzpicture}
\caption{Connecting a Coordinate Formula to the Source Variables}\label{fig:fixp-feedback}
\end{figure}
\paragraph{Arithmetic Identities and Intervals.}
The gate equations give the required arithmetic identities and
keep their outputs within the declared intervals.
For multiplication, the two harmonic branches give a difference
of squares. Using the preceding gate equations, write
$u_\pm=c+v\pm w$. Each branch computes
\begin{equation}\label{eq:fixp-square}
 R_c(u):=H_{1/3}(u,4c+1-2u)
       =\frac{3c+1}{2}-\frac{3(u-c)^2}{2(c+1)}.
\end{equation}
Thus $a(R_c(u_-)-R_c(u_+))=vw$ for $a=(c+1)/6$, and the final
complement and addition gadgets give
\begin{equation}\label{eq:fixp-product}
 W=a\bigl(R_c(u_-)+N_D(R_c(u_+))\bigr)+1=vw+b_W.
\end{equation}
The other arithmetic identities and all interval checks are proved
in \Cref{app:fixp-arithmetic-estimates}. That analysis also shows
that each gadget has at most 26 gates, and that every variable input
from outside a gadget first enters a reciprocal gate whose only
variable input is that external input. Each coordinate formula has
polynomial size. Its stored bounds, offsets, and derivative products
are also polynomially bounded in the input length.

\FloatBarrier

\subsubsection{Implementing the Gates by a Reciprocal Graph}\label{sec:fixp-gates}\label{sec:fixp-closure}
We replace each circuit gate with a reciprocal subgraph gadget and identify nodes that share the same variable. Each circuit variable corresponds to a node associated with the identical interval. Additionally, we introduce a single shared reference node $r_0=1$ alongside a dedicated private node for each gate. In this \emph{intermediate reciprocal graph}, node $0$ is held fixed at value one with no box condition, whereas all other nodes must satisfy the box conditions in \Cref{eq:fixp-box}. In \Cref{sec:fixp-integer}, this fixed value is replaced with a box condition to produce a valid \ReciprocalVIProblem{} instance.

\paragraph{From Reciprocal Terms to Edges.}
A unified translation rule maps each gate specification into graph edges. Nodes are identified by their corresponding variables, with $0$ denoting the shared reference node. For constants $c, \alpha, \beta > 0$, a positive term $c/(\alpha r_i + \beta r_j)$ at node $i$ is implemented by a directed edge $i \to j$ with ratio $\beta / \alpha$ and multiplicity $c / \beta$. Conversely, a negative term $-c/(\alpha r_i + \beta r_j)$ at node $i$ is implemented by a directed edge $j \to i$ with ratio $\alpha / \beta$ and multiplicity $c / \alpha$. In accordance with \Cref{eq:fixp-edge-field}, the opposite endpoint receives a contribution of equal magnitude and opposite sign. Terms with zero coefficients contribute no edges.

Constant terms follow the same rule: a denominator of the form $r_i + a$ is rewritten as $r_i + a r_0$, obviating the need for auxiliary fixed nodes. At this stage, edge multiplicities are positive rational numbers; they will be scaled to integers once the complete circuit construction is assembled.

\paragraph{The Basic Gate.}
A basic gate connects input nodes $A_q$ to an output node $x$ via a private auxiliary node $y$ and the shared reference node $0$. The gate is specified by a target evaluation point $T > 0$, a sign parameter $\sigma \in \{1,-1\}$, and rational coefficients $w_q, d_q, e_q$ with $d_q, e_q > 0$. These parameters define the algebraic expressions
\[
 Q(y,A) = \sum_q \frac{w_q}{d_q y + e_q A_q}, \qquad x^\star(A) = \frac{1}{Q(T,A)} - T,
\]
where the target relation is $x = x^\star(A)$. The subgraph is constructed in three steps according to the edge parameters listed in \Cref{tab:fixp-gate-edges}. (The accompanying schematic illustrates a single input for $\sigma = 1$ and $w_q > 0$, where the dashed edge denotes the input-compensation edge.)

\begin{table}[!htbp]
\centering
\caption{Connections in a Basic Gate}\label{tab:fixp-gate-edges}
\begin{minipage}[c]{.40\linewidth}
\centering
\begin{tikzpicture}
  \node[fx role] at (2.3,1.2) {$\sigma=1,\quad w_q>0$};
  \node[fx output,label={[fx role]above:output}] (bx) at (0,0) {$x$};
  \node[fx private,label={[fx role]above:private}] (by) at (2.3,0) {$y$};
  \node[fx node,label={[fx role]above:input}] (ba) at (4.6,0) {$A_q$};
  \node[fx constant,label={[fx role]below:constant}] (b0) at (2.3,-2.2) {$r_0=1$};
  \draw[fx edge] (by) -- (bx);
  \draw[fx edge] (bx.south) -- (b0.west);
  \draw[fx edge] (ba) -- (by);
  \draw[fx edge,dashed] (b0.east) -- node[fx label,below right,font=\scriptsize]
    {compensation} (ba.south);
\end{tikzpicture}
\end{minipage}\hfill
\begin{minipage}[c]{.58\linewidth}
\centering
\small
\renewcommand{\arraystretch}{1.25}
\begin{tabular}{@{}lccc@{}}
\toprule
Condition & Directed edge & Ratio & Multiplicity \\
\midrule
\multirow{2}{*}{$\sigma=1$}
 & $y\to x$ & $1$ & $1$ \\
 & $x\to0$ & $T$ & $1/T$ \\
\addlinespace
\multirow{2}{*}{$\sigma=-1$}
 & $0\to x$ & $1/T$ & $1$ \\
 & $x\to y$ & $1$ & $1$ \\
\midrule
\multirow{2}{*}{$\sigma w_q>0$}
 & $0\to A_q$ & $e_q/(d_qT)$ & $|w_q|/e_q$ \\
 & $A_q\to y$ & $d_q/e_q$ & $|w_q|/d_q$ \\
\addlinespace
\multirow{2}{*}{$\sigma w_q<0$}
 & $y\to A_q$ & $e_q/d_q$ & $|w_q|/e_q$ \\
 & $A_q\to0$ & $d_qT/e_q$ & $|w_q|/(d_qT)$ \\
\bottomrule
\end{tabular}
\end{minipage}
\end{table}

\begin{enumerate}[label=(\arabic*),leftmargin=*,itemsep=0.3em]
\item Introduce an output node $x \in [\ell,u]$ and a private auxiliary node $y \in [T-\Delta, T+\Delta]$, where $\Delta = T/1000$. The node $y$ is exclusive to this gate, while $x$ may serve as an input to subsequent gates. (A feedback gate utilizes an existing source node as its output instead of creating a new one.)
\item Link the output node to the constant and private nodes: for $\sigma=1$, insert the path $y \to x \to 0$; for $\sigma=-1$, insert the path $0 \to x \to y$, following the first two cases of \Cref{tab:fixp-gate-edges}.
\item For each non-zero term in $Q$, instantiate two edges connected to the input node $A_q$: if $\sigma w_q > 0$, insert the path $0 \to A_q \to y$; if $\sigma w_q < 0$, insert the path $y \to A_q \to 0$. Apply the parameters from the remaining cases of \Cref{tab:fixp-gate-edges}, repeating this step for every term while sharing the nodes $x, y, 0$.
\end{enumerate}
Consequently, each variable input contributes a two-edge path connecting $y$ and $0$, in parallel with the path routed through $x$.

For a constant input $A_q = a$, the private node $y$ is connected directly to node $0$. In the edge entry between $y$ and $A_q$ in \Cref{tab:fixp-gate-edges}, $A_q$ is replaced by $0$ and $e_q$ by $e_q a$, while the input-compensation edge is omitted. This realizes the denominator $d_q y + e_q a r_0$ without instantiating a dedicated node for $a$.

For subsequent balance calculations, define $K_x = \frac{1}{x+y}$ and $B_x = \frac{1}{x+T}$. The edge contributions induced by \Cref{tab:fixp-gate-edges} evaluate to
\begin{equation}\label{eq:fixp-gate-fields}
\begin{aligned}
 h &= \sigma(B_x - K_x) = \frac{\sigma(y - T)}{(x + T)(x + y)}, &
 \mathcal G_y &= \sigma(K_x - Q), \\
 P_q &= \sigma w_q \left( \frac{1}{d_q y + e_q A_q} - \frac{1}{d_q T + e_q A_q} \right).
\end{aligned}
\end{equation}
Here, $h$ denotes the contribution at output node $x$, $\mathcal G_y$ represents the net field contribution at private node $y$, and $P_q$ signifies the perturbation contribution at input node $A_q$.

\FloatBarrier
\paragraph{The Harmonic Gate.}
The harmonic gate $H_\rho(A,B)$ operates on input nodes $A \in [a_1,b_1]$ and $B \in [a_2,b_2]$. Its reciprocal subgraph is constructed as follows:
\begin{enumerate}[label=(\arabic*),leftmargin=*,itemsep=0.25em]
\item Instantiate an output node $x \in [H_\rho(a_1,a_2), H_\rho(b_1,b_2)]$ and a private auxiliary node $y \in [0.999, 1.001]$. A larger positive output interval may be declared provided it contains this range.
\item Insert unit-multiplicity edges $0 \to x$ and $x \to y$. Insert the path $y \to A \to 0$ with edge multiplicity $\rho$, and the path $y \to B \to 0$ with edge multiplicity $1-\rho$. Every edge in this gadget has ratio one.
\end{enumerate}
Under the basic-gate framework, this configuration corresponds to setting $T=1$, $\sigma=-1$, and $Q(y,A,B) = \frac{\rho}{y+A} + \frac{1-\rho}{y+B}$.

\begin{figure}[!htb]
\noindent
\Cref{fig:fixp-harmonic,fig:fixp-reciprocal} show the complete
harmonic and reciprocal gates. Each edge label gives its ratio
and multiplicity, in that order. Shaded circles
are private variables, and the rectangle is the shared constant
node. Dashed edges compensate for contributions at the inputs;
they are ordinary edges of the constructed graph.
\par\medskip
\centering
\begin{tikzpicture}
  \node[fx output,label={[fx role]above:output}] (hx) at (0,0) {$x$};
  \node[fx private,label={[fx role]above:private}] (hy) at (3.4,0) {$y$};
  \node[fx node,label={[fx role]above:input}] (ha) at (6.9,1.1) {$A$};
  \node[fx node,label={[fx role]below:input}] (hb) at (6.9,-1.1) {$B$};
  \node[fx constant,label={[fx role]above:constant}] (h0) at (10.4,0) {$r_0=1$};
  \draw[fx edge] (hx) -- node[fx label,above] {$(1,1)$} (hy);
  \draw[fx edge] (hy) -- node[fx label,above=4pt] {$(1,\rho)$} (ha);
  \draw[fx edge] (hy) -- node[fx label,below=9pt] {$(1,1-\rho)$} (hb);
  \draw[fx edge,dashed] (ha) -- node[fx label,above=4pt] {$(1,\rho)$} (h0);
  \draw[fx edge,dashed] (hb) -- node[fx label,below=9pt] {$(1,1-\rho)$} (h0);
  \draw[fx edge] (h0.south) -- (10.4,-2.25)
    -- node[fx label,below] {$(1,1)$} (0,-2.25) -- (hx.south);
\end{tikzpicture}
\caption{The Harmonic Gate}\label{fig:fixp-harmonic}
\end{figure}

\FloatBarrier

\paragraph{The Reciprocal Gate.}
The reciprocal gate $J(A) = \frac{\kappa}{A+c} + \beta$ takes input node $A \in [a,b]$ and evaluates onto declared output interval $[\ell,u]$. Its reciprocal subgraph construction proceeds as follows:
\begin{enumerate}[label=(\arabic*),leftmargin=*,itemsep=0.25em]
\item Select an integer parameter $T \ge \max\{1, 20(u+|\beta|)\}$, and define auxiliary constants $R_T = T + \beta$ and $D_T = \kappa + R_T c$.
\item Instantiate an output node $x \in [\ell,u]$ and a private auxiliary node $y \in [0.999T, 1.001T]$.
\item Insert the three-edge cycle $y \to x \to 0 \to y$ alongside the two-edge path $y \to A \to 0$. Assign ratios and multiplicities according to \Cref{fig:fixp-reciprocal}.
\end{enumerate}
In the basic-gate framework, this corresponds to setting $\sigma=1$ and defining
\begin{equation}\label{eq:fixp-reciprocal-Q}
 Q(y,A) = \frac{(101T)/R_T}{y+100T} - \frac{\kappa/R_T}{(D_T/T)y + R_T A}.
\end{equation}
The first term is anchored to the shared reference node $0$, while the second involves the input node $A$. If $A$ represents a constant scalar, the constant-input reduction following \Cref{tab:fixp-gate-edges} applies.

To satisfy an input perturbation constraint of the form $|P_A| \le \varepsilon |h|$ for a prescribed tolerance $\varepsilon > 0$—bounding the input-side contribution relative to the output field—the center parameter $T$ is chosen as:
\begin{equation}\label{eq:fixp-reciprocal-center}
 T = \left\lceil \max\left\{ 1, \, 20(u+|\beta|), \, \frac{128\kappa(\kappa+c)}{(a+c)^2 \varepsilon} \right\} \right\rceil.
\end{equation}
Otherwise, $T$ is set to the minimal integer satisfying the lower bound in step~(1). Fixing $T$ uniquely determines all five edge parameters.

\begin{figure}[!htb]
\centering
\begin{tikzpicture}
  \node[fx output,label={[fx role]above:output}] (rx) at (0,0) {$x$};
  \node[fx private,label={[fx role]above:private}] (ry) at (5,0) {$y$};
  \node[fx node,label={[fx role]above:input}] (ra) at (10,0) {$A$};
  \node[fx constant,label={[fx role]below:constant}] (r0) at (5,-2.5) {$r_0=1$};
  \draw[fx edge] (ry) -- node[fx label,above] {$(1,1)$} (rx);
  \draw[fx edge] (rx.south) -- node[fx label,below left] {$(T,1/T)$} (r0.west);
  \draw[fx edge] (r0.north) -- node[fx label,right]
    {$\left(\frac1{100T},\frac{101T}{R_T}\right)$} (ry.south);
  \draw[fx edge] (ry) -- node[fx label,above]
    {$\left(\frac{TR_T}{D_T},\frac{\kappa}{R_T^2}\right)$} (ra);
  \draw[fx edge,dashed] (ra.south) -- node[fx label,below right]
    {$\left(\frac{D_T}{R_T},\frac{\kappa}{R_TD_T}\right)$} (r0.east);
\end{tikzpicture}
\caption{The Reciprocal Gate}\label{fig:fixp-reciprocal}
\end{figure}
\FloatBarrier
\paragraph{The Feedback Gate.}
To enforce the circuit relation $z_j = Y_j + 1 - \beta_j$, we utilize the existing nodes $Y_j, z_j, 0$ and construct the feedback reciprocal subgraph in the following sequence:
\begin{enumerate}[label=(\arabic*),leftmargin=*,itemsep=0.3em]
\item Set parameter $T_j = 2\beta_j$ and coefficient $d_j = \frac{\beta_j+1}{2\beta_j}$, and instantiate a private auxiliary node $y_j \in [0.999T_j, 1.001T_j]$. The nodes $Y_j, z_j, 0$ are reused from prior constructions; only $y_j$ is newly introduced.
\item Insert the four-edge directed cycle $0 \to Y_j \to y_j \to z_j \to 0$ using the ratios and multiplicities specified in \Cref{tab:fixp-feedback-edges}. The edge $0 \to Y_j$ serves as the input-compensation edge.
\item Clear rational multiplicities to obtain integers: let $C_j^{\rm fb}$ denote the product of their denominators in lowest terms. Scale all four multiplicities by $C_j^{\rm fb}$, then replace each resulting integer multiplicity with that number of parallel unit edges.
\end{enumerate}
The cycle in step~(2) instantiated within the basic-gate framework corresponds to $x = z_j$, $A = Y_j$, $\sigma = 1$, and parameter triplet $(w,d,e) = (1, d_j, 1)$. In particular, this yields $Q(y_j, Y_j) = \frac{1}{Y_j + d_j y_j}$. We denote by $h_j^{\rm fb}$ and $P_j^{\rm fb}$ the resulting field contributions at output node $z_j$ and input node $Y_j$, respectively.

\begin{table}[!htbp]
\centering
\caption{Connections of a Feedback Gate}\label{tab:fixp-feedback-edges}
\begin{minipage}[c]{.45\linewidth}
\centering
\begin{tikzpicture}
  \node[fx output,label={[fx role]above:output}] (fz) at (0,0) {$z_j$};
  \node[fx private,label={[fx role]above:private}] (fv) at (2.3,0) {$y_j$};
  \node[fx node,label={[fx role]above:input}] (fi) at (4.6,0) {$Y_j$};
  \node[fx constant,label={[fx role]below:constant}] (f0) at (2.3,-2.2) {$r_0=1$};
  \draw[fx edge] (fv) -- (fz);
  \draw[fx edge] (fz.south) -- (f0.west);
  \draw[fx edge] (fi) -- (fv);
  \draw[fx edge,dashed] (f0.east) -- node[fx label,below right,font=\scriptsize]
    {compensation} (fi.south);
\end{tikzpicture}
\end{minipage}\hfill
\begin{minipage}[c]{.53\linewidth}
\centering
\renewcommand{\arraystretch}{1.2}
\begin{tabular}{@{}lcc@{}}
\toprule
Directed edge & Ratio & Multiplicity \\
\midrule
$0\to Y_j$ & $1/(d_jT_j)$ & $1$ \\
$Y_j\to y_j$ & $d_j$ & $1/d_j$ \\
$y_j\to z_j$ & $1$ & $1$ \\
$z_j\to0$ & $T_j$ & $1/T_j$ \\
\bottomrule
\end{tabular}
\end{minipage}
\end{table}

\FloatBarrier

We assemble the complete network by taking the union of all gate subgraphs, merging nodes that represent the same circuit variable while preserving all edges and their respective multiplicities. Each private auxiliary node is dedicated exclusively to a single gate. The designated node pairs are defined as $(A_j, D_j) = (z_j, N_4(z_j))$, and node $0$ remains fixed at value one. The parameter choices specified below convert all remaining rational multiplicities into parallel unit edges.
\paragraph{Why Error Analysis Is Needed.}
Interconnecting the reciprocal subgraphs introduces cross-gate interactions at shared nodes. Specifically, when a downstream gate uses an upstream gate's output as an input, it injects a perturbation contribution back into that node. Although the upstream gate's private variable can shift to absorb this extra contribution, such shifts risk perturbing the output variable from its intended arithmetic value. The input-compensation edges are designed to perfectly nullify each gate's input perturbations precisely when its private variable is held at its target center. Consequently, we must establish that the network-wide box conditions force all private variables to concentrate at their respective target centers, even in the presence of full subgraph interconnections.

\paragraph{From Private-Node Balance to Output Error.}
We first relate each gate's arithmetic error to its contribution at the output node. The private-node box condition gives the following relation for all three gate types.

\begin{lemma}\label{lem:fixp-gate}
For every harmonic, reciprocal, or feedback gate, the private-node
box condition implies
\begin{equation}\label{eq:fixp-gate-sign}
 \operatorname{sgn}h=-\operatorname{sgn}(x-x^\star(A)),\qquad
 h=0\ \Longleftrightarrow\ y=T\ \Longleftrightarrow\ x=x^\star(A).
\end{equation}
\end{lemma}
The proof is given in \Cref{app:fixp-private}. Thus, an output above its intended value receives a negative contribution from its own gate, while an output below its intended value receives a positive contribution. When $x^\star(A)$ lies in the declared output interval, the output-node box condition also gives $|h|\le |J^{\rm ext}|$, where $J^{\rm ext}$ is the total contribution at that node from other gates.

\paragraph{Controlling Input Contributions.}
We next control the perturbation that a reciprocal gate introduces at its input node. Increasing the center parameter $T$ makes this contribution arbitrarily small relative to the output contribution, without changing the intended reciprocal function.

\begin{lemma}\label{lem:fixp-reciprocal-input}
For the center chosen in \Cref{eq:fixp-reciprocal-center}, the
reciprocal gate satisfies $|P_A|\le\varepsilon|h|$ throughout its box.
\end{lemma}
The proof is given in \Cref{app:fixp-gate-estimates}.

\paragraph{Controlling Each Arithmetic Gadget.}
We combine the gate bounds to control each arithmetic gadget's output error and its contributions at external inputs. For a gadget $g$, let $x_g$ denote its output variable, $A_g$ its vector of inputs, and $\varphi_g$ its prescribed arithmetic function. Its output error is $e_g=|x_g-\varphi_g(A_g)|$.

The parameters are chosen in two steps. We first choose the centers of gates that read external inputs to reduce their perturbation contributions. We then multiply all edge multiplicities in the gadget by a common integer $M_g$, making its arithmetic error small relative to its output contribution. Let $h_g$ be the resulting contribution of the final gate at the gadget's output node, and let $J_g^{\rm ext}$ be the total contribution at that node from outside the gadget. For a rational tolerance $0<\eta<1$, the required estimates are
\begin{equation}\label{eq:fixp-group-bounds}
 e_g\le\eta|h_g|,\qquad
 \sum_{\text{external inputs of }g}|P|\le\eta|h_g|,
 \qquad |h_g|\le|J_g^{\rm ext}|.
\end{equation}
The middle sum counts each use of an external input separately. The following lemma gives all three estimates while keeping the construction polynomial in size.

\begin{lemma}\label{lem:fixp-group}
For every rational $0<\eta<1$, the specified parameter choices
make each arithmetic gadget satisfy \Cref{eq:fixp-group-bounds}
at every solution of the box conditions. They preserve its intended function and are
polynomially bounded in the fixed gadget data and $1/\eta$.
\end{lemma}
The parameter choices are specified in \Cref{app:fixp-parameters}, and the proof is given in \Cref{app:fixp-group-estimates}.

\paragraph{How the Error Analysis Forces Exactness.}\label{par:fixp-parameters}
We now combine the bounds in \Cref{lem:fixp-gate,lem:fixp-group} to show that non-zero arithmetic errors cannot balance across the network. A polynomially bounded constant $K \ge 1$ controls error propagation across the formula trees and feedback paths. Setting $\eta = \frac{1}{2(1+K+K^3)}$ and aggregating the node balance conditions yields the inequality $\mathcal R \le \eta(1+K+K^3)\mathcal R = \frac{1}{2}\mathcal R$, where $\mathcal R = \sum_g |h_g|$. This forces $\mathcal R = 0$, guaranteeing that all gate evaluation errors vanish identically. \Cref{app:fixp-parameters} provides the explicit rational specifications for $K$, the center parameters, and the scaling integers $M_g$, as well as the complete expansion of rational multiplicities into parallel unit edges.

\begin{lemma}\label{lem:fixp-exact}
Any solution satisfying the box conditions of the intermediate reciprocal graph evaluates all arithmetic and feedback gates exactly. In particular, it satisfies $z = \Phi(z)$ and $D_j = 4 - z_j$ for every source coordinate $j$.
\end{lemma}
The complete proof is provided in \Cref{app:fixp-exactness}, drawing upon the detailed gate, gadget, and feedback estimates established in \Cref{app:fixp-network}.
\subsubsection{Removing the Fixed Node and Scaling the Intervals}\label{sec:fixp-integer}
We convert the intermediate reciprocal graph into a standard \ReciprocalVIProblem{} instance in which every node satisfies a box condition. Currently, reference node $0$ lacks such a condition. We assign it the interval $[1,2]$ and introduce incoming edges that constrain its value to one. Finally, we attach designated edges to each designated pair and scale all node intervals so that $a_i \ge 1$ and $\lambda_e \le a_i$ hold for every edge $e:i \to j$.

\paragraph{1. Make the Constant Node Variable.}
Starting with the intermediate reciprocal graph (retaining all existing parallel edges), assign the former constant node $r_0$ the interval $[1,2]$. All other intervals, edge endpoints, and edge ratios remain unchanged.

\paragraph{2. Add the Designated Edges.}
For each designated pair $(A_j, D_j)$, insert two unit-multiplicity edges $A_j \to D_j$ and $D_j \to A_j$, both carrying ratio one. The newly added forward edge $A_j \to D_j$ is specified as the designated edge in the output instance. These unit edges are added unconditionally, even if identical edges already exist.

\paragraph{3. Add the Edges Entering Node Zero.}
We next introduce an auxiliary node whose incident edges all enter node $0$. Let $W_0$ denote the total current edge count and $\ell_* > 0$ denote the smallest lower interval endpoint across all nodes (including node $0$). Introduce an auxiliary node $r_* \in [3,4]$ and instantiate
\[
  K_0 = \left\lfloor \frac{6W_0}{\ell_*} \right\rfloor + 1
\]
parallel ratio-one edges $r_* \to r_0$. These constitute the sole edges incident to $r_*$.

\Cref{fig:fixp-constant} illustrates the added edges prior to interval scaling.
\begin{figure}[!htb]
\centering
\begin{tikzpicture}
  \node[fx panel] at (-.5,1.7) {(a) New Node and Edges Entering Node $0$};
  \node[fx node,label={[fx role]below:$[3,4]$}] (cs) at (0,0) {$r_*$};
  \node[fx node,label={[fx role]below:$[1,2]$}] (cz) at (5.2,0) {$r_0$};
  \draw[fx edge] (cs) -- node[fx label,above,align=center]
    {$K_0$ parallel edges\\ratio $1$} (cz);
  \node[fx role,align=left,anchor=west] at (6.4,0)
    {Old incident edges stay in place.\\The new edges enter node $0$.};

  \node[fx panel] at (-.5,-1.7) {(b) Two New Edges for Each Designated Pair};
  \node[fx node,label={[fx role]below:$[1,2]$}] (ca) at (0,-3) {$A_j$};
  \node[fx node,label={[fx role]below:$[2,3]$}] (cd) at (5.2,-3) {$D_j$};
  \draw[fx edge] (ca) to[bend left=22] node[fx label,above] {ratio $1$: designated} (cd);
  \draw[fx edge] (cd) to[bend left=22] node[fx label,below] {ratio $1$} (ca);
  \node[fx role,align=left,anchor=west] at (6.4,-3)
    {One new unit edge in each direction.\\The old designated pair stays in place.};
\end{tikzpicture}
\caption{Edges Added Before Scaling}\label{fig:fixp-constant}
\end{figure}

\paragraph{4. Scale the Intervals.}
We uniformly scale all interval endpoints by a common integer multiplier $M$. Denoting the unscaled variables and intervals by $\widetilde r_i \in [\widetilde a_i, \widetilde b_i]$, we set
\[
  M = \left\lceil \max \left\{ 1, \, \max_i \frac{1}{\widetilde a_i}, \, \max_{e:i\to j} \frac{\lambda_e}{\widetilde a_i} \right\} \right\rceil.
\]
We define the scaled quantities $r_i = M\widetilde r_i$, $a_i = M\widetilde a_i$, and $b_i = M\widetilde b_i$ while preserving all edge ratios $\lambda_e$. The resulting graph, together with the scaling factor $M$ and the designated pairs (along with their newly inserted forward edges), forms the complete input data for the target \ReciprocalVIProblem{} instance.

\paragraph{Proof of Correctness}We now prove that the constructed target instance has polynomial size and that any solution yields a fixed point of $\mathcal F$ via a separable linear (SL) recovery map.

\begin{proof}[Proof of \Cref{thm:fixp-vi}]
We first recover a valid solution of the intermediate reciprocal graph from any solution $r$ to the target \ReciprocalVIProblem{} instance. Define the unscaled vector $\widetilde r = r/M$. Scaling all node values and interval endpoints by $1/M$ multiplies each edge field contribution by $M$; hence, $\widetilde r$ satisfies the unscaled box conditions in \Cref{eq:fixp-box}.

The newly introduced edges force node $0$ to take its required reference value of one. Prior to inserting the auxiliary node $r_*$, each of the $W_0$ existing edges contributes a field magnitude of at most $1/\ell_*$ at node $0$, since $\frac{\lambda_e}{\widetilde r_i + \lambda_e \widetilde r_j} \le \frac{1}{\widetilde r_j} \le \frac{1}{\ell_*}$. Because $\widetilde r_* + \widetilde r_0 \le 4 + 2 = 6$, the $K_0$ newly added edges directed into node $0$ contribute a negative field component of magnitude at least $\frac{K_0}{6} > \frac{W_0}{\ell_*}$. Thus, $\mathcal G_0(\widetilde r) < 0$ throughout the box, so the box condition forces $\widetilde r_0 = 1$ at the lower endpoint of its interval $[1,2]$. Similarly, the net positive field generated at node $r_*$ forces its value to the upper bound $\widetilde r_* = 4$. Furthermore, because the two newly added unit edges across each designated pair $(A_j,D_j)$ carry opposite directions and ratio one, their contributions cancel identically at both endpoints. Removing these designated edges and the auxiliary node $r_*$ leaves a valid solution to the intermediate reciprocal graph with node $0$ held fixed at one.

\Cref{lem:fixp-gate} relates each gate's output contribution to its arithmetic error, while \Cref{lem:fixp-group} bounds each gadget's error and external input contributions. Applying \Cref{lem:fixp-exact}, every gate in the intermediate graph evaluates exactly. Setting $z_j = \widetilde r_{A_j}$, the coordinate update and feedback equations enforce $z = \Phi(z)$, while each auxiliary complement gadget guarantees $\widetilde r_{D_j} = 4 - z_j$.

The final graph conforms to all structural requirements of \Cref{def:fixp-vi}. The choice of scaling factor $M$ ensures $a_i \ge 1$ and $\lambda_e \le a_i$ for every edge $e:i\to j$. The designated intervals evaluate to $[M, 2M]$ and $[2M, 3M]$, respectively, and the relation $r_{A_j} + r_{D_j} = 4M$ holds at every solution. Each designated forward edge $A_j \to D_j$ is a single unit edge with ratio one.

We now verify polynomial complexity. The bounds established throughout the construction ensure that the target instance size is polynomially bounded in the bit length of the input quadratic map. Each coordinate formula tree has polynomial size and contains polynomially bounded intervals, offsets, and evaluation ranges. The reciprocal gate parameters remain polynomially bounded because their positive lower bounds and interval gaps have polynomially bounded reciprocals. The feedback analysis yields a polynomial upper bound on $K$, and consequently on $1/\eta$. Because each arithmetic gadget comprises a constant number of gates, \Cref{lem:fixp-group} ensures that its center parameters, local denominator products, and scaling multipliers $M_g$ are polynomially bounded. Expanding the integer multiplicities produces polynomially many unit edges, and all rational edge parameters have polynomially bounded numerators and denominators.

The final graph transformations preserve these complexity bounds. In particular, the parameters $W_0$, $1/\ell_*$, all original interval endpoints, and all edge ratios are polynomially bounded, ensuring that $K_0$ and $M$ remain polynomial. Consequently, adding the $K_0$ auxiliary edges and multiplying interval endpoints by $M$ yields a graph description whose length is polynomial under unary encoding. Furthermore, every edge denominator remains strictly positive across its declared domain box.

Finally, any solution vector $z \in [1,2]^d$ extends to a valid solution $r$ with $r_{A_j} = M z_j$. As established above, $z = \Phi(z)$ implies that $x_j = z_j - 1$ is an exact fixed point of $\mathcal F$. This coordinate-wise affine transformation constitutes an SL reduction from \QFPProblem{} to \ReciprocalVIProblem. Combined with the FIXP-hardness of \QFPProblem{} (\Cref{thm:fixp-quadratic}), this completes the proof of FIXP-hardness for \ReciprocalVIProblem.
\end{proof}
\subsection{\texorpdfstring{From \ReciprocalVIProblem{} to \HZEquilibriumProblem}{From Reciprocal VI to HZ Equilibrium}}\label{sec:fixp-market}
We reduce \ReciprocalVIProblem{} to \HZEquilibriumProblem{} by constructing market agents and goods whose equilibrium price structures and market-clearing conditions enforce the graph's box conditions.

\begin{restatedtheorem}{thm:fixp-market}
\HZEquilibriumHardnessStatement
\end{restatedtheorem}

We first present the market construction and subsequently prove its correctness.

\subsubsection{Construction}
We construct an HZ market where node agents dictate prices and edge agents translate these prices into reciprocal edge contributions. The input graph comprises $k$ nodes and $W$ unit edges.

Additionally, each designated coordinate is recovered directly from a single allocation entry. For each designated edge, the corresponding edge agent derives positive utility from a single reserved high good; the agent's equilibrium share of this good yields the target coordinate via an affine map. Below, we specify the good classifications, agent population sizes, supply quantities, and utility profiles.

\paragraph{1. Create the Good Groups.}
For each node $i$, we define low-good groups corresponding to its outgoing edges and a high-good group corresponding to its incoming edges according to the following classification:
\begin{itemize}[leftmargin=*,itemsep=0.15em]
\item For every outgoing unit edge $e:i\to j$, instantiate a distinct low-good group at node $i$ with ratio parameter $\lambda_e$. Parallel edges are assigned distinct low-good groups.
\item Add an auxiliary low-good group at node $i$ with parameter $a_i$. No edge agents are associated with this auxiliary group.
\item Instantiate a single high-good group at node $i$, shared by all incoming edges (subject to the high-good reservation rule specified below). This group is created even if no edges enter node $i$.
\end{itemize}
We index the low-good groups at node $i$ by $t$ and denote their associated parameters by $\lambda_{it}$. Finally, we create a global group of reference goods shared across all node classes.

\paragraph{2. Set the Numbers of Goods and Agents.}
We next determine the sizes of all good groups and agent classes. Express each parameter as $\lambda_{it} = P_{it}/Q_{it}$ in lowest terms with positive integers $P_{it}, Q_{it}$, and set
\begin{equation}\label{eq:fixp-supplies}
 \ell_0 = 2W + 2, \qquad L_{it} = \ell_0 Q_{it}, \qquad H_i = \ell_0 \sum_t P_{it}.
\end{equation}
We instantiate $L_{it}$ distinct goods in low group $it$ and $H_i$ distinct goods in high group $i$, each with unit supply. Goods within a given low group possess identical utility profiles; goods within a high group are also identical, with the sole exception of the reserved high good defined below.

To fix the overall populations, let $S_i = H_i + \sum_t L_{it}$ and $S_* = \max_i S_i$. We set
\[
 B_0 = 3S_*, \qquad N = kB_0 + W, \qquad S_R = N - \sum_i S_i.
\]
We introduce $B_0$ node agents for each node $i$, one edge agent for each unit edge, and $S_R$ reference goods. Each agent is endowed with a unit budget and demands a total commodity quantity of one, yielding a balanced economy with $N$ agents and $N$ goods.

\paragraph{3. Fill the Node-Agent Rows.}
Every node agent associated with node $i$ is assigned the utility profile defined in \Cref{tab:fixp-utilities}, where the listed utility applies uniformly to every individual commodity within the specified group.

\begin{table}[htbp]
\centering
\caption{Utilities of a Node Agent}\label{tab:fixp-utilities}
\begin{tabular}{@{}lc@{}}
\toprule
Goods & Utility of a node-$i$ agent \\
\midrule
Own low group $t$ & $(b_i-\lambda_{it})/(b_i+1)$ \\
Reference goods & $b_i/(b_i+1)$ \\
Own high goods & $1$ \\
Goods at other nodes & $0$ \\
\bottomrule
\end{tabular}
\end{table}

\paragraph{4. Fill the Edge-Agent Rows.}
For each unit edge $e:i\to j$, we create a dedicated edge agent $w_e$. Agent $w_e$ derives utility $1 - \varepsilon_w$ for every good in the low group corresponding to edge $e$, where $\varepsilon_w = [2(1+b_*)]^{-1}$ and $b_* = \max_i b_i$. The agent derives unit utility ($1$) for its designated subset of high goods at node $j$ (specified below) and zero utility for all other commodities, including reference goods.

To partition the high goods, we reserve one specific high good $h_j^*$ at node $D_j$ for each designated edge $A_j \to D_j$. Since designated pairs are disjoint, each node contains at most one reserved high good. High-good utility allocations are then assigned as follows:
\begin{itemize}[leftmargin=*,itemsep=0.2em]
\item The agent $w_j$ for the designated edge $A_j \to D_j$ derives utility one exclusively from the reserved good $h_j^*$ among all high goods at $D_j$.
\item All other edge agents entering node $D_j$ derive utility one from all high goods at $D_j$ except the reserved good $h_j^*$.
\item At any node without a reserved high good, every incoming edge agent derives utility one from the entire high-good group.
\end{itemize}
All unlisted utility entries for edge agents are set to zero. Note that the node-agent utility profiles from Step 3 assign unit utility to all own high goods, including any reserved high goods.

\Cref{fig:fixp-market} summarizes the constructed utility matrix. Edge labels denote utility values, and panel~(c) illustrates the selective valuation of reserved high goods.
\begin{figure}[!htb]
\centering
\begin{tikzpicture}
  \node[fx panel] at (-.6,1.2) {(a) Node Agents and Their Goods};
  \node[fx gate,text width=2.8cm] (ml) at (0,0) {Own Low Group $t$\\$L_{it}$ Goods};
  \node[fx gate,text width=2.8cm] (mref) at (5.3,0) {Reference Group\\$S_R$ Goods};
  \node[fx gate,text width=2.8cm] (mhi) at (10.6,0) {Own High Group\\$H_i$ Goods};
  \node[fx group,text width=3.4cm] (mn) at (5.3,-2.4) {$B_0$ Node-$i$ Agents};
  \draw[fx flow] (mn.west) -- node[fx label,midway,below left=2pt] {$\dfrac{b_i-\lambda_{it}}{b_i+1}$} (ml.south);
  \draw[fx flow] (mn.north) -- node[fx label,right] {$\dfrac{b_i}{b_i+1}$} (mref.south);
  \draw[fx flow] (mn.east) -- node[fx label,above left=2pt] {$1$} (mhi.south);

  \node[fx panel] at (-.6,-3.45) {(b) The Agent for Edge $e:i\to j$};
  \node[fx gate,text width=3.1cm] (mel) at (0,-4.7) {Low Group for $e$\\at Node $i$};
  \node[fx group,text width=2.5cm] (mea) at (5.3,-4.7) {Edge Agent $w_e$};
  \node[fx gate,text width=3.1cm] (meh) at (10.6,-4.7) {Selected High Goods\\at Node $j$};
  \draw[fx flow] (mea) -- node[fx label,above] {$1-\varepsilon_w$} (mel);
  \draw[fx flow] (mea) -- node[fx label,above] {$1$} (meh);
  \node[fx role,below=3mm of mea] {Zero utility for\\all other goods};
\end{tikzpicture}
\par\smallskip
{\small
\textbf{(c) High-Good Utilities at a Designated Node $D_j$}\par\smallskip
\begin{tabular}{@{}lcc@{}}
\toprule
Agent & Reserved high good $h_j^*$ & Other high goods at $D_j$ \\
\midrule
Designated edge agent $w_j$ & $1$ & $0$ \\
Other edge agents entering $D_j$ & $0$ & $1$ \\
Node-$D_j$ agents & $1$ & $1$ \\
\bottomrule
\end{tabular}}
\caption{Agents and Goods for a Node and an Edge}\label{fig:fixp-market}
\end{figure}

The resulting utility matrix defines the market instance. Structuring node agents, edge agents, low goods, high goods, and reference goods in their construction order establishes fixed row and column indices, uniquely determining each pair $(w_j, h_j^*)$ used for coordinate recovery.

\subsubsection{Proof of Correctness}
We demonstrate that every market equilibrium allocation recovers a solution to the input graph, allowing designated coordinates to be extracted directly from individual allocation values. Let $\pi_R$ denote the reference price, $\pi_i^+$ the high-good price at node $i$, and $\pi_{it}^-$ the price of low-good group $it$. The following lemma establishes the supporting price structures, demand allocations, and box-condition properties essential for coordinate recovery.

\begin{lemma}\label{lem:fixp-market-correctness}
Every equilibrium allocation admits a system of supporting prices of the form
\begin{equation}\label{eq:fixp-price-form}
 \pi_R = 1, \qquad \pi_i^+ = 1 + \frac{1}{r_i}, \qquad
 \pi_{it}^- = 1 - \frac{\lambda_{it}}{r_i}, \qquad r_i \in [a_i, b_i].
\end{equation}
The vector $r$ satisfies the graph box conditions in \Cref{eq:fixp-box}. At these supporting prices, every optimal utility-maximizing bundle exhausts the agent's unit budget. Furthermore, any optimal consumption bundle for the edge agent corresponding to $e:i\to j$ concentrates exclusively on its selected low- and high-good groups, with total quantities given by
\begin{equation}\label{eq:fixp-edge-demand}
 q_e^- = \frac{r_i}{r_i + \lambda_e r_j}, \qquad
 q_e^+ = \frac{\lambda_e r_j}{r_i + \lambda_e r_j}.
\end{equation}
\end{lemma}
The detailed proof of \Cref{lem:fixp-market-correctness} is provided in \Cref{app:fixp-market-proof}.
\paragraph{Intuition.}\label{sec:fixp-recovery}
The node utilities make price differences encode a graph variable.
The price argument first establishes full spending and normalizes
the minimum price to zero. The supply bounds then force each node
class to consume its own high and low groups and the reference goods.
Let $\gamma_i>0$ and $\mu_i\ge0$ be the slope and intercept of
its optimal demand dual. Complementary slackness gives
\[
 \gamma_i(\pi_i^+-\pi_R)=\frac1{b_i+1},\qquad
 \gamma_i(\pi_{it}^--\pi_R)=-\frac{\lambda_{it}}{b_i+1}.
\]
Set $r_i=(b_i+1)\gamma_i$. The supply identity
$H_i=\sum_t\lambda_{it}L_{it}$ cancels the price differences
when we sum over all goods, giving $N=N\pi_R$ and hence
$\pi_R=1$. The extra low group has price $1-a_i/r_i\ge0$,
while $\mu_i=(b_i-r_i)/(b_i+1)\ge0$. Thus $a_i\le r_i\le b_i$,
as required in \Cref{eq:fixp-price-form}.

Each edge agent converts these prices into one reciprocal contribution.
Its small utility gap makes every optimal bundle use only the chosen
low and high groups. Unit quantity gives $q_e^-+q_e^+=1$;
subtracting this equality from the unit-budget equation equates
the saving on low goods with the extra payment for high goods:
\begin{equation}\label{eq:fixp-demand-current}
 \frac{\lambda_e}{r_i}q_e^- =\frac1{r_j}q_e^+
   =\frac{\lambda_e}{r_i+\lambda_er_j}=I_e(r).
\end{equation}
The two quantity and budget equations determine the amounts in
\Cref{eq:fixp-edge-demand}. Consequently the purchases of the agent
for $i\to j$ give the outgoing term at $i$ and the incoming term
at $j$ in the graph field.

Market clearing turns these contributions into the box conditions.
Let $Z_i$ be the quantity of other nodes' goods consumed by node-$i$
agents, and let $T_i$ be the quantity of node-$i$ low goods consumed
by other nodes' agents. Such consumption can use only free goods.
Subtracting total quantity from total spending for the node-$i$
class cancels the reference goods and, by the supply identity,
the fixed high and low supplies. The remaining balance is
\[
 \mathcal G_i(r)=Z_i-T_i,\qquad
 Z_i=0\ \text{if }r_i<b_i,\qquad
 T_i=0\ \text{if }r_i>a_i.
\]
The restrictions on $Z_i$ and $T_i$ follow from the demand-dual
conditions and the prices of free low goods. Therefore the field
is zero in the interior, nonpositive at the lower endpoint, and
nonnegative at the upper endpoint, exactly as in \Cref{eq:fixp-box}.

\begin{proof}[Proof of \Cref{thm:fixp-market}]
The market can be constructed in polynomial time because the graph
parameters are encoded in unary and its edges are listed explicitly.
For a common bound $B\ge2$ on the input quantities, the supplies in
\Cref{eq:fixp-supplies} satisfy $L_{it}=O(B^2)$ and
$H_i,S_*=O(B^3)$, and the market size is $N=O(B^4)$.
The rational utility matrix therefore has polynomial size, with
all entries in $[0,1]$.

Take any returned equilibrium $(X,\mathbf p)$.
\Cref{lem:fixp-market-correctness} gives a graph solution $r$ and
supporting prices without changing $X$. In particular, each edge
agent purchases the amounts in \Cref{eq:fixp-edge-demand}.

The reserved high goods make the recovery map depend only on the
allocation. For the designated edge $A_j\to D_j$, its agent
$w_j$ buys only the reserved high good
$h_j^*$ on the high side. Since the edge ratio is one,
\Cref{eq:fixp-edge-demand} and the designated-pair constraint give
\[
 X_{w_j,h_j^*}=\frac{r_{D_j}}{r_{A_j}+r_{D_j}}
             =1-\frac{r_{A_j}}{4M}.
\]
Consequently $z_j=4-4X_{w_j,h_j^*}$ equals $r_{A_j}/M$, so $z$
is a valid output of \ReciprocalVIProblem. The agent and good indices
are fixed by the construction. Each coordinate is therefore an affine
function of a single returned allocation entry, proving an SL
reduction. \Cref{thm:fixp-vi} proves FIXP-hardness. Since the
recovery map uses only $X$, the reduction also proves hardness
when only the equilibrium allocation is returned.
\end{proof}

\section{An Approximation Algorithm}\label{sec:algorithm-ideas}\label{sec:algorithm}
We construct allocation-dependent assignment weights, find an allocation
that is approximately optimal for its own weights, and obtain prices
from an assignment dual. This gives the following guarantee.

\begin{restatedtheorem}{thm:upper}
\HZUpperStatement
\end{restatedtheorem}

\Cref{sec:framework-overview} introduces the construction and its four
main guarantees. \Cref{sec:field} defines the weights and explains
the VI allocation guarantees, and \Cref{sec:conversion} gives the
conversion to prices. \Cref{sec:demand} establishes the separation
condition for two fixed coefficient choices. Finally,
\Cref{sec:rational-output} specifies the computational precisions,
accounts for the total error, and presents the complete algorithm.

\subsection{Overview of the Construction}\label{sec:framework-overview}
The construction starts by assigning new weights to the agent--good
pairs. We seek an allocation that is optimal under the weights it
induces, and use the corresponding assignment dual to obtain prices.
Four propositions describe the weight rules, the condition needed for
the utility guarantee, a rule satisfying that condition, and the
finite-precision implementation. Throughout the first three statements,
utilities are normalized as in \Cref{sec:model}; the fourth statement
handles arbitrary rational utilities in $[0,1]$.

\paragraph{Coefficient Properties and Allocation Computation.}
Given a feasible allocation $X\in\B_n$ for which our weight rule is
defined, assign a new weight
$G_{ij}(X)$ to each edge $(i,j)$ between an agent and a good.
The coefficient matrix $G(X)$ depends on $X$. Hold these weights
fixed and compare feasible allocations $Y$ by their total weight
$\langle G(X),Y\rangle=\sum_{i,j}G_{ij}(X)Y_{ij}$.
Our goal is to find an $X$ that itself maximizes this objective:
\[
 X\in\arg\max_{Y\in\B_n}\langle G(X),Y\rangle.
\]
In this maximization, the weights are evaluated at $X$ throughout;
only the assignment $Y$ varies. Equivalently, no feasible reassignment
improves the total weight:
\begin{equation}\label{eq:vi}
 X\in\B_n,\qquad \langle G(X),Y-X\rangle\le0
 \quad\text{for every }Y\in\B_n.
\end{equation}
This is the variational inequality for $-G$.

We construct the weights separately for each agent. For a normalized
nonconstant utility row $\mathbf v_i$ and assigned bundle
$\mathbf X_i$, let $g_i(t;\mathbf X_i)$ specify the new weight at
utility level $t\in[0,1]$. The weight of edge $(i,j)$ is then
$G_{ij}(X)=g_i(v_{ij};\mathbf X_i)$.
Below, $g(t;\mathbf x)$ denotes the coefficient function for a single fixed utility row.

For computation, we allow a very small positive tolerance $\eta$
in the assignment condition:
\begin{equation}\label{eq:approx-vi}
 \max_{Y\in\B_n}\langle G(X),Y-X\rangle\le\eta.
\end{equation}
Thus, under the fixed weights $G(X)$, even the best reassignment
improves on $X$ by at most $\eta$. This numerical tolerance can be
made arbitrarily small. We will choose it, together with the other
computational tolerances, sufficiently small that their combined
contribution to the final equilibrium error is at most $\delta$.
The first proposition provides a family of weight rules for which
exact solutions exist and this approximate condition can be computed.

The computational step uses monotonicity to control how assignment
weights change with the allocation. A map $F$ is \emph{monotone} if
$\langle F(X)-F(Y),X-Y\rangle\ge0$ for all $X,Y$ in its domain
\cite{HartmanStampacchia1966,FacchineiPang2003}.

\begin{proposition}\label{prop:allocation-step}
There is a family of coefficient functions, each specified by finitely
many rational parameters independent of the market, such that every
fixed member has the following properties for every normalized utility
matrix.
\begin{enumerate}[label=(\roman*),leftmargin=*]
\item On every nonconstant row and for every bundle in its domain,
$g(\cdot;\mathbf x)$ is nonnegative, nondecreasing, and convex, with
$g(0;\mathbf x)=0$ and $\sum_jx_jg(v_j;\mathbf x)=1$.
Constant rows have zero coefficients, and the map $-G$ is monotone.
\item There is a solution of \Cref{eq:vi} in the domain of $G$.
\item For rational utilities and rational $0<\eta\le1$, a rational
$X\in\B_n$ in the domain of $G$ satisfying \Cref{eq:approx-vi} and
$0\le G_{ij}(X)\le n+\eta$ is computable in time polynomial in the
input length and $1/\eta$.
\end{enumerate}
\end{proposition}

\Cref{sec:field} defines the coefficient family, proves its structural
and allocation guarantees, and concludes with a proof of
\Cref{prop:allocation-step}.

\paragraph{Prices from Assignment Duality.}
Fix $0<\epsilon<1$ and, for now, an exact solution $X$ of
\Cref{eq:vi}. Let $\mathbf q$ be the vector of good prices in an optimal
assignment dual for $G(X)$, normalized by $\min_jq_j=0$.
Nonnegative coefficients, their assigned means, and complementary
slackness imply that each agent spends at most
one at $\mathbf q$, so
$\mathbf p=(1+\epsilon)\mathbf q$ satisfies the spending bound
$1+\epsilon$. It remains to show that every bundle affordable with the
original budget of one at $\mathbf p$ has utility at most the agent's
assigned utility plus $\epsilon$.

The properties in \Cref{prop:allocation-step}(i) give the spending
bound. For the utility bound, we impose an additional requirement on
the weight rule. For a nonconstant row $\mathbf v$ and a bundle
$\mathbf x$ in its domain with $\mathbf v\cdot\mathbf x+\epsilon<1$,
we require
\begin{equation}\label{eq:target-separation}
 g(\mathbf v\cdot\mathbf x+\epsilon;\mathbf x)
       -\min_{j:x_j>0}g(v_j;\mathbf x)\ge\frac1{1+\epsilon}.
\end{equation}
It compares the coefficient at the target utility
$\mathbf v\cdot\mathbf x+\epsilon$ with the smallest coefficient
of a good that the agent receives. The required gap makes a utility
improvement beyond this target too expensive at the scaled dual
prices. If the target is at least one, every unit bundle already
has utility at most the target.

The next proposition derives the equilibrium guarantee from
\Cref{eq:target-separation}. We then give a weight rule satisfying
this condition in \Cref{prop:demand-step}.

\begin{proposition}\label{prop:price-step}
Fix $0<\epsilon<1$. Suppose the coefficients satisfy
\Cref{prop:allocation-step}(i) and \Cref{eq:target-separation} for every
normalized nonconstant row and every bundle in its domain, and vanish
on constant rows. If $X$ solves \Cref{eq:vi} and $\mathbf q$ is the
vector of good prices in an optimal assignment dual for $G(X)$ with
$\min_jq_j=0$, then $(X,(1+\epsilon)\mathbf q)$ is an
$\epsilon$-approximate HZ equilibrium.
\end{proposition}

\Cref{sec:conversion} proves the dual-price and single-agent results,
then combines them to prove \Cref{prop:price-step}.

\paragraph{A Scalar Criterion for the Utility Guarantee.}
Checking \Cref{eq:target-separation} directly requires a guarantee
for every normalized nonconstant utility row and every bundle in its
domain. Both the utility values and the assigned probabilities vary,
and the number of variables grows with the number of goods.
We reduce this task to verifying a sufficient condition in three
variables. Regard the assigned bundle $\mathbf x$ as a lottery,
in which utility $v_j$ occurs with probability $x_j$.
Whenever the target utility is below one, set
\[
 b=\min\{v_j:x_j>0\},\qquad
 t=\mathbf v\cdot\mathbf x+\epsilon,\qquad
 u=\frac{\sum_jx_j(v_j-t)_+}{1-t}.
\]
Here $b$ is the lowest supported utility, $t$ is the target utility,
and $u$ is the expected excess above $t$, normalized by the largest
possible excess $1-t$. Since $g$ is nondecreasing, the coefficient
difference in \Cref{eq:target-separation} is
$g(t;\mathbf x)-g(b;\mathbf x)$.

\Cref{sec:demand} bounds this difference using the three quantities
$b,t,u$. Under the coefficient bounds stated there, verifying
$S(b,t,u)\ge0$ throughout an explicit three-variable domain is
sufficient to establish \Cref{eq:target-separation} for every row
and bundle. This replaces the original family of checks, whose
dimension depends on the market size, by one fixed-dimensional test.
For the chosen rational coefficient data and target error,
we prove the inequality using rigorous interval arithmetic,
together with analytic estimates covering the remaining regions.
These estimates cover the entire domain and give the uniform separation
guarantee in the following proposition.

\begin{proposition}\label{prop:demand-step}
There is a fixed member of the family in \Cref{prop:allocation-step}
that satisfies \Cref{eq:target-separation} at $\epsilon=0.0755211$ for
every normalized nonconstant row and every bundle in the domain of $g$.
\end{proposition}

\Cref{sec:certified-consequences} proves \Cref{prop:demand-step}
from the scalar criterion and coefficient bounds stated in
\Cref{sec:demand}; their supporting proofs are in \Cref{app:ch3-demand}.

\paragraph{Rational Computation and the Error Budget.}
The first three parts establish the guarantee for an exact VI
solution. For computation, we round the utilities, compute a rational
approximate VI allocation and coefficient matrix, and obtain prices
from an optimal assignment dual. The VI, coefficient, and rounding
errors must be controlled together so that the output satisfies the
spending and utility bounds with at most $\delta$ additional error,
while keeping unit demand and market clearing exact.

\begin{proposition}\label{prop:rational-equilibrium}
Fix a member of the family in \Cref{prop:allocation-step} and a
rational target $0<\epsilon<1$ for which \Cref{eq:target-separation} holds for every
normalized nonconstant row and every bundle in its domain. For rational
input $v\in[0,1]^{n\times n}$ and $\delta>0$, a rational pair forming an
$(\epsilon+\delta)$-approximate HZ equilibrium is computable in time
polynomial in the input length and $1/\delta$.
\end{proposition}

\Cref{sec:rational-output} specifies the precisions and proves
\Cref{prop:rational-equilibrium} by combining the computation and
rounding bounds. The proofs of those bounds are in \Cref{app:ch3-rational}.

\paragraph{How the Guarantees Fit Together.}\label{app:ch3-overview}
For an exact VI allocation, \Cref{prop:allocation-step,prop:demand-step,prop:price-step}
give an equilibrium with error $\epsilon_0=0.0755211$.
The algorithm instead uses an approximate VI allocation and the
finite-precision guarantee in \Cref{prop:rational-equilibrium}.

\begin{proof}[Proof of \Cref{thm:upper}]
Choose the fixed coefficients supplied by \Cref{prop:demand-step}.
These coefficients satisfy separation at the rational target $\epsilon_0=0.0755211$,
so \Cref{prop:rational-equilibrium} applies to these coefficients and target.
For every rational input utility matrix in $[0,1]^{n\times n}$ and
requested accuracy $\delta>0$, that proposition returns a rational
$(\epsilon_0+\delta)$-approximate HZ equilibrium in time polynomial in
the input length and $1/\delta$. By the definition of approximate HZ
equilibrium, unit demand and market clearing hold exactly, and the
minimum price is zero. This is the claimed guarantee.
\end{proof}

\subsection{Coefficient Properties and Allocation Computation}\label{sec:field}\label{app:ch3-field}
We prove the three parts of \Cref{prop:allocation-step} in four stages.
First, we define the coefficients and establish their shape,
normalization, and monotonicity. Second, restricted VI solutions
provide an exact allocation in the coefficient domain. Third, we
bound coefficient variation and evaluation error. Finally, these
bounds yield a rational approximate VI allocation. The quantitative
mass and stability estimates will also be used in
\Cref{sec:rational-output}.

\subsubsection{The Coefficient Family and Its Properties}\label{sec:structure}\label{app:ch3-structure}
\paragraph{Definition of the Coefficients.}\label{sec:field-definition}
For a normalized nonconstant utility row $\mathbf v$ and a unit bundle
$\mathbf x\in\Delta_n$ assigning positive mass to utility-one goods,
let $L(a)=\sum_{j:v_j\ge a}x_j$ be the probability of receiving utility
at least $a$. We define the coefficients on this domain as follows.

\begin{definition}[Allocation-Dependent Coefficients]\label{def:field}
Choose $0<s<1$ and $0<T<1$, put $\alpha=s^2$ and $\beta=(s+2)^2$, and set
\begin{equation}\label{eq:tail-weight}
 d(z)=\begin{cases}
 T^{\alpha-\beta}z^{-\alpha},&0<z\le T,\\
 z^{-\beta},&T\le z\le1.
 \end{cases}
\end{equation}
Choose a correction function $f:[0,1]\to\mathbb R_{\ge0}$ such that $(1-t)f(t)$ is
continuous, nonnegative, concave, and piecewise
affine on $[0,1]$, with value zero at both endpoints. At one, $f$
takes its finite left limit. Require a constant $A\ge0$ such that
\begin{equation}\label{eq:shape}
 (1-t)f(t)\le At\quad(0\le t\le1),\qquad
 f(1)<1,\qquad A\le T^{\beta-2},\qquad 2T-1\ge T^\beta.
\end{equation}
For each normalized nonconstant row $\mathbf v$ and each
$\mathbf x\in\Delta_n$ with $\sum_{j:v_j=1}x_j>0$, define, for $0\le t\le1$,
\begin{equation}\label{eq:field}
 g(t;\mathbf{x})=
 \frac{\displaystyle\int_0^t d(L(a))\,da}
      {\displaystyle\int_0^1L(a)d(L(a))\,da}
 +\sum_j(t-v_j)f(\min\{t,v_j\})x_j.
\end{equation}
For a constant utility row, set $g(t;\mathbf{x})=0$.
The defining data are $s,T$ and the breakpoints and values of
the piecewise-linear function $h(t)=(1-t)f(t)$.
\end{definition}

For agent $i$, write $g_i$ for its coefficient function. These
functions define the coefficient matrix $G(X)$ by
$G_{ij}(X)=g_i(v_{ij};\mathbf X_i)$. The domain of the map
$X\mapsto G(X)$ consists of allocations in $\B_n$ giving positive
mass to utility-one goods in every nonconstant row. Rational choices of the defining data, fixed
independently of the market, form the family in
\Cref{prop:allocation-step}.

The two terms in \Cref{eq:field} have complementary roles:
\begin{itemize}[leftmargin=*]
\item \textbf{The integral term has assigned mean one.}
Its decreasing tail weight $d$ makes coefficients increase faster over
utility levels that the assigned lottery rarely reaches. The two
powers in \Cref{eq:tail-weight} control this response on the two sides
of $T$; their paired exponents preserve monotonicity with respect
to the allocation.
\item \textbf{The correction term has assigned mean zero.}
Its contributions cancel in pairs when averaged with the assigned
masses. It therefore changes the coefficient differences used in
\Cref{eq:target-separation} without changing their assigned mean.
The shape and size restrictions in \Cref{eq:shape} preserve the
properties summarized below.
\end{itemize}

Throughout this subsection, for a fixed normalized nonconstant row
$\mathbf v$ and a unit bundle $\mathbf x$ in the coefficient domain,
write
\[
 L(a)=\sum_{j:v_j\ge a}x_j,\qquad
 D(t)=\int_0^t d(L(a))\,da,\qquad
 E=\int_0^1L(a)d(L(a))\,da.
\]
Thus $D(t)/E$ is the integral term of $g(t;\mathbf x)$.
The symbols $L,D,E$ always refer to the same row and bundle;
subscripts will identify a different row or a bundle varying along
a segment. Positive mass on utility-one goods keeps every tail
positive, so both integrals are finite sums of finite quantities
over the utility intervals. We first bound the denominator and
verify normalization.

\begin{lemma}\label{lem:normalizer}
For the coefficients in \Cref{def:field}, fix a normalized nonconstant
row $\mathbf v$ and a unit bundle $\mathbf x$ assigning positive mass
to utility one. With $E=\int_0^1L(a)d(L(a))\,da$, write
$\bar v=\mathbf v\cdot\mathbf x$.
Then
\[
 0<E\le T^{1-\beta},\qquad
 \bar v E\le T^{2-\beta},\qquad
 \sum_jx_jg(v_j;\mathbf{x})=1,\qquad g(0;\mathbf{x})=0.
\]
\end{lemma}
\begin{proof}[Proof of \Cref{lem:normalizer}]
Fix $\mathbf v,\mathbf x$ as in the statement and write
$\bar v=\mathbf v\cdot\mathbf x$.
We first bound $z d(z)$, extending it continuously to value zero
at $z=0$. Its low branch is increasing and its high branch is
decreasing, so its maximum is $T^{1-\beta}$. Moreover,
\[
 z d(z)\le T^{1-\beta}(2-z/T)\qquad(0\le z\le1).
\]
For $z\le T$, the right side is at least the maximum
$T^{1-\beta}$. For $T\le z\le1$, the convex function
$z d(z)=z^{1-\beta}$ is at most the chord joining its values at
$T$ and $1$. The displayed line dominates that chord: the lines
agree at $T$, and their comparison at $1$ is
$2T-1\ge T^\beta$, as required in \Cref{eq:shape}.

Every tail is positive, so $E>0$. Integrating the maximum bound and
the linear bound, and using $\int_0^1L(a)\,da=\bar v$, gives
\[
 E\le T^{1-\beta},\qquad
 E\le T^{1-\beta}(2-\bar v/T).
\]
Multiplying the latter by $\bar v$ and using
$\bar v(2-\bar v/T)\le T$ proves $\bar v E\le T^{2-\beta}$.

Exchanging the finite sum over goods with the integral gives
$\sum_jx_jD(v_j)=\int_0^1L(a)d(L(a))\,da=E$.
The correction term in \Cref{eq:field} has assigned mean zero because
$(v_j-v_k)f(\min\{v_j,v_k\})$ changes sign when $j,k$ are exchanged,
so $\sum_jx_jg(v_j;\mathbf x)=1$.
Finally, $(1-0)f(0)=0$ gives $f(0)=0$, and $D(0)=0$; hence
$g(0;\mathbf x)=0$.
\end{proof}

The next lemma gives the coefficient's shape as a function of
utility. These properties will later turn a utility improvement
into a lower bound on its price.

\begin{lemma}\label{lem:coefficient-shape}
For the coefficients in \Cref{def:field}, $f$ is nondecreasing.
For every normalized nonconstant row and every bundle in its
coefficient domain, $g(\cdot;\mathbf x)$ is nonnegative,
nondecreasing, and convex.
\end{lemma}
\begin{proof}
Fix a normalized nonconstant row $\mathbf v$ and a bundle
$\mathbf x$ in its coefficient domain. Write $D,E$ for
the integral numerator and normalizer associated with this row and
bundle, and put $\bar v=\mathbf v\cdot\mathbf x$.

On an affine piece, write $(1-t)f(t)=a_f t+b_f$.
Concavity and the zero right endpoint give $a_f+b_f\ge0$, so
$f(t)=(a_f+b_f)/(1-t)-a_f$ is nondecreasing on that piece.
Continuity joins the interior pieces, and the final piece gives a
constant $f$ with its prescribed finite left limit at one; thus $f$
is nondecreasing on $[0,1]$.

For fixed $w\in[0,1]$, consider the kernel
$(t-w)f(\min\{t,w\})$, which is linear for $t\ge w$.
For $t<w$, its second derivative within an affine piece is
$2(1-w)(a_f+b_f)/(1-t)^3\ge0$.
At an interior knot, the slope of $(1-t)f(t)$ decreases, and the
factor $t-w<0$ makes the kernel's first derivative jump upward.
The derivatives match at $t=w$, so the kernel is convex; for $w=0$
it vanishes, and for $w=1$ it is the continuous function
$-(1-t)f(t)$.

The numerator $D$ is convex because $L$ is nonincreasing and $d$
is decreasing, so division by $E>0$ and addition of the convex
correction term preserve convexity. Also, $d(z)\ge1$ gives $D(t)\ge t$;
discarding the nonnegative correction terms and using
$(v_j-t)_+\le(1-t)v_j$ yields
\[
 g(t;\mathbf x)\ge t/E-f(t)\sum_jx_j(v_j-t)_+
 \ge t(1/E-A\bar v)\ge0.
\]
The middle inequality uses $(1-t)f(t)\le At$, and the last uses
\Cref{lem:normalizer} and $A\le T^{\beta-2}$.
Since \Cref{lem:normalizer} gives $g(0;\mathbf x)=0$,
nonnegativity and convexity imply that $g(\cdot;\mathbf x)$
is nondecreasing on $[0,1]$.
\end{proof}

We next show that the negative coefficient map is monotone as the
allocation changes. Together with \Cref{lem:normalizer,lem:coefficient-shape},
this gives all the properties needed in \Cref{prop:allocation-step}(i).

\begin{proposition}\label{prop:structure}
Fix coefficients from \Cref{def:field} and a normalized
nonconstant utility row $\mathbf v$. The map
$\mathbf x\mapsto-(g(v_j;\mathbf x))_{j=1}^n$ is monotone
on the set of bundles in $\Delta_n$ assigning positive mass
to utility-one goods. With zero coefficients on constant rows, the matrix
map $-G$ is monotone on its domain in $\B_n$.
\end{proposition}
\begin{proof}
Fix two bundles $\mathbf x,\mathbf y$
in the domain for the same utility row, and join them by
$\mathbf x(\theta)=(1-\theta)\mathbf x+\theta\mathbf y$,
$0\le\theta\le1$. Put $\Delta_j=y_j-x_j$ and write
$L_\theta,D_\theta,E_\theta$ for the corresponding tail and integrals;
the mass at utility one is at least the smaller of the two endpoint
masses, so every tail stays uniformly positive. We have
\[
 \frac{d x_j(\theta)}{d\theta}=\Delta_j,
 \qquad
 \frac{dL_\theta(a)}{d\theta}=\sum_{j:v_j\ge a}\Delta_j.
\]
Define the relative tail variation and its weighted mean by
\[
 \psi_\theta(a)=\frac{1}{L_\theta(a)}\frac{dL_\theta(a)}{d\theta},
 \qquad
 \bar\psi_\theta=\frac1{E_\theta}\int_0^1
       L_\theta(a)d(L_\theta(a))\psi_\theta(a)\,da.
\]
The weights $L_\theta(a)d(L_\theta(a))/E_\theta$ integrate to one.
Set $\rho_\theta(a)=\alpha$ when $L_\theta(a)\le T$ and
$\rho_\theta(a)=\beta$ when $L_\theta(a)>T$.

There are finitely many utility intervals, and their tail values
vary affinely in $\theta$. A tail that stays equal to $T$ has
zero derivative; the other tails cross $T$ at only finitely many
values of $\theta$. Away from these crossings, differentiating and
exchanging finite sums with the integrals gives
\[
 \begin{aligned}
 \sum_j\Delta_j\frac{dD_\theta(v_j)}{d\theta}
   &=-\int_0^1L_\theta(a)d(L_\theta(a))\rho_\theta(a)\psi_\theta(a)^2\,da,\\
 \sum_j\Delta_jD_\theta(v_j)&=E_\theta\bar\psi_\theta,\\
 \frac{dE_\theta}{d\theta}
   &=\int_0^1L_\theta(a)d(L_\theta(a))(1-\rho_\theta(a))\psi_\theta(a)\,da.
 \end{aligned}
\]
The first exchange gives
$\int_0^1d'(L_\theta(a))(dL_\theta(a)/d\theta)^2\,da$,
which has the displayed form because
$\rho_\theta=-L_\theta d'(L_\theta)/d(L_\theta)$ on each branch.
The second identity is the same exchange without differentiation,
and the third is the product rule; constant tails at $T$ contribute
zero to all derivative terms.

To prove monotonicity, it remains to show that the negative
coefficient change has a nonnegative inner product with
$\boldsymbol\Delta$. The quotient rule expresses the contribution
of the integral term as
\begin{equation}\label{eq:jacobian}
 \begin{aligned}
 -\sum_j\Delta_j\frac{d}{d\theta}\left(\frac{D_\theta(v_j)}{E_\theta}\right)
 =\frac1{E_\theta}\int_0^1& L_\theta(a)d(L_\theta(a))\\[-2pt]
 &{}\cdot\bigl[\rho_\theta(a)\psi_\theta(a)^2
       +(1-\rho_\theta(a))\bar\psi_\theta\psi_\theta(a)\bigr]\,da.
 \end{aligned}
\end{equation}
The first term comes from differentiating the numerator and the
second from differentiating the denominator; we next show their
weighted integral is nonnegative.

The paired exponents $\alpha=s^2$ and $\beta=(s+2)^2$ were
chosen so that the same square completion applies on both branches:
for both $\rho\in\{\alpha,\beta\}$,
$1-\rho-(s+1)^2=-2(s+1)\sqrt\rho$.
Using the mean identity
$\int_0^1L_\theta(a)d(L_\theta(a))\psi_\theta(a)\,da
=E_\theta\bar\psi_\theta$ completes the square after integration:
\begin{equation}\label{eq:quadratic}
 \begin{aligned}
 -\sum_j\Delta_j\frac{d}{d\theta}\left(\frac{D_\theta(v_j)}{E_\theta}\right)
 =\frac1{E_\theta}\int_0^1& L_\theta(a)d(L_\theta(a))\\[-2pt]
 &{}\cdot\bigl(\sqrt{\rho_\theta(a)}\psi_\theta(a)
                 -(s+1)\bar\psi_\theta\bigr)^2\,da\ge0.
 \end{aligned}
\end{equation}

Integrating this inequality along the segment, with continuity across
the finite branch crossings, gives
\[
 -\sum_j\Delta_j
       \left(\frac{D_1(v_j)}{E_1}-\frac{D_0(v_j)}{E_0}\right)\ge0.
\]
The endpoints are the integral coefficients at $\mathbf y$ and
$\mathbf x$, respectively, so this is their negative map's
monotonicity inequality. For the correction term, define the matrix
$C_{jk}=(v_j-v_k)f(\min\{v_j,v_k\})$.
The matrix $C$ is antisymmetric. When the allocation changes from
$\mathbf x$ to $\mathbf y$, the correction vector changes by
$C\boldsymbol\Delta$. Hence
$\boldsymbol\Delta^{\mathsf T}C\boldsymbol\Delta=0$,
so the correction term does not alter the monotonicity inequality.
Adding the two contributions gives
\[
 -\sum_j\Delta_j\bigl(g(v_j;\mathbf y)-g(v_j;\mathbf x)\bigr)\ge0.
\]
Summing over nonconstant rows and using zero coefficients on constant
rows proves matrix monotonicity of $-G$; its domain in $\B_n$ is
convex because the mass on utility-one goods remains positive
in every nonconstant row along a segment.
\end{proof}

\subsubsection{Existence of an Exact VI Allocation}\label{sec:exact-vi}\label{app:ch3-exact-vi}
The coefficient domain excludes allocations giving no mass to utility-one
goods in a nonconstant row. We obtain an exact solution by first
requiring every entry to be positive and then letting this
restriction vanish. A uniform positive lower bound on the mass
on utility-one goods keeps the limit in the coefficient domain.

For $0<\kappa<1/(2n)$, define
\[
 \mathcal P_\kappa=\{X\in\B_n:X_{ij}\ge\kappa\}
 =\{\kappa J_{\rm all}+(1-n\kappa)Y:Y\in\B_n\},
\]
where $J_{\rm all}$ is the all-ones matrix.
A bound on the coefficients prevents an agent's mass on utility-one
goods from approaching zero. We state this implication first,
since it also applies to the approximate allocations computed later.

\begin{lemma}\label{lem:best-good-mass}
Fix coefficients from \Cref{def:field} and a normalized utility
matrix whose entries below one are at most $1-\nu$, with $0<\nu\le1$.
If $X$ lies in the coefficient domain and $G_{ij}(X)\le M$ for
some $M\ge1$, then every nonconstant row $i$ satisfies
\[
 \sum_{j:v_{ij}=1}X_{ij}
 \ge\left(\frac{\nu}{ME_{\max}}\right)^{k_{\rm mass}},
\]
for every $E_{\max}\ge T^{1-\beta}$ and integer $k_{\rm mass}\ge1/\alpha$.
\end{lemma}
\begin{proof}
Fix a nonconstant row $i$, and write
$\ell_i=\sum_{j:v_{ij}=1}X_{ij}=L_i(1)>0$, with $D_i,E_i$
referring to the same row and bundle $\mathbf X_i$.
The interval from the largest utility below one to one has length
at least $\nu$ and tail $\ell_i$, so $D_i(1)\ge\nu d(\ell_i)$.
The correction term at utility one is nonnegative,
\Cref{lem:normalizer} gives $E_i\le T^{1-\beta}\le E_{\max}$,
and both branches of \Cref{eq:tail-weight} satisfy
$d(z)\ge z^{-\alpha}$. Consequently,
\[
 \begin{aligned}
 M&\ge g_i(1;\mathbf X_i)
          \ge\frac{\nu}{E_{\max}}\ell_i^{-\alpha},\\
 \ell_i&\ge\left(\frac{\nu}{M E_{\max}}\right)^{1/\alpha}
          \ge\left(\frac{\nu}{M E_{\max}}\right)^{k_{\rm mass}}.
 \end{aligned}
\]
The last inequality uses $k_{\rm mass}\ge1/\alpha$ and a base at
most one: $\nu\le1$, $M\ge1$, and
$E_{\max}\ge T^{1-\beta}>1$.
\end{proof}

We next extend an approximate assignment inequality from
$\mathcal P_\kappa$ to all of $\B_n$. Its coefficient bound will
also let us apply \Cref{lem:best-good-mass} uniformly in $\kappa$.

\begin{lemma}\label{lem:restricted-vi-estimates}
Fix coefficients from \Cref{def:field} and a normalized utility
matrix. Let $0<\kappa<1/(2n)$ and $0\le\eta_\kappa\le1$.
If $X\in\mathcal P_\kappa$ satisfies
\[
 \max_{Y\in\mathcal P_\kappa}\langle G(X),Y-X\rangle
 \le\eta_\kappa,
\]
then, with $M_\kappa=(n+\eta_\kappa)/(1-n\kappa)$,
\[
 0\le G_{ij}(X)\le M_\kappa\le4n,
 \qquad
 \max_{Y\in\B_n}\langle G(X),Y-X\rangle
 \le\frac{\eta_\kappa+n^2\kappa}{1-n\kappa}.
\]
\end{lemma}
\begin{proof}[Proof of \Cref{lem:restricted-vi-estimates}]
\emph{Coefficient bound.}
By \Cref{lem:normalizer,lem:coefficient-shape}, nonconstant rows have nonnegative
coefficients with assigned mean one; constant rows have zero
coefficients. Thus $\langle G(X),X\rangle\le n$.
For a permutation matrix $P$, comparison with its feasible affine
image gives
\[
 \kappa\langle G(X),J_{\rm all}\rangle
 +(1-n\kappa)\langle G(X),P\rangle\le n+\eta_\kappa.
\]
Choose a permutation containing any prescribed entry and discard the
other nonnegative terms to obtain $G_{ij}(X)\le M_\kappa$.
Since $\eta_\kappa\le1$, $n\ge1$, and $1-n\kappa>1/2$, this
bound is at most $2(n+1)\le4n$.

\emph{Comparison with all feasible allocations.}
For arbitrary $Y\in\B_n$, its affine image
$\kappa J_{\rm all}+(1-n\kappa)Y$ also belongs to
$\mathcal P_\kappa$. Use that image in the restricted VI, discard the
nonnegative $\kappa\langle G(X),J_{\rm all}\rangle$ term, and subtract
$(1-n\kappa)\langle G(X),X\rangle$ to obtain
\[
 (1-n\kappa)\langle G(X),Y-X\rangle
 \le\eta_\kappa+n\kappa\langle G(X),X\rangle
 \le\eta_\kappa+n^2\kappa.
\]
Dividing by $1-n\kappa>0$ and maximizing over $Y\in\B_n$
gives the stated bound.
\end{proof}

The uniform mass bound lets us pass to the limit as $\kappa$
tends to zero and obtain an exact solution.

\begin{corollary}\label{cor:exact-vi}
For every finite normalized utility matrix, the coefficient map $G$ from
\Cref{def:field} admits an allocation $X\in\B_n$ giving positive mass to utility-one
goods in every nonconstant row and satisfying
\[
 \langle G(X),Y-X\rangle\le0\qquad(Y\in\B_n).
\]
\end{corollary}
\begin{proof}
If all rows are constant, then $G=0$ and any allocation in $\B_n$
suffices. Otherwise, let $\nu$ be the minimum of $1-v_{ij}$ over all
entries below one in nonconstant rows. Finiteness and normalization
give $\nu>0$. Fix
$E_{\max}\ge T^{1-\beta}$ and an integer
$k_{\rm mass}\ge1/\alpha$.
The map $G$ is continuous on its domain: for each fixed utility row,
the integrals in \Cref{eq:field} are finite sums of continuous
functions of positive tails over fixed utility intervals, their
denominator is positive, and the correction term is linear in the
allocation. For every $0<\kappa<1/(2n)$, the polytope
$\mathcal P_\kappa$ is nonempty, compact, convex, and contained in
this domain. The continuous VI existence theorem therefore gives an
exact restricted solution $X^{(\kappa)}$ \cite{FacchineiPang2003}.

Apply \Cref{lem:restricted-vi-estimates} with $\eta_\kappa=0$,
and then \Cref{lem:best-good-mass} with $M=4n$.
These give the common positive lower bound
$m_0=(\nu/(4nE_{\max}))^{k_{\rm mass}}$ on the mass assigned
to utility-one goods in every nonconstant row, together with
\[
 \max_{Y\in\B_n}\langle G(X^{(\kappa)}),Y-X^{(\kappa)}\rangle
 \le\frac{n^2\kappa}{1-n\kappa},
\]
whose right side tends to zero with $\kappa$.
Choose a sequence of positive $\kappa$ tending to zero and use
compactness of $\B_n$ to take a convergent subsequence, with limit
$X$. The mass on utility-one goods is a continuous linear function of each row, so
its lower bound $m_0>0$ persists at $X$, keeping $X$ in the domain
of $G$. By continuity of $G$ on its domain, for every fixed
$Y\in\B_n$ the displayed inequality passes to the limit and gives
$\langle G(X),Y-X\rangle\le0$.
Thus $X$ satisfies \Cref{eq:vi}.
\end{proof}

\subsubsection{Coefficient Stability and Evaluation}\label{sec:vi-rounding}\label{app:ch3-lipschitz}
To compute a VI allocation, we need quantitative continuity and
accurate evaluations of $G$. A lower bound on the mass assigned to utility-one goods gives
both. The same bound will later determine the computational
precisions in \Cref{sec:rational-output}.

We express this stability by a Lipschitz bound: a map $F$ is
$K$-Lipschitz if $\|F(x)-F(y)\|_{\rm out}\le K\|x-y\|_{\rm in}$ for
all $x,y$ in its domain, with the input and output norms specified
\cite{FacchineiPang2003}. The next lemma uses $\ell_1$ for bundles
and $\ell_\infty$ for coefficient vectors, so it bounds every coefficient
error by the total change in the bundle.

\begin{lemma}\label{lem:field-lipschitz}
Fix coefficients from \Cref{def:field} and a normalized nonconstant utility
row $\mathbf v$. On the convex set of bundles in $\Delta_n$ assigning mass at least
$m>0$ to utility-one goods,
the map $\mathbf x\mapsto(g(v_j;\mathbf x))_{j=1}^n$ is Lipschitz from
$\ell_1$ to $\ell_\infty$ with the bound
\begin{equation}\label{eq:field-lipschitz}
 \max_j|g(v_j;\mathbf x)-g(v_j;\mathbf y)|
 \le K_m\|\mathbf x-\mathbf y\|_1,
 \qquad K_m=\frac{2\beta-1}{2m^2}+\frac{f(1)}2.
\end{equation}
Its entries are at most $m^{-1}+f(1)$.
\end{lemma}
\begin{proof}[Proof of \Cref{lem:field-lipschitz}]
We may assume $0<m\le1$, since a unit bundle has total mass one.

\emph{Entry bound.}
The integral normalization in the proof of \Cref{lem:normalizer}
gives $\sum_jx_jD(v_j)/E=1$.
Utility-one goods have integral coefficient $D(1)/E$ and total mass
at least $m$, while all integral coefficients are nonnegative, so
$mD(1)/E\le1$. Since $D$ is nondecreasing and the correction
kernel is bounded above by $f(1)$, every entry is at most
$m^{-1}+f(1)$.

\emph{Lipschitz bound.}
Fix two bundles $\mathbf x,\mathbf y$ in this domain and join them by
$\mathbf x(\theta)=(1-\theta)\mathbf x+\theta\mathbf y$,
$0\le\theta\le1$; each intermediate bundle still gives mass at least
$m$ to utility-one goods. Put $\boldsymbol\Delta=\mathbf y-\mathbf x$
and use $L_\theta,D_\theta,E_\theta$ as in the structural proof.
The vector $\boldsymbol\Delta$ has sum zero, so its positive and
negative parts each have mass $\|\boldsymbol\Delta\|_1/2$, giving
\[
 \left|\frac{dL_\theta(a)}{d\theta}\right|
 =\left|\sum_{j:v_j\ge a}\Delta_j\right|
 \le\frac{\|\boldsymbol\Delta\|_1}{2}.
\]

On either branch, $\rho=-z d'(z)/d(z)$ belongs to
$\{\alpha,\beta\}$; since $0<\alpha<1$ and $\beta>4$, we have
$\rho\le\beta$ and $|1-\rho|\le\beta-1$.
Every tail is at least $m$, so the preceding variation bound gives
\[
 \left|\frac{dD_\theta(t)}{d\theta}\right|
 \le\frac{\beta D_\theta(t)}{2m}\|\boldsymbol\Delta\|_1,
 \qquad
 \left|\frac{dE_\theta}{d\theta}\right|
 \le\frac{(\beta-1)E_\theta}{2m}\|\boldsymbol\Delta\|_1.
\]
Indeed, the pointwise estimates are
$|d'(L_\theta)|\le\beta d(L_\theta)/m$ and
$|d(L_\theta)+L_\theta d'(L_\theta)|
\le(\beta-1)d(L_\theta)$, with
$d(L_\theta)\le L_\theta d(L_\theta)/m$ for the denominator bound.
The quotient rule and $D_\theta(t)/E_\theta\le1/m$ then give
\[
 \left|\frac{d}{d\theta}\left(\frac{D_\theta(t)}{E_\theta}\right)\right|
 \le\frac{2\beta-1}{2m^2}\|\boldsymbol\Delta\|_1.
\]

For fixed $t$, the correction kernel
$(t-v_j)f(\min\{t,v_j\})$ lies in
$[-(1-t)f(1),tf(1)]$, an interval of length $f(1)$.
Subtracting its midpoint before summing against the zero-sum vector
$\boldsymbol\Delta$ bounds the change in the correction term by
$f(1)\|\boldsymbol\Delta\|_1/2$.
Adding the two estimates and integrating along the segment yields
\[
 |g(t;\mathbf y)-g(t;\mathbf x)|
 \le\left(\frac{2\beta-1}{2m^2}+\frac{f(1)}2\right)
       \|\mathbf y-\mathbf x\|_1.
\]
Taking the maximum over $t=v_j$ proves \Cref{eq:field-lipschitz};
tails that stay at $T$ have zero derivative, and continuity joins
the estimates across the remaining finite branch crossings.
\end{proof}

The Lipschitz bound controls how input errors affect coefficients.
We also need to compute the coefficients themselves to a specified
accuracy. Their finite-sum representation gives this second guarantee.

\begin{lemma}\label{lem:field-evaluation}
For rational input utilities and allocations in the coefficient
domain, and fixed rational defining data for $d,f$, the coefficients
can be evaluated to any rational entrywise error $\sigma>0$
in time polynomial in the input bit length and $\log(1+1/\sigma)$.
\end{lemma}
\begin{proof}
For rational utilities and allocations, $L$ is constant on intervals
with rational endpoints, so the integrals are finite sums of fixed
rational powers of positive rationals, including the fixed factor
$T^{\alpha-\beta}$.
Range reduction and logarithm and exponential series with explicit
remainder bounds evaluate these sums in time polynomial in the input
bit length and the requested number of accuracy bits.
Since $d(z)\ge1$, the denominator is at least the mass on utility-one goods;
the logarithm of the reciprocal of this positive rational mass is
bounded by the input bit length, so division requires only polynomially
many additional accuracy bits.
The function $f$ is evaluated from its piecewise formula, using the
finite left limit at utility one, which gives the claimed entrywise
accuracy and running time.
\end{proof}

\subsubsection{Computing an Approximate VI Allocation}\label{sec:vi-subroutine}\label{app:ch3-computation}
We apply a general monotone-VI subroutine to $\mathcal P_\kappa$
and then use \Cref{lem:restricted-vi-estimates} to obtain the
bound over $\B_n$. We first explain how linear minimization suffices
for the projection steps used by the subroutine. All norms below
are Euclidean.

Let $\mathcal P\subseteq\mathbb R^{N_{\rm VI}}$ be a nonempty
compact rational polytope with a given rational point. Assume a
polynomial-time linear minimization oracle whose outputs have bit
length bounded by a polynomial in the description length of
$\mathcal P$, and let $D_{\rm VI}\ge1$ be a rational diameter bound.
Write $\Pi_{\mathcal P}$ for Euclidean projection onto $\mathcal P$.

\begin{lemma}\label{lem:rational-projection}
For rational $\mathbf z$ and $0<\epsilon_{\rm proj}\le1$, one can
compute a rational $\mathbf x\in\mathcal P$ with
$\|\mathbf x-\Pi_{\mathcal P}(\mathbf z)\|\le\epsilon_{\rm proj}$
using $O(D_{\rm VI}^2/\epsilon_{\rm proj}^2)$ oracle calls.
The output bit length is polynomial in the description length of
$\mathcal P$, $D_{\rm VI}$, and $1/\epsilon_{\rm proj}$,
independently of $\mathbf z$.
\end{lemma}
\begin{proof}
For a rational point $\mathbf z$, minimize
$Q_{\mathbf z}(\mathbf x)=\|\mathbf x-\mathbf z\|^2/2$
over $\mathcal P$ using the linear minimization oracle.
The curvature of this quadratic on $\mathcal P$ is at most
$D_{\rm VI}^2$. Frank--Wolfe steps of size $2/(k+2)$ give
objective error at most $2D_{\rm VI}^2/(k+2)$ for $k\ge1$
\cite{Jaggi2013}. Projection optimality also gives
\[
 \|\mathbf x-\Pi_{\mathcal P}(\mathbf z)\|^2
 \le2\bigl(Q_{\mathbf z}(\mathbf x)
                  -\min_{\mathbf y\in\mathcal P}Q_{\mathbf z}(\mathbf y)\bigr).
\]
Hence $O(D_{\rm VI}^2/\epsilon_{\rm proj}^2)$ oracle calls
produce a point in $\mathcal P$ within any prescribed distance
$\epsilon_{\rm proj}>0$ of the exact projection.

Each projection computation starts with a full Frank--Wolfe step,
so its output is a convex combination solely of the oracle outputs
from that computation. These outputs have a common polynomial
bit-length bound by hypothesis. The rational step sizes and the
polynomial number of projection steps therefore give a polynomial
bit-length bound for the projection output, independent of $\mathbf z$.
\end{proof}

We now use these projections to compute a solution of a monotone
variational inequality. Let $\mathcal F:\mathcal P\to\mathbb R^{N_{\rm VI}}$
be monotone and $K_{\rm VI}$-Lipschitz, with
$\|\mathcal F(\mathbf x)\|\le B_{\rm VI}$ on $\mathcal P$, where
$K_{\rm VI},B_{\rm VI}\ge1$ are rational. Assume that
$\mathcal F$ can be evaluated with a specified error in time
polynomial in the input length and the requested number of accuracy bits.

\begin{lemma}\label{lem:rational-vi}
Under these assumptions, for every rational $0<\eta\le1$ one can
compute a rational $\mathbf x\in\mathcal P$ satisfying
\[
 \max_{\mathbf y\in\mathcal P}
 \langle\mathcal F(\mathbf x),\mathbf x-\mathbf y\rangle\le\eta
\]
in time polynomial in the input length,
$D_{\rm VI},K_{\rm VI},B_{\rm VI}$, and $1/\eta$.
\end{lemma}
\begin{proof}
All norms in this proof are Euclidean. Let $\mathbf x_0$ be the
given rational point in $\mathcal P$.

\emph{Regularization and the target accuracy.}
Set
\[
 \begin{gathered}
 \lambda_{\rm reg}=\frac{\eta}{4D_{\rm VI}^2},\qquad
 \mathcal F_{\rm reg}(\mathbf x)
   =\mathcal F(\mathbf x)+\lambda_{\rm reg}(\mathbf x-\mathbf x_0),\\
 \gamma_{\rm VI}=\frac{\lambda_{\rm reg}}
                         {(K_{\rm VI}+\lambda_{\rm reg})^2},\qquad
 \rho_{\rm VI}=\frac{\lambda_{\rm reg}^2}
                       {2(K_{\rm VI}+\lambda_{\rm reg})^2},\qquad
 \vartheta_{\rm VI}=\frac{\eta}
                       {4(K_{\rm VI}D_{\rm VI}+B_{\rm VI}+1)}.
 \end{gathered}
\]
The operator $\mathcal F_{\rm reg}$ is
$\lambda_{\rm reg}$-strongly monotone and
$(K_{\rm VI}+\lambda_{\rm reg})$-Lipschitz.
Define
$T_{\rm VI}(\mathbf x)=\Pi_{\mathcal P}
(\mathbf x-\gamma_{\rm VI}\mathcal F_{\rm reg}(\mathbf x))$.
Nonexpansiveness of projection, strong monotonicity, and Lipschitz
continuity give, for $\mathbf x,\mathbf y\in\mathcal P$,
\[
 \begin{aligned}
 \|T_{\rm VI}(\mathbf x)-T_{\rm VI}(\mathbf y)\|^2
 &\le\bigl(1-2\gamma_{\rm VI}\lambda_{\rm reg}
       +\gamma_{\rm VI}^2(K_{\rm VI}+\lambda_{\rm reg})^2\bigr)
          \|\mathbf x-\mathbf y\|^2\\
 &=(1-2\rho_{\rm VI})\|\mathbf x-\mathbf y\|^2.
 \end{aligned}
\]
Thus $T_{\rm VI}$ is a contraction with factor at most
$1-\rho_{\rm VI}$. Its unique fixed point $\mathbf x_{\rm reg}$
solves the VI for $\mathcal F_{\rm reg}$, by the characterization
of Euclidean projection \cite{FacchineiPang2003}.
For the original operator, measure the largest violation of its
variational inequality by
\[
 R(\mathbf x)=\max_{\mathbf y\in\mathcal P}
       \langle\mathcal F(\mathbf x),\mathbf x-\mathbf y\rangle.
\]
Removing the regularization term from the VI gives
$R(\mathbf x_{\rm reg})\le\lambda_{\rm reg}D_{\rm VI}^2=\eta/4$.
For $\mathbf x,\mathbf y,\mathbf z\in\mathcal P$,
\[
 \begin{aligned}
 &\langle\mathcal F(\mathbf x),\mathbf x-\mathbf z\rangle
  -\langle\mathcal F(\mathbf y),\mathbf y-\mathbf z\rangle\\
 &\quad=\langle\mathcal F(\mathbf x)-\mathcal F(\mathbf y),
                    \mathbf x-\mathbf z\rangle
        +\langle\mathcal F(\mathbf y),\mathbf x-\mathbf y\rangle\\
 &\quad\le(K_{\rm VI}D_{\rm VI}+B_{\rm VI})\|\mathbf x-\mathbf y\|.
 \end{aligned}
\]
Taking maxima over $\mathbf z$ and exchanging $\mathbf x,\mathbf y$
therefore gives
$|R(\mathbf x)-R(\mathbf y)|\le
(K_{\rm VI}D_{\rm VI}+B_{\rm VI})\|\mathbf x-\mathbf y\|$.
Consequently, any $\mathbf x\in\mathcal P$ with
$\|\mathbf x-\mathbf x_{\rm reg}\|\le\vartheta_{\rm VI}$
satisfies $R(\mathbf x)\le\eta/2\le\eta$.

\emph{Approximate contraction iteration.}
Starting from $\mathbf x_0$, evaluate $\mathcal F_{\rm reg}$
with Euclidean error at most
$\rho_{\rm VI}\vartheta_{\rm VI}/(4\gamma_{\rm VI})$.
Use \Cref{lem:rational-projection} to return a point in
$\mathcal P$ within Euclidean distance
$\rho_{\rm VI}\vartheta_{\rm VI}/4$ of the exact projection
of the resulting rational point.
Nonexpansiveness of projection bounds the total error of one such
step relative to $T_{\rm VI}$ by
$\rho_{\rm VI}\vartheta_{\rm VI}/2$. The computed iterates obey
\[
 \begin{aligned}
 \|\mathbf x_{k+1}-\mathbf x_{\rm reg}\|
 &\le(1-\rho_{\rm VI})\|\mathbf x_k-\mathbf x_{\rm reg}\|
                   +\rho_{\rm VI}\vartheta_{\rm VI}/2,\\
 \|\mathbf x_k-\mathbf x_{\rm reg}\|
 &\le e^{-k\rho_{\rm VI}}D_{\rm VI}+\vartheta_{\rm VI}/2.
 \end{aligned}
\]
An integer upper bound on
$\rho_{\rm VI}^{-1}\log(2D_{\rm VI}/\vartheta_{\rm VI})$
therefore gives enough iterations to reach the required distance.

\emph{Rationality and running time.}
By \Cref{lem:rational-projection}, every outer iterate has
polynomial bit length independently of earlier iterates.
Evaluating $\mathcal F_{\rm reg}$
at these rational points and forming the next projection input
also take polynomial time. Finally, the iteration counts, the
reciprocals of the error tolerances, and the requested accuracy
bits are polynomially bounded in the input length,
$D_{\rm VI},K_{\rm VI},B_{\rm VI}$, and $1/\eta$.
The resulting rational point has $R(\mathbf x)\le\eta$.
\end{proof}

\paragraph{Applying the Subroutine to the Coefficient Map.}\label{sec:vi-allocation}
For our coefficient map, \Cref{prop:structure} supplies monotonicity
and \Cref{lem:field-lipschitz,lem:field-evaluation} supply the
continuity bounds and coefficient evaluations. Linear minimization over $\mathcal P_\kappa$ is an
assignment problem. The following proposition records the resulting
allocation guarantee. Here $\nu$ is a lower bound on the gap between
utility one and any smaller input utility.

Fix rational defining data $s,T$ and the affine pieces of
$(1-t)f(t)$ satisfying \Cref{def:field}. Fix an upper bound $E_{\max}\ge T^{1-\beta}$ and an integer $k_{\rm mass}\ge1/\alpha$.
Let $v$ be a normalized rational input utility matrix whose entries below
one are at most $1-\nu$, where $0<\nu\le1$.

\begin{proposition}\label{prop:computed-vi}
For rational $\kappa$ and $\eta_\kappa$ satisfying
$0<\kappa<1/(2n)$ and $0<\eta_\kappa\le1$, one can compute
a rational allocation $X\in\B_n$ with the following properties.
\begin{enumerate}[label=(\roman*),leftmargin=*]
\item Every entry is at least $\kappa$, and
\begin{equation}\label{eq:restricted-vi}
 \max_{\substack{Y\in\B_n\\Y_{ij}\ge\kappa}}
       \langle G(X),Y-X\rangle\le\eta_\kappa.
\end{equation}
\item Its coefficients and maximum over $Y\in\B_n$ satisfy
\begin{align}
 0\le G_{ij}(X)&\le\frac{n+\eta_\kappa}{1-n\kappa}\le4n,
                    \label{eq:computed-entry-bound}\\
 \max_{Y\in\B_n}\langle G(X),Y-X\rangle
   &\le\frac{\eta_\kappa+n^2\kappa}{1-n\kappa}.
                    \label{eq:computed-full-gap}
\end{align}
\item Every nonconstant row assigns mass at least
\begin{equation}\label{eq:computed-top-mass}
 \left(\frac{\nu}{4nE_{\max}}\right)^{k_{\rm mass}}
\end{equation}
to utility-one goods.
\end{enumerate}
The running time is polynomial in the input length, $n$, $1/\kappa$,
and $1/\eta_\kappa$, with $d,f$ fixed.
\end{proposition}
\begin{proof}[Proof of \Cref{prop:computed-vi}]
Fix the rational data and parameters in the statement.

\emph{Verifying the subroutine's inputs.}
Use the
polytope $\mathcal P_\kappa$ defined in \Cref{sec:exact-vi}; its affine
representation shows that it is a nonempty compact rational polytope
with known rational point $J_{\rm all}/n$.
Linear minimization reduces to an assignment problem and returns the
affine image of a permutation matrix, with bit length polynomial in
$n$ and the description length of $\kappa$.
Every nonconstant row has a utility-one good and assigns at least
$\kappa$ mass to such goods throughout $\mathcal P_\kappa$, so
\Cref{lem:field-lipschitz} applies with $m=\kappa$.

Viewing matrices as vectors in $\mathbb R^{n^2}$, we may take the
following Euclidean diameter, Lipschitz, and norm bounds, enlarging
them to at least one if necessary:
\[
 D_{\rm VI}=2n,\qquad K_{\rm VI}=nK_\kappa,\qquad
 B_{\rm VI}=n\bigl(\kappa^{-1}+f(1)\bigr).
\]
The diameter bound follows from the total mass $n$ of an allocation.
The row estimate in \Cref{lem:field-lipschitz}, together with
$\|\mathbf z\|_1\le\sqrt n\|\mathbf z\|_2$ and summation of the
squared row norms, gives the matrix Lipschitz bound $nK_\kappa$;
the $n^2$ entry bounds give $B_{\rm VI}$.
\Cref{prop:structure} supplies monotonicity of $\mathcal F=-G$,
and \Cref{lem:field-evaluation} supplies evaluations with
error bounds. Euclidean evaluation error $\epsilon_{\rm eval}$
requires only entrywise error $\epsilon_{\rm eval}/n$, so the
requested accuracy still has polynomial bit length.

\emph{Allocation and assignment bounds.}
Apply \Cref{lem:rational-vi} on $\mathcal P_\kappa$ with
$\mathcal F=-G$ and positive tolerance $\eta_\kappa$ to obtain a
rational $X$ satisfying \Cref{eq:restricted-vi}, proving
\Cref{prop:computed-vi}(i).
Applying \Cref{lem:restricted-vi-estimates} to this output gives
\Cref{eq:computed-entry-bound,eq:computed-full-gap} and
\Cref{lem:best-good-mass} gives \Cref{eq:computed-top-mass}.

\emph{Running time.}
For fixed $d,f$, the bound $K_\kappa=O(\kappa^{-2})$, the other
subroutine bounds, the assignment oracle, and the coefficient evaluations
satisfy all polynomial bounds required by \Cref{lem:rational-vi}.
The running time and output bit length are therefore polynomial in
the input length, $n$, $1/\kappa$, and $1/\eta_\kappa$.
\end{proof}
The following choice of $\kappa$ and $\eta_\kappa$ gives a
prescribed tolerance over all feasible allocations. It is the
version used by the final algorithm.

\begin{corollary}\label{cor:computed-full-vi}
Under the hypotheses of \Cref{prop:computed-vi},
for rational $0<\eta\le1$, choose
\[
 \kappa=\frac{\eta}{2n(n+\eta)},\qquad \eta_\kappa=\frac\eta2.
\]
With these parameters, \Cref{eq:approx-vi} holds with tolerance $\eta$ and
$0\le G_{ij}(X)\le n+\eta$. The computed allocation also satisfies,
for every nonconstant row $i$,
\begin{equation}\label{eq:full-top-mass}
 \sum_{j:v_{ij}=1}X_{ij}\ge
 \left(\frac{\nu}{(n+\eta)E_{\max}}\right)^{k_{\rm mass}}.
\end{equation}
The running time is polynomial in the input length and $1/\eta$.
\end{corollary}
\begin{proof}
For $\kappa=\eta/[2n(n+\eta)]$ and $\eta_\kappa=\eta/2$,
we have $0<\kappa<1/(2n)$ and $0<\eta_\kappa\le1$.
Substitution into \Cref{lem:restricted-vi-estimates} gives
\[
 \frac{\eta/2+n^2\kappa}{1-n\kappa}=\eta,
 \qquad
 M_\kappa=\frac{n+\eta/2}{1-n\kappa}=n+\eta.
\]
Thus \Cref{eq:approx-vi} holds with tolerance $\eta$ and every coefficient
is at most $n+\eta$.
Applying \Cref{lem:best-good-mass} with this
value of $M_\kappa$ gives \Cref{eq:full-top-mass}.
Since $\eta\le1$, the denominator in that mass bound may also be
replaced by $(n+1)E_{\max}$ when selecting a lower bound before
computation.

Both $1/\kappa$ and $1/\eta_\kappa$ are polynomial in $n,1/\eta$,
so the running-time bound follows from \Cref{prop:computed-vi}.
\end{proof}

\begin{proof}[Proof of \Cref{prop:allocation-step}]
Use the rational defining data of \Cref{def:field} to specify the
coefficient family. Part (i) follows from
\Cref{lem:normalizer,lem:coefficient-shape,prop:structure},
and part (ii) follows from \Cref{cor:exact-vi}.
For part (iii), fix rational utilities and rational $0<\eta\le1$.
Let $\nu$ be the minimum positive value of $1-v_{ij}$ among entries
below one, or let $\nu=1$ if there are no such entries.
Apply \Cref{cor:computed-full-vi}.
The corollary gives a rational allocation in the domain of $G$
satisfying \Cref{eq:approx-vi} and with coefficients at most
$n+\eta$. Its running time is polynomial in the input length and $1/\eta$.
\end{proof}

\subsection{Prices from Assignment Duality}\label{sec:conversion}\label{sec:price-part}\label{app:ch3-prices}
The conversion in \Cref{prop:price-step} has two ingredients:
assignment duality controls spending, and coefficient separation
controls utility. We prove these two statements in that order and
then combine them to prove \Cref{prop:price-step}. Throughout this
subsection, $\mathbf q$ denotes the common assignment-dual prices,
whereas $\mathbf q^{\rm row}(\mathbf v,\mathbf x)$ denotes auxiliary
prices defined separately for a single row and bundle.

\subsubsection{Assignment Prices and Spending}\label{sec:assignment}
For $X\in\B_n$ with $G(X)$ defined, the assignment dual for the
fixed weights $G(X)$ is
\[
 \min_{\mathbf q,\boldsymbol\mu\in\mathbb R^n}
 \left\{\sum_jq_j+\sum_i\mu_i:
 q_j+\mu_i\ge G_{ij}(X)\quad(i,j\in[n])\right\}.
\]
The variables $\mathbf q$ correspond to goods and
$\boldsymbol\mu$ to agents; neither is initially constrained in sign.

For a normalized nonconstant utility row $\mathbf v$ with coefficient
function $g$ and a bundle $\mathbf x$ in its domain, define the auxiliary
row prices
\[
 q_j^{\rm row}(\mathbf v,\mathbf x)=
 \left(g(v_j;\mathbf x)-\min_{k:x_k>0}g(v_k;\mathbf x)\right)_+
 \qquad(j\in[n]).
\]
These prices depend on the row and its bundle.

The next lemma gives common prices with spending at most one and
compares them with the row prices. This comparison will let us turn
coefficient separation into the utility guarantee.

\begin{lemma}\label{lem:assignment-dual}
If $X$ satisfies \Cref{eq:vi}, an optimal assignment dual exists.
Every optimal assignment dual $(\mathbf q,\boldsymbol\mu)$, after
normalization by $\min_jq_j=0$, satisfies
\begin{equation}\label{eq:dual}
 q_j+\mu_i\ge G_{ij}(X),\qquad
 X_{ij}>0\Longrightarrow q_j+\mu_i=G_{ij}(X),\qquad
 \min_jq_j=0.
\end{equation}
If $G(X)\ge0$, every nonconstant row satisfies
$\sum_jX_{ij}G_{ij}(X)=1$, and constant rows have zero coefficients,
then $\mathbf q\ge0$, $\boldsymbol\mu\ge0$, and
$\mathbf q\cdot\mathbf X_i\le1$ for every agent.
For every nonconstant row $i$ and every good $j$,
\[
 q_j\ge q_j^{\rm row}(\mathbf v_i,\mathbf X_i).
\]
Constant rows have zero payment.
\end{lemma}

\begin{proof}
Hold the coefficient matrix $G(X)$ fixed. By \Cref{eq:vi}, $X$ is
optimal for the assignment problem with these weights. The feasible
set $\B_n$ is nonempty and compact, so linear programming duality
gives an optimal dual. We establish the three claims in turn.

\emph{Normalization and complementary slackness.}
Start with any optimal dual $(\mathbf q^0,\boldsymbol\mu^0)$ and
set $c=\min_jq_j^0$. Replacing this pair by
$\mathbf q=\mathbf q^0-c\mathbf1$ and
$\boldsymbol\mu=\boldsymbol\mu^0+c\mathbf1$ preserves every
constraint and the objective value, because both vectors have $n$
entries. The new pair is therefore optimal and satisfies
$\min_jq_j=0$. Complementary slackness with the optimal allocation
$X$ gives $q_j+\mu_i=G_{ij}(X)$ whenever $X_{ij}>0$.
This proves \Cref{eq:dual} for any normalized optimal dual.

\emph{Nonnegativity and spending.}
Now assume the stated coefficient properties. Normalization gives
$\mathbf q\ge0$. Choose a zero-price good $j_0$; its dual constraint
gives $\mu_i\ge G_{ij_0}(X)\ge0$ for every agent $i$.
For each row $i$, average the complementary equalities over the
assigned bundle $\mathbf X_i$. Since $\sum_jX_{ij}=1$, this yields
\[
 \mathbf q\cdot\mathbf X_i
   =\sum_jX_{ij}G_{ij}(X)-\mu_i.
\]
For a nonconstant row, the assigned coefficient mean is one, so
$\mathbf q\cdot\mathbf X_i=1-\mu_i\le1$.
For a constant row, the mean is zero, so
$\mathbf q\cdot\mathbf X_i=-\mu_i\le0$.
Nonnegative prices and allocations make this payment nonnegative;
hence every constant row has zero payment.

\emph{Lower bounds from the auxiliary row prices.}
Fix a nonconstant row $i$ and write
$a_i=\min_{k:X_{ik}>0}G_{ik}(X)$.
A supported good attaining this minimum satisfies
$\mu_i=a_i-q_k\le a_i$ by complementary slackness.
Thus dual feasibility gives $q_j\ge G_{ij}(X)-\mu_i\ge G_{ij}(X)-a_i$
for every good $j$. Combining this inequality with $q_j\ge0$, we get
\[
 q_j\ge\bigl(G_{ij}(X)-a_i\bigr)_+
      =q_j^{\rm row}(\mathbf v_i,\mathbf X_i).
\]
\end{proof}

\subsubsection{The Single-Agent Utility Guarantee}\label{sec:price-conversion}
The pointwise price bound makes every bundle affordable at the scaled
dual prices affordable at the corresponding scaled row prices.
The following single-agent comparison supplies the utility guarantee
and will also be used in \Cref{sec:rational-output}.

\begin{theorem}\label{thm:conversion}
Fix $0<\epsilon<1$ and a normalized nonconstant utility row
$\mathbf v$ with coefficient function $g$. Suppose that, for every
bundle $\mathbf x$ in its domain, $g(\cdot;\mathbf x)$ is nonnegative,
nondecreasing, and convex, has assigned mean one, and satisfies
\Cref{eq:target-separation}. Then, for every such $\mathbf x$, every
$\mathbf y\in\Delta_n$ satisfying
\begin{equation}\label{eq:row-affordability}
 (1+\epsilon)\mathbf q^{\rm row}(\mathbf v,\mathbf x)
                 \cdot\mathbf y\le1
\end{equation}
satisfies $\mathbf v\cdot(\mathbf y-\mathbf x)\le\epsilon$.
\end{theorem}

\begin{proof}
Fix the utility row $\mathbf v$ and the assigned bundle $\mathbf x$;
all coefficient evaluations in this proof use this same $\mathbf x$.
Write
\[
 b=\min\{v_j:x_j>0\},\qquad
 t=\mathbf v\cdot\mathbf x+\epsilon.
\]
Our goal is to show that any alternative unit bundle with utility
strictly greater than $t$ costs more than one at the scaled row prices.
If $t\ge1$, there is no such bundle, because every normalized utility
is at most one. Hence suppose $b\le\mathbf v\cdot\mathbf x<t<1$.

\emph{A scalar price function.}
Since $g(\cdot;\mathbf x)$ is nondecreasing, the smallest coefficient
on the support of $\mathbf x$ is $g(b;\mathbf x)$. Define
\[
 \pi(z)=\bigl(g(z;\mathbf x)-g(b;\mathbf x)\bigr)_+,
 \qquad 0\le z\le1.
\]
The function $\pi$ is convex: it is the maximum of the convex function
$g(\cdot;\mathbf x)-g(b;\mathbf x)$ and zero. It satisfies
$\pi(b)=0$ and $\pi(v_j)=q_j^{\rm row}(\mathbf v,\mathbf x)$ for every good
$j$. The separation assumption \Cref{eq:target-separation} gives
$\pi(t)\ge1/(1+\epsilon)$. For any $w>t$, convexity applied to
$t\in(b,w)$ gives
\[
 \pi(t)\le\frac{t-b}{w-b}\pi(w),
 \qquad\text{and hence}\qquad
 \pi(w)\ge\frac{w-b}{t-b}\pi(t)>\frac1{1+\epsilon}.
\]
The strict inequality follows from $w>t>b$ and $\pi(t)>0$.

\emph{From utility to expenditure.}
Consider an alternative bundle $\mathbf y\in\Delta_n$ with
$w=\mathbf v\cdot\mathbf y>t$. Jensen's inequality for $\pi$ gives
\[
 \mathbf q^{\rm row}(\mathbf v,\mathbf x)\cdot\mathbf y
   =\sum_jy_j\pi(v_j)
   \ge \pi\!\left(\sum_jy_jv_j\right)
   =\pi(w)>\frac1{1+\epsilon}.
\]
Thus $\mathbf y$ costs more than one at
$(1+\epsilon)\mathbf q^{\rm row}(\mathbf v,\mathbf x)$.
Every bundle satisfying \Cref{eq:row-affordability} must therefore
have utility at most $t=\mathbf v\cdot\mathbf x+\epsilon$, as required.
\end{proof}

\paragraph{The Resulting Equilibrium.}
We now combine the two results to obtain common equilibrium prices.

\begin{proof}[Proof of \Cref{prop:price-step}]
Let $X$ be the exact VI allocation in the proposition, and let
$(\mathbf q,\boldsymbol\mu)$ be the given normalized optimal
assignment dual. The coefficient assumptions are those
needed for \Cref{lem:assignment-dual}. Consequently,
$\mathbf q\ge0$, $\min_jq_j=0$, and
$\mathbf q\cdot\mathbf X_i\le1$ for every agent $i$.
Set $\mathbf p=(1+\epsilon)\mathbf q$. Then
$\min_jp_j=0$ and $\mathbf p\cdot\mathbf X_i\le1+\epsilon$,
which proves the required spending bound.

To check utility, fix a nonconstant row $i$ and any unit bundle
$\mathbf y$ affordable at $\mathbf p$ with budget one.
The row-price domination in \Cref{lem:assignment-dual} implies
\[
 (1+\epsilon)\mathbf q^{\rm row}(\mathbf v_i,\mathbf X_i)
       \cdot\mathbf y\le\mathbf p\cdot\mathbf y\le1.
\]
Apply \Cref{thm:conversion} to the row $\mathbf v_i$ and assigned
bundle $\mathbf X_i$: every such $\mathbf y$ has utility at most
$\mathbf v_i\cdot\mathbf X_i+\epsilon$.
Taking the maximum over affordable $\mathbf y$ gives
$\OPT_i(\mathbf p)-\mathbf v_i\cdot\mathbf X_i\le\epsilon$.
For a constant utility row, every unit bundle has the same utility,
so its utility loss is zero. Finally, $X\in\B_n$ gives exact unit
demand and market clearing. These are all the conditions in
\Cref{def:approx-hz}.
\end{proof}

\subsection{A Scalar Criterion for the Utility Guarantee}\label{sec:demand}
The price conversion requires coefficient separation for every row and
bundle. This subsection reduces that requirement to nonnegativity of a
function of three real variables, independently of the number of goods.
We then give two fixed coefficient choices that pass this test.
The 144-knot choice supplies \Cref{prop:demand-step}; the ten-knot
choice gives a slightly larger approximation error with fewer data.

\subsubsection{From a Row's Lottery to Three Variables}\label{sec:scalar-demand-test}
Fix a normalized nonconstant row $\mathbf v$ and a bundle $\mathbf x$
in the coefficient domain. Recall that
$L(a)=\sum_{j:v_j\ge a}x_j$ is its utility tail. For a target error
$0<\epsilon<1$, the relevant comparison is at utility
$t=\mathbf v\cdot\mathbf x+\epsilon$. If $t\le1$, put
\[
 b=\min\{v_j:x_j>0\},\qquad
 \ell=\frac{1}{t-b}\int_b^tL(a)\,da,\qquad
 u=\begin{cases}
 \displaystyle\frac{1}{1-t}\int_t^1L(a)\,da,&t<1,\\[3pt]
 0,&t=1.
 \end{cases}
\]
Thus $b$ is the lowest allocated utility, while $\ell$ and $u$ are
average tail masses between $b$ and $t$ and above $t$, respectively;
$u=0$ is the endpoint convention at $t=1$. Every such row gives a
triple in the domain below. The mean-utility identity
\[
 t-\epsilon=b+(t-b)\ell+(1-t)u
\]
determines $\ell$ from $b,t,u$, giving \Cref{eq:middle-average}.
Thus only three scalar variables are needed for the test.
\Cref{lemma:demand-geometry} proves these identities and domain bounds.

The function $S(b,t,u)$ defined below combines a lower bound on the
integral contribution with a lower bound on the correction term
in $g(t;\mathbf x)-g(b;\mathbf x)$.
Its definition depends only on $b,t,u$, the fixed coefficients, and
$\epsilon$.

\begin{definition}[Scalar Demand Test]\label{def:scalar-test}
Fix $0<\epsilon<1$ and coefficients from \Cref{def:field}.
Let $\Omega_\epsilon$ be the set of triples $(b,t,u)$ satisfying
\begin{equation}\label{eq:demand-domain}
 0\le b<1-\epsilon,\qquad b+\epsilon<t\le1,\qquad
 0\le u\le\frac{t-b-\epsilon}{1-b}.
\end{equation}
For $(b,t,u)\in\Omega_\epsilon$, define
\begin{equation}\label{eq:middle-average}
 \ell=1-\frac{\epsilon+(1-t)u}{t-b}.
\end{equation}
The domain constraints imply $0<\ell<1$ and $0\le u\le\ell$.
Define the auxiliary function
\begin{equation}\label{eq:tail-majorant}
 \overline N(z)=
 \begin{cases}
 T^{\alpha-\beta}z^{1-\alpha},&0\le z\le T,\\[2pt]
 \displaystyle
 \frac{(1-z)T^{1-\beta}+z-T}{1-T},&T\le z\le1.
 \end{cases}
\end{equation}
At each such triple, put
\begin{equation}\label{eq:demand-xi}
 \xi=\frac{1}{1+\epsilon}-(t-b)f(b)
       +(1-t)u\bigl(f(t)-f(b)\bigr).
\end{equation}
For this $\xi$, set $\Phi_\xi(z)=d(z)-\xi z d(z)$, and let
$\operatorname{co}\Phi_\xi$ denote the largest convex function below
$\Phi_\xi$ on $(0,1]$. Define
\begin{equation}\label{eq:demand-residual}
 S(b,t,u)=(t-b)\operatorname{co}\Phi_\xi(\ell)
          -(1-t)\xi\overline N(u)-b\xi.
\end{equation}
\end{definition}

The next proposition turns $S(b,t,u)\ge0$ throughout
$\Omega_\epsilon$ into the uniform coefficient separation required by
\Cref{prop:demand-step}.

\begin{proposition}\label{prop:row-demand}\label{prop:demand}
Fix $0<\epsilon<1$ and coefficients from \Cref{def:field}.
Suppose $(1-z)f(z)\le H<1/(1+\epsilon)$ for all $z\in[0,1]$, and
$S(b,t,u)\ge0$ throughout the domain in \Cref{eq:demand-domain}.
For every normalized nonconstant row $\mathbf v$ and
$\mathbf x\in\Delta_n$ with $\sum_{j:v_j=1}x_j>0$, if
$\mathbf v\cdot\mathbf x+\epsilon<1$, then
\[
 g(\mathbf{v}\cdot \mathbf{x}+\epsilon;\mathbf{x})-\min_{j:x_j>0}g(v_j;\mathbf{x})\ge\frac1{1+\epsilon}.
\]
\end{proposition}

The criterion is uniform over the scalar domain, so it applies to every
row and bundle without enumerating their utility values.
The reduction and its supporting lemmas are proved in
\Cref{sec:demand-geometry}.

\subsubsection{Coefficient Choices}\label{sec:certified-choices}
The same weight function $d$ works with two choices of the correction function $f$.
The complete data are in \Cref{app:rational-data}; the coefficient
bounds and lower bounds on $S$ needed here are stated below.

\begin{definition}[Two Choices of the Coefficients]\label{def:numerical-field}
All decimal constants in this definition are exact rationals. Set
\[
 s=0.6395298,\qquad T=0.76700663,
\]
and define $d$ by \Cref{eq:tail-weight}, with $\alpha=s^2$ and
$\beta=(s+2)^2$. For the 144-knot choice, let $h$ be the piecewise
linear interpolation of the exact knots in \Cref{app:data}, paired
with $\epsilon=0.075521017$. For the ten-knot choice, interpolate the
knots in \Cref{app:sparse}, paired with $\epsilon=0.076$.
In either case, set $f(t)=h(t)/(1-t)$ for $t<1$ and take its finite
left limit at $t=1$. For both choices, use
\[
 A=0.24,\qquad H=0.149589,
\]
for the bounds in \Cref{eq:shape} and \Cref{prop:row-demand},
respectively. The coefficient formula is \Cref{eq:field}, and the
complete parameter and knot data are collected in
\Cref{app:rational-data}.
\end{definition}

The first result places these data in the required coefficient family
and gives the bound on $H$ needed by the scalar criterion.
The second gives a uniform positive lower bound on $S$ for each
choice at its corresponding target error.

\begin{proposition}\label{prop:fixed-structure}\label{prop:data}
Both choices in \Cref{def:numerical-field} define members of the
coefficient family in \Cref{def:field}. Each choice, with its
corresponding target $\epsilon$, satisfies \Cref{eq:shape} and
\[
 (1-t)f(t)\le H<\frac1{1+\epsilon},\qquad
 \left(\frac1{1+\epsilon}+H\right)(\beta-1)<\beta+1.
\]
One may take $E_{\max}=5$ and $k_{\rm mass}=3$ in
\Cref{lem:best-good-mass}.
\end{proposition}

Both choices of $f$ are constant on $[0.876,1]$, with the same value
in \Cref{eq:plateau}, approximately $0.972387$. The next proposition
gives the bound on $S$ needed by \Cref{prop:row-demand}.

\begin{proposition}\label{prop:fixed-certificate}\label{prop:numerical}
For each of the two paired choices of $f,\epsilon$ in
\Cref{def:numerical-field},
\begin{equation}\label{eq:positive-residual}
 S(b,t,u)\ge10^{-10}
\end{equation}
throughout the domain in \Cref{eq:demand-domain}.
\end{proposition}

The coefficient assumptions are proved in \Cref{sec:choices}.
The proof of the lower bound on $S$ is given in
\Cref{sec:certification}, using the interval estimates in
\Cref{app:artifact}.

\subsubsection{Resulting Separation and Approximation Guarantees}
\label{sec:certified-consequences}
The scalar criterion and the uniform lower bound on $S$ give
coefficient separation for both choices. The 144-knot choice in
\Cref{app:data} satisfies the criterion at $0.075521017$ and gives the output
error $0.0755211+\delta$ used in the main theorem. The ten-knot
choice in \Cref{app:sparse} satisfies the criterion at $0.076$ and gives output
error $0.076+\delta$.

\begin{corollary}\label{cor:row-demand}
For either choice in \Cref{def:numerical-field}, with its corresponding
target $\epsilon$, every normalized
nonconstant row $\mathbf v$ and every $\mathbf x\in\Delta_n$ assigning
positive mass to utility one satisfy \Cref{eq:target-separation}
whenever $\mathbf v\cdot\mathbf x+\epsilon<1$.
\end{corollary}

\begin{proof}[Proof of \Cref{prop:demand-step}]
Take the 144-knot choice in \Cref{def:numerical-field}, and write
$\epsilon_*=0.075521017$ and $\epsilon_0=0.0755211$.
\Cref{prop:fixed-structure} places these coefficients in the required
family. \Cref{cor:row-demand} gives separation at $\epsilon_*$.

Fix a normalized nonconstant row $\mathbf v$ and a bundle $\mathbf x$
in the coefficient domain with $\mathbf v\cdot\mathbf x+\epsilon_0<1$.
The smaller target $\mathbf v\cdot\mathbf x+\epsilon_*$ is also below
one. Put $a=\min_{j:x_j>0}g(v_j;\mathbf x)$.
Since $g(\cdot;\mathbf x)$ is nondecreasing by
\Cref{lem:coefficient-shape},
\[
 g(\mathbf v\cdot\mathbf x+\epsilon_0;\mathbf x)-a
 \ge g(\mathbf v\cdot\mathbf x+\epsilon_*;\mathbf x)-a
 \ge \frac1{1+\epsilon_*}>\frac1{1+\epsilon_0}.
\]
This is the required separation at $\epsilon_0$ for every such row
and bundle.
\end{proof}

The ten-knot choice gives the following alternative by
\Cref{prop:rational-equilibrium}.

\begin{corollary}\label{prop:sparse}
For every rational input utility matrix $v\in[0,1]^{n\times n}$ and
accuracy $\delta>0$,
the coefficients defined by the ten knots in \Cref{app:sparse} yield a
polynomial-time algorithm for a $(0.076+\delta)$-approximate HZ
equilibrium. Its allocation and prices
have polynomial encoding length, and unit demand and market clearing hold exactly.
\end{corollary}

The proofs of \Cref{cor:row-demand,prop:sparse} are in
\Cref{app:ch3-demand-consequences}.

\subsection{Rational Computation and the Error Budget}\label{sec:rational-output}
This subsection specifies the finite-precision construction in
\Cref{prop:rational-equilibrium}. Its three steps are input
rounding, computation of rational allocation and dual data, and
precision selection. The formulas below determine all tolerances
before the allocation is computed. Proofs of the three lemmas
and the underlying stability estimates are in \Cref{app:ch3-rational}.

Fix rational defining data for a member of \Cref{def:field} and a
rational target $0<\epsilon<1$ for which
\Cref{eq:target-separation} holds for every normalized nonconstant
row and every bundle in its coefficient domain.
The algorithm takes a rational allowance $\delta>0$; if $\delta>1$,
run it with allowance one. We therefore use $0<\delta\le1$ below.
An arbitrary requested positive tolerance can equivalently be
specified through an integer $N$ and the rational allowance $1/N$:
$N=\lceil\max\{1,1/\delta_{\rm req}\}\rceil$ suffices for a
request $\delta_{\rm req}>0$.

\paragraph{Input Normalization and Rounding.}
For an original nonconstant utility row $\mathbf u\in[0,1]^n$, set
$v_j=(u_j-\min_k u_k)/(\max_k u_k-\min_k u_k)$; replace a constant
row by zero. Choose
\[
 \nu=\frac1{\lceil4/\delta\rceil},\qquad
 \widehat v_{ij}=\nu\left\lfloor\frac{v_{ij}}\nu\right\rfloor.
\]
All subsequent coefficient, VI, and assignment-dual computations use
$\widehat v$. The same output allocation and prices are evaluated
under the rounded, normalized, and original utilities.

The next lemma transfers the utility guarantee from the rounded
market back to the original market, adding at most $\nu$ to the error.

\begin{lemma}\label{lem:utility-rounding}
Let $\mathbf{v}$ be a normalized nonconstant utility row. If $1/\nu$ is a positive
integer, round $\mathbf{v}$ down to multiples of $\nu$ and denote the
result by $\widehat{\mathbf{v}}$. Then zero and one are preserved, the gap below
utility one is at least $\nu$, and for every $\mathbf x\in\Delta_n$
and nonnegative price vector $\mathbf p$ with $\min_jp_j=0$,
\[
 \OPT_{\mathbf{v}}(\mathbf{p})-\mathbf{v}\cdot \mathbf{x}
 \le \OPT_{\widehat{\mathbf{v}}}(\mathbf{p})-\widehat{\mathbf{v}}\cdot \mathbf{x}+\nu.
\]
Any nonnegative upper bound on utility loss for a normalized row also holds
for its original utilities in $[0,1]$. Constant rows may be replaced
by the zero row without changing their utility loss.
\end{lemma}

\paragraph{Rational Allocation and Dual Prices.}
Fix a rational bound $E_{\max}\ge T^{1-\beta}$ and an integer
$k_{\rm mass}\ge1/\alpha$, independently of the input market.
For a utility gap $0<\nu\le1$, define
\begin{equation}\label{eq:rational-computation-bounds}
 m=\left(\frac{\nu}{(n+1)E_{\max}}\right)^{k_{\rm mass}},
 \qquad K=K_{m/2},
\end{equation}
where $K_{m/2}$ is given by \Cref{lem:field-lipschitz}.
The lower bound $m$ on mass at utility one and the Lipschitz bound
$K$ depend only on $n$, the input gap, and the fixed coefficient data.
The following computation is parameterized by tolerances $\eta,\sigma$;
the next paragraph specifies them before it is invoked.

The next lemma supplies the allocation, approximate coefficients,
and dual prices used by the final algorithm, with the error bounds
needed to choose its precision.

\begin{lemma}\label{lem:rational-computation}
Let $v$ be a rational
normalized utility matrix whose entries below one are at most $1-\nu$,
where $0<\nu\le1$ is rational. For rational $0<\eta\le1$ and
$\sigma>0$, one can compute a rational allocation $X\in\B_n$ in the
domain of $G$, a nonnegative rational matrix $\widetilde G$, and a
rational optimal assignment dual $(\mathbf q,\boldsymbol\mu)$ for
$\widetilde G$ such that
\[
 \max_{Y\in\B_n}\langle G(X),Y-X\rangle\le\eta,
 \qquad |\widetilde G_{ij}-G_{ij}(X)|\le\sigma,
 \qquad \min_jq_j=0.
\]
Each nonconstant row of $X$ assigns mass at least $m$ from
\Cref{eq:rational-computation-bounds} to utility-one goods.
The running time and output encoding length are polynomial in the
input length, $1/\nu$, $1/\eta$, and $\log(1+1/\sigma)$, with the
coefficient data fixed.
\end{lemma}

Because the coefficient data are fixed, the formulas in
\Cref{eq:rational-computation-bounds} give
$m^{-1},K\le\operatorname{poly}(n,1/\nu)$.

\paragraph{Precision Selection and Output Prices.}
Use the following tolerances, where the first line records the
input-dependent bounds already defined above:
\begin{equation}\label{eq:overview-tolerances}
 \begin{gathered}
 \nu=\frac1{\lceil4/\delta\rceil},\qquad
 m=\left(\frac{\nu}{(n+1)E_{\max}}\right)^{k_{\rm mass}},
 \qquad K=K_{m/2},\\
 r=\min\left\{\frac m2,
             \frac{\delta}{4(1+4(1+\epsilon)K)}\right\},\qquad
 \eta=\frac{r\delta}{8(1+\epsilon)},\\
 \sigma=\min\left\{\frac{\delta}{8(1+\epsilon)},
                         \frac{r\delta}{16n(1+\epsilon)}\right\}.
 \end{gathered}
\end{equation}
For both numerical choices in \Cref{def:numerical-field}, take
$E_{\max}=5$ and $k_{\rm mass}=3$, as guaranteed by
\Cref{prop:fixed-structure}. The main algorithm uses
$\epsilon=0.0755211$; the ten-knot version uses $\epsilon=0.076$.

The next lemma turns these tolerances into spending and utility
bounds. Together with the rounding estimate, it will give error
$\epsilon+\delta$. More generally, fix rational $0<m\le1$ and
$K\ge K_{m/2}$, and use the last two lines of
\Cref{eq:overview-tolerances} with these bounds and $0<\delta\le1$.
As throughout this subsection, the fixed coefficients and rational
target $\epsilon$ satisfy \Cref{eq:target-separation}.

\begin{lemma}\label{lem:implementation}
The tolerances just specified satisfy $0<\eta,\sigma<1$ and
$1/\eta,1/\sigma\le\operatorname{poly}(n,1/\delta,1/m,K)$.
For every normalized rational $n\times n$ utility matrix, suppose that
a rational $X\in\B_n$ satisfies \Cref{eq:approx-vi} with this $\eta$
and gives each nonconstant row mass at least $m$ on utility-one goods, and
that a nonnegative rational $\widetilde G$ satisfies
$|\widetilde G_{ij}-G_{ij}(X)|\le\sigma$.
For any rational optimal assignment dual for $\widetilde G$ with
$\min_jq_j=0$, the prices $\mathbf p=(1+\epsilon)\mathbf q$ satisfy
\[
 \mathbf p\cdot\mathbf X_i\le1+\epsilon+\delta/4,\qquad
 \OPT_i(\mathbf p)-\mathbf v_i\cdot\mathbf X_i
       \le\epsilon+3\delta/4
 \quad(i\in[n]).
\]
The precision choice and price computation take time polynomial in
the input encoding length, $n,1/\delta,1/m,K$, with $\epsilon$ fixed.
\end{lemma}

\paragraph{The Complete Error Budget.}
We now combine the three steps to prove the implementation
claim from the overview.

\begin{proof}[Proof of \Cref{prop:rational-equilibrium}]
Fix the choice of coefficients and the rational target $\epsilon$ in the
proposition. As explained at the start of this subsection, it suffices
to handle a rational allowance $0<\delta\le1$; for a requested
$\delta_{\rm req}>0$, use $\delta=1/N$ with
$N=\lceil\max\{1,1/\delta_{\rm req}\}\rceil$.
Write $u$ for the original rational utility matrix, $v$ for its
row-normalized version, and $\widehat v$ for the rounded version.

\emph{Normalize, round, and compute.}
Normalize every nonconstant row of $u$ and replace constant rows by
zero. Set $\nu=1/\lceil4/\delta\rceil$ and round $v$ down to multiples
of $\nu$. By \Cref{lem:utility-rounding}, $\widehat v$ is normalized
and has gap at least $\nu$ below one. Determine $m,K$ from
\Cref{eq:rational-computation-bounds}, then select $r,\eta,\sigma$
by \Cref{eq:overview-tolerances}. All these choices precede the
allocation computation.

Apply \Cref{lem:rational-computation} to $\widehat v$ with tolerances
$\eta,\sigma$. Its rational output $X,\widetilde G,\mathbf q,\boldsymbol\mu$
satisfies the hypotheses of \Cref{lem:implementation}, including the
mass bound $m$ and the normalized optimal dual for $\widetilde G$.
Therefore the rational prices $\mathbf p=(1+\epsilon)\mathbf q$ have
minimum zero, spending at most $1+\epsilon+\delta/4$, and utility
loss for every rounded row at most $\epsilon+3\delta/4$.

\emph{Restore the original utilities.}
Keep the same allocation $X$ and prices $\mathbf p$.
For every nonconstant row $i$, \Cref{lem:utility-rounding} first gives
\[
 \OPT_{\mathbf v_i}(\mathbf p)-\mathbf v_i\cdot\mathbf X_i
 \le \OPT_{\widehat{\mathbf v}_i}(\mathbf p)
       -\widehat{\mathbf v}_i\cdot\mathbf X_i+\nu
 \le\epsilon+3\delta/4+\nu\le\epsilon+\delta.
\]
The same lemma transfers this nonnegative bound from $\mathbf v_i$
to the original row $\mathbf u_i$. Constant rows have zero utility
loss. Spending is unchanged because neither $X$ nor $\mathbf p$
has changed, and $X\in\B_n$ gives exact unit demand and market
clearing. Thus $(X,\mathbf p)$ is the required approximate equilibrium
for the original input $u$.

\emph{Running time and encoding length.}
The rounding grid satisfies $\nu^{-1}=O(1/\delta)$.
\Cref{lem:rational-computation} bounds $m^{-1},K$ polynomially in
$n,1/\delta$, and \Cref{lem:implementation} then bounds the reciprocal
tolerances polynomially in the same quantities. Applying the running-time
and encoding-length guarantees of these two lemmas, together with
rational normalization and rounding, proves the claimed polynomial
bounds. Since $N=O(1+1/\delta_{\rm req})$, these bounds also hold for
the original requested tolerance.
\end{proof}

At $\epsilon=0.0755211$, the complete procedure is given below.

\subsubsection{The Complete Algorithm}\label{sec:complete-algorithm}
\begin{algorithm}[!tbp]
\caption{Approximate HZ Equilibrium}\label{alg:hz}
\begin{algorithmic}[1]
\Require Original rational utilities $u\in[0,1]^{n\times n}$;
         rational additional error $\delta>0$.
\Ensure A rational $(\epsilon_0+\delta)$-approximate HZ equilibrium
         $(X,\mathbf p)$ for $u$.
\State \textbf{Choose precisions.} Set $\bar\delta=\min\{\delta,1\}$
and choose $\nu,m,K,r,\eta,\sigma$ by \Cref{eq:overview-tolerances},
using $\epsilon=\epsilon_0$, $\delta=\bar\delta$, and the fixed
$E_{\max},k_{\rm mass}$ above.
\State \textbf{Normalize and round.} In each nonconstant row of $u$,
subtract the row minimum and divide by the row range to obtain $v$;
replace constant rows by zero. Set
$\widehat v_{ij}=\nu\lfloor v_{ij}/\nu\rfloor$.
\State \textbf{Compute the allocation and weights.} For $\widehat v$,
compute a rational $X\in\B_n$ satisfying
\[
 \max_{Y\in\B_n}\langle G(X),Y-X\rangle\le\eta,
 \qquad
 \sum_{j:\widehat v_{ij}=1}X_{ij}\ge m
 \quad\text{for every nonconstant row }i.
\]
\Statex \hspace{\algorithmicindent}Compute a nonnegative rational matrix
$\widetilde G$ with
$|\widetilde G_{ij}-G_{ij}(X)|\le\sigma$ for every $i,j$.
\State \textbf{Compute common prices.} Find a rational optimal solution
$(\mathbf q,\boldsymbol\mu)$ of
\[
 \min_{\mathbf q,\boldsymbol\mu\in\mathbb R^n}\left\{\sum_jq_j+\sum_i\mu_i:
          q_j+\mu_i\ge\widetilde G_{ij}\quad(i,j\in[n])\right\}.
\]
\Statex \hspace{\algorithmicindent}Put $c=\min_jq_j$ and replace
$q_j$ by $q_j-c$ and $\mu_i$ by $\mu_i+c$.
\State \textbf{Return} $(X,\mathbf p)$, where
$\mathbf p=(1+\epsilon_0)\mathbf q$.
\end{algorithmic}
\end{algorithm}

We now instantiate the construction with the 144-knot coefficients in
\Cref{def:numerical-field}. Fix
\[
 \epsilon_0=0.0755211,\qquad E_{\max}=5,\qquad k_{\rm mass}=3.
\]

The input matrix $u$ contains the original utilities. The algorithm
normalizes it to $v$ and rounds $v$ to $\widehat v$; all weights and
assignment duals are computed for $\widehat v$.

For requested additional error $\delta>0$, use the allowance
$\bar\delta=\min\{\delta,1\}$. In \Cref{eq:overview-tolerances},
$\nu$ is the rounding grid size, $\eta$ is the allowed improvement
in the frozen-weight assignment objective, and $\sigma$ is the
allowed error in each computed weight. The same formula gives the
lower bound $m$ on the mass assigned to utility-one goods, the
coefficient Lipschitz bound $K$, and the auxiliary precision
parameter $r$. All are determined from $n$
and $\bar\delta$ before computing $X$.

The allocation computation in Step 3 uses the VI routine of
\Cref{prop:computed-vi} with
$\kappa=\eta/[2n(n+\eta)]$ and $\eta_\kappa=\eta/2$.
\Cref{lem:rational-computation} gives the stated mass bound, weight
approximations, and rational dual computation. The preceding error
budget applies to the original utilities $u$, so the returned pair
has spending at most $1+\epsilon_0+\bar\delta/4$ and utility loss at
most $\epsilon_0+\bar\delta\le\epsilon_0+\delta$.
The output allocation is the computed $X$ itself; unit demand, market
clearing, and minimum-price normalization hold exactly.
\FloatBarrier

\section{Conclusion}\label{sec:conclusion}
We have established that exact HZ equilibrium computation is FIXP-complete, with hardness persisting even when considering equilibrium allocations alone, where each source probability is recovered from a single allocation entry. Our reduction couples arithmetic in a box variational inequality with a tailored market construction whose budgets and clearing conditions enforce the system, underpinned by a global error bound that forces exact gate evaluations. On the algorithmic side, we present a deterministic polynomial-time scheme achieving an additive error of $0.0755211 + \delta$ while maintaining exact unit demand and market clearing.

Closing the gap to the optimal polynomial-time error bound remains an open problem; further progress may rely on novel coefficient functions, tighter scalar bounds, or alternative allocation-to-price mappings. Finally, a natural open question for exact computation is whether FIXP-hardness persists when agents are restricted to three or four distinct utility values.

\section*{Use of AI Tools}
\phantomsection
\addcontentsline{toc}{section}{Use of AI Tools}
OpenAI Codex assisted with the coefficient-function analysis, computation of allocations and prices, checking tools, interval proofs, and exposition. The human authors are responsible for the results, references, and final verification.

\appendix
\crefalias{paragraph}{appendix}
\section{Additional Related Work}\label{app:related-work}
Work on HZ computation includes FIXP membership and a strongly polynomial
algorithm for bi-valued utilities~\cite{VaziraniYannakakis2025}, as well as
polynomial-time algorithms with a fixed number of agent or good
types~\cite{AlaeiJalalyTardos2017}. These tractable cases contrast with
PPAD-hardness of inverse-polynomial approximation under three or four
utility values~\cite{ChenChenPengYannakakis2022,PanageasSheinTaherijamVazirani2026}.
Our FIXP-hardness result completes the classification of exact HZ
computation. For general normalized utilities, our approximation algorithm
improves the $1/e$ guarantee of an earlier version of this
work~\cite{YanLiu2026}.

Fisher markets provide a close comparison because buyers also have fixed
budgets. For linear utilities, convex programs and exact polynomial-time
algorithms are available~\cite{EisenbergGale1959,DevanurPapadimitriouSaberiVazirani2008,Orlin2010}.
Broader convex formulations cover Eisenberg--Gale markets and variants
with caps on seller earnings or buyer
utility~\cite{JainVazirani2010,ColeDevanurGkatzelisJainMaiVaziraniYazdanbod2017}.
The complexity changes with preferences and budget formation: computing
equilibria in separable piecewise-linear concave Fisher markets is
PPAD-complete under standard existence
conditions~\cite{VaziraniYannakakis2011}, whereas Leontief exchange markets
are FIXP-hard~\cite{GargMehtaVaziraniYazdanbod2017}.
Unlike exchange markets, HZ fixes budgets independently of prices;
unlike linear Fisher markets, it also requires each bundle to have total
quantity one. Our result shows that the interaction of these budget and
quantity constraints yields FIXP-hardness even with linear utilities.

HZ belongs to a broader literature on fair random assignment using
cardinal or ordinal
preferences~\cite{HyllandZeckhauser1979,AbdulkadirogluSonmez1998,BogomolnaiaMoulin2001,BogomolnaiaMoulin2004}.
Related models incorporate priorities, general allocation constraints,
and initial
endowments~\cite{HeMirallesPyciaYan2018,BudishCheKojimaMilgrom2013,EcheniqueMirallesZhang2021,GargTrobstVazirani2024}.
Nash bargaining offers other matching mechanisms and algorithmic
tradeoffs between efficiency and
envy-freeness~\cite{HosseiniVazirani2022,PanageasTrobstVazirani2026,TrobstVazirani2024}.
Approximation also has different meanings across these models. For
example, approximate competitive equilibrium for combinatorial
assignment allows small budget differences and market-clearing
error~\cite{Budish2011}; our approximation keeps unit demand and market
clearing exact and bounds additive budget and utility errors.

Market methods also support applications beyond matching:
spending-constraint utilities model advertising markets, and covering
constraints model scheduling and communication
problems~\cite{Vazirani2010Spending,DevanurGargMehtaVaziraniYazdanbod2018}.
Other tractable variants use perfect price discrimination for
piecewise-linear concave utilities~\cite{GoelVazirani2011}.
Fisher-market relaxations and convex duality also yield approximation
algorithms for Nash social welfare with indivisible
goods~\cite{ColeGkatzelis2015,ColeDevanurGkatzelisJainMaiVaziraniYazdanbod2017}.

\section{Technical Details of the Hardness Reduction}\label{app:fixp-estimates}
This appendix gives the gate and parameter bounds used in
\Cref{sec:fixp-network} and the proof of the market correctness
lemma used in \Cref{sec:fixp-market}.

\subsection{\texorpdfstring{Omitted Proofs for \Cref{sec:fixp-network}}{Omitted Proofs for Section \getrefnumber{sec:fixp-network}}}\label{app:fixp-network}
We analyze the intermediate reciprocal graph in \Cref{sec:fixp-network} with
node $0$ fixed at one and all intervals unscaled, before the final
node replacement and scaling. We first establish the local gate
and gadget bounds, then specify the parameters and prove that
the intermediate reciprocal graph has zero error.

\subsubsection{The Private Variable and the Output Contribution}\label{app:fixp-private}
We prove the sign relation in \Cref{lem:fixp-gate} by comparing
the private variable with its center. Set $f(y,A)=1/Q(y,A)-y$.
The harmonic, reciprocal, and feedback constructions satisfy
\[
 Q(y,A)>0,\qquad
 \sigma\frac{\partial f}{\partial y}(y,A)\le-m_0<0
\]
on their declared intervals. The gate estimates below verify this
condition for the first two types; for feedback,
$f(y_j,Y_j)=Y_j+(d_j-1)y_j$ and $d_j<1$.
The same argument also gives the bound
\[
 |h|\le
 \frac{|x-x^\star(A)|}{m_0(\ell+T)(\ell+T-\Delta)},
\]
which will be used for the feedback gates.

\begin{proof}[Proof of \Cref{lem:fixp-gate}]
We first compare the private variable with its center to determine
the sign of the arithmetic error.
Since $K_x-Q=Q(f-x)/(x+y)$, the sign of $\mathcal G_y$ is
$\sigma\operatorname{sgn}(f-x)$. Comparing the private variable's box condition
with the interior point $T$ therefore gives
$\sigma(x-f(y,A))(y-T)\le0$. Applying the derivative bound between
$y$ and $T$ yields
\begin{equation}\label{eq:fixp-gate-restore}
 \sigma(x-x^\star(A))(y-T)\le-m_0(y-T)^2.
\end{equation}
For $y\ne T$, this proves the sign assertion and
$|y-T|\le|x-x^\star(A)|/m_0$. If $y=T$, its interior box condition instead
gives $\mathcal G_y=0$, hence $x=x^\star(A)$; conversely, $x=x^\star(A)$ in
\Cref{eq:fixp-gate-restore} forces $y=T$. Substituting into the
expression for $h$ in \Cref{eq:fixp-gate-fields} proves all the claims.
\end{proof}

The output-node condition bounds the gate's contribution by the
contributions from later gates whenever $x^\star(A)\in[\ell,u]$.
Write those contributions as $J^{\rm ext}$. At $x=\ell$, the sign
relation gives $h\ge0$, while the box condition gives
$h+J^{\rm ext}\le0$; at $x=u$ both inequalities reverse.
In the interior, $h+J^{\rm ext}=0$. Thus
$|h|\le|J^{\rm ext}|$ in all three cases.

\subsubsection{Gate Bounds}\label{app:fixp-gate-estimates}
We bound the local arithmetic error and the contributions at input nodes, as used in
\Cref{sec:fixp-closure}, and verify the derivative condition used in the proof of
\Cref{lem:fixp-gate} for the two basic gate types.

The edge rule gives the stated reciprocal contributions directly.
For a positive term, its ratio and multiplicity give
\begin{equation}\label{eq:fixp-rational-edge}
 \frac{c}{\beta}\,
 \frac{\beta/\alpha}{r_i+(\beta/\alpha)r_j}
 =\frac{c}{\alpha r_i+\beta r_j}.
\end{equation}
Reversing the edge and exchanging $\alpha,\beta$ gives the negative
term. Summing the connections in \Cref{tab:fixp-gate-edges} at
each endpoint yields \Cref{eq:fixp-gate-fields}.

We next bound the arithmetic error and the input contributions
in terms of $|h|$. Let $M_f$ bound $|\partial f/\partial y|$ on the
independent intervals. For an interior $y$, its equation gives
$x=f(y,A)$, so the mean value theorem bounds the error by
$M_f|y-T|$. At an endpoint, use $|x-x^\star(A)|\le u-\ell$ and
$|y-T|=\Delta$ instead. Together with
$|y-T|=(x+T)(x+y)|h|$, these two cases give
\begin{equation}\label{eq:fixp-gate-error}
 |x-x^\star(A)|\le C|h|,\qquad
 C=(u+T)(u+T+\Delta)\max\{M_f,(u-\ell)/\Delta\}.
\end{equation}
For the inputs, subtracting the two fractions in $P_q$ gives
\begin{equation}\label{eq:fixp-input-bound}
 |P_q|\le\Gamma_q|h|,\qquad
 \Gamma_q\ge\sup
 \frac{|w_q|d_q(x+T)(x+y)}
 {(d_qy+e_qA_q)(d_qT+e_qA_q)}.
\end{equation}
The supremum ranges over the declared independent intervals. If
$a_q$ is the lower endpoint of input $A_q$, the explicit choice
\begin{equation}\label{eq:fixp-input-endpoint}
 \Gamma_q=
 \frac{|w_q|d_q(u+T)(u+T+\Delta)}
 {(d_q(T-\Delta)+e_qa_q)(d_qT+e_qa_q)}
\end{equation}
bounds the numerator from above and both denominator factors from
below. This is the rational endpoint bound used in the construction.
The bound on $|P_q|$ and
\Cref{lem:fixp-gate} also apply to targets outside the output interval,
as needed for the feedback equations in \Cref{sec:fixp-closure}.

\paragraph{Bounds for the Harmonic Gate.}
We verify that the separated intervals make the harmonic gate
satisfy the condition in \Cref{app:fixp-private}.
For the input intervals $[a_1,b_1]$ and $[a_2,b_2]$, set
\[
 m_0=\frac{\rho(1-\rho)(a_2-b_1)^2}
              {(1.001+\rho b_2+(1-\rho)b_1)^2},\qquad
 M_f=\frac{\rho(1-\rho)(b_2-a_1)^2}
              {(0.999+\rho a_2+(1-\rho)a_1)^2}.
\]

\begin{lemma}\label{lem:fixp-harmonic-bounds}
Throughout the harmonic gate's declared intervals, $Q>0$,
$f(1,A,B)=H_\rho(A,B)$, and $0<m_0\le f_y\le M_f$.
The intended output $H_\rho(A,B)$ lies in the declared output interval.
\end{lemma}
\begin{proof}
The two positive terms of $Q$ give $Q>0$, and substituting $y=1$
gives the stated value of $f$. Differentiation yields
\[
 \frac{\partial f}{\partial y}
 =\frac{\rho(1-\rho)(A-B)^2}
        {(y+\rho B+(1-\rho)A)^2}.
\]
Since $A\in[a_1,b_1]$, $B\in[a_2,b_2]$, and $b_1<a_2$,
we have $a_2-b_1\le B-A\le b_2-a_1$.
Bounding the denominator at its endpoints gives
$m_0\le f_y\le M_f$. In particular, $\sigma=-1$ gives
the derivative condition in \Cref{app:fixp-private}.
Finally, $H_\rho$ is increasing in both inputs and lies between
them. Its endpoint values therefore give the declared positive
output interval, including when that interval is enlarged.
\end{proof}

\paragraph{Bounds for the Reciprocal Gate.}
We verify the reciprocal gate's derivative condition and show that
the center chosen in \Cref{eq:fixp-reciprocal-center} gives the
required bound on the input contribution. Only this second
estimate requires the center to grow with $1/\varepsilon$.

\begin{lemma}\label{lem:fixp-reciprocal-bounds}
Throughout the reciprocal gate's declared intervals,
$Q>0.9/T$, $f(T,A)=J(A)$, and $f_y<-0.98$.
\end{lemma}
\begin{proof}
The prescribed two reciprocal terms first give the intended
function at the private center. With $R_T=T+\beta$,
\[
 \frac1{T+J(A)}
 =\frac1{R_T}-\frac{\kappa}{R_T\bigl(\kappa+R_T(A+c)\bigr)}.
\]
This is $Q(T,A)$ in \Cref{eq:fixp-reciprocal-Q}, so
$f(T,A)=J(A)$. To apply the argument in \Cref{app:fixp-private}, we check positivity of $Q$ and
the derivative bound throughout the declared intervals. The choice of $T$ gives
$0.95T\le R_T\le1.05T$ and
$\kappa/(A+c)\le u+|\beta|\le T/20$, so
$R_T(A+c)\ge19\kappa$. Since $y\ge0.999T$, the second denominator
is at least $0.999\bigl(\kappa+R_T(A+c)\bigr)\ge19.98\kappa$.
It follows that
\[
 Q\ge\frac1{R_T}\left(\frac{101}{101.001}-\frac1{19.98}\right)
      >\frac{0.9}{T}.
\]
In the derivative of \Cref{eq:fixp-reciprocal-Q}, the second term
contributes positively. Hence
\[
 -Q_y\le\frac{101T}{R_T(y+100T)^2}<\frac{0.011}{T^2},
 \qquad f_y=-Q_y/Q^2-1<-0.98.
\]
Thus the argument in \Cref{app:fixp-private} applies with $m_0=0.98$.
\end{proof}

For the error estimate in \Cref{eq:fixp-gate-error}, an explicit
upper bound on $|f_y|$ is
\[
 M_f=1+\frac{T^2}{0.9^2}\left(
 \frac{101T}{R_T(100.999T)^2}
 +\frac{\kappa D_T}{R_TT(0.999D_T+R_Ta)^2}\right).
\]
Indeed, the parentheses bound $|Q_y|$ by the sum of its two
absolute terms, and $Q^2\ge0.9^2/T^2$.

We now prove \Cref{lem:fixp-reciprocal-input}, which bounds how much a reciprocal gate changes the
equation at its input node. This bound will control the interaction
between arithmetic gadgets.

\begin{proof}[Proof of \Cref{lem:fixp-reciprocal-input}]
We bound the ratio in \Cref{eq:fixp-input-bound} as $T$ increases.
Use
$|w|=\kappa/R_T$, $d=D_T/T$, and $e=R_T$. The coarser bounds
$R_T\in[T/2,2T]$, $y\ge T/2$, and
$(x+T)(x+y)\le4T^2$ suffice. The two denominator factors are at
least $T(A+c)/4$ and $T(A+c)/2$, while the numerator is at most
$8\kappa D_T$. Therefore
\[
 \Gamma\le\frac{64\kappa D_T}{T^2(a+c)^2}
 \le\frac{128\kappa(\kappa+c)}{T(a+c)^2}.
\]
Substituting the chosen center in the last bound gives
$|P_A|\le\varepsilon|h|$.
\end{proof}

\subsubsection{Arithmetic Identities and Interval Bounds}\label{app:fixp-arithmetic-estimates}
We verify that the connected gates perform the stated operations
and that every gate's target lies in its declared output interval.
These checks concern independent input values; the later error
argument will make all gates attain their targets simultaneously.

\emph{Affine functions and addition.}
The two affine gates give $\alpha/(1/A)+\beta=\alpha A+\beta$.
For addition, substitute the first two outputs of
\Cref{eq:fixp-addition} into the harmonic reciprocal sum:
$\rho/(1+A_1)+(1-\rho)/(1+B_1)
=(A+B)/(K_1+K_2)$. Hence the last gate gives $X=A+B$.
The first two output intervals are separated by $6U$, so the
harmonic gate applies to every pair of values in those intervals.
Its endpoint range is the interval declared for $Z$, which the
last reciprocal gate maps to $[2,2U]$. The same monotonicity
gives the stated intervals for narrower inputs and affine sums.

\emph{Complement and negation.}
The first two gates in \Cref{eq:fixp-complement} give
$t_1=D(D-X)/X$ and $t_2=D^2/(D-X)$, so the last gives
$N_D(X)=D-X$. Since $D>b\ge X\ge a>0$, all intermediate
values are positive. Taking endpoint values in these three
monotone functions yields exactly the declared intervals.
For a stored value $X=v+b_X$, the output is
$-v+(D-b_X)$; with $D=2U+2$, the new offset belongs to $[0,D]$.

\emph{Multiplication.}
The first step of the multiplication gadget produces the two
shifted sums $u_+=c+v+w$ and $u_-=c+v-w$ when its gates are
exact. To check their ranges even for independent intermediate
inputs, the complement of $Y$ lies in $[U+2,2U+1]$.
Thus the actual input sum for the gadget producing $u_-$ lies in
$[U+3,3U+1]$. Its additive shift lies in
$[c-3U-2,c-U-2]$, giving
$u_-\in[c-2U+1,c+2U-1]\subseteq\mathcal J$.
Similarly, $u_+\in[c-2U+2,c+2U]\subseteq\mathcal J$.

For either branch, the complement with constant $2c+1$ is valid
because $2c+1>c+4U$. Its rescaled output is
$V=4c+1-2u$ in the declared interval. The gap between the upper
endpoint for $u$ and the lower endpoint for $V$ is
$c+1-12U=6$. Both intervals are contained in $[1,3c]$.
Consequently the harmonic gate is valid on the independent input
box and its output lies in $[1,3c]$.

The harmonic output provides the quadratic term needed for a
product. At the intended relation, write $t=u-c$. The reciprocal
sum is $1/[3(c+1+t)]+1/[3(c+1-t)]$, which gives \Cref{eq:fixp-square}.
Using $u_\pm=c+v\pm w$ in \Cref{eq:fixp-square} gives
$a(R_c(u_-)-R_c(u_+))=vw$. Therefore the last two gadgets give \Cref{eq:fixp-product}.
For independent values $R_\pm\in[1,3c]$, the complement is in
$[3c+2,6c+1]$. Their sum is positive and at most $9c+1$,
which verifies $W\in[1,U_W]$ and $0\le b_W\le U_W$.

The multiplication gadget also has the size and input connections
used in the parameter choice. It has 26 gates: three for the
initial complement, four for each shifted sum, four for each
harmonic branch, three for the last complement, and four for
the final affine sum. In this gadget and all the other arithmetic
gadgets, every variable input from outside the gadget first enters
a reciprocal gate whose only variable input is that external
input. Increasing those gates' centers therefore leaves all
internal-read ratios unchanged.

\emph{Rational coefficients.}
The bit recurrence gives
$A_j-1=2^{-j}\sum_{s=0}^{j-1}\varepsilon_s2^s$ by induction.
Thus $A_p=1+p/2^b$ and $A_q=1+q/2^b$. Each update maps
$[1,2]$ into itself, and the final update for $q$ maps it into
$[3/2,2]$. In the nine-gate gadget, $C=3-A_q$,
$E=6-2A_q$, and $V=2(A_q-1)$, so $B=1/(A_q-1)$.
The stated intermediate intervals follow in the same order.
The final multiplication therefore gives
$(A_p-1)B=p/q$ with its recorded offset.

\emph{Formula size and intermediate values.}
Each monomial uses its own coefficient chain and a constant
number of multiplication gadgets. Copying branches for separate
uses therefore gives polynomial-size formula trees. Each path
contains only constantly many multiplication gadgets: those for
coefficient generation, the source product, and the coefficient
times that product. A multiplication changes a bound $U$ to
$O(U^2)$, but this happens only constantly many times on a path.
Bit updates stay in fixed intervals, negations occur only a
bounded number of times per term, and additions sum the offsets
and bounds of polynomially many terms. All stored bounds and
offsets are consequently polynomial in the input length.

The same formulas bound the derivatives needed for error
propagation. Addition and negation have derivatives of magnitude
one, and bit updates have derivative $1/2$. The coefficient
quotient $(A_p-1)/(A_q-1)$ has partial derivatives bounded by
two and four on $[1,2]\times[3/2,2]$; the nine-gate reciprocal
has derivative magnitude at most four. Only constantly many
factors along any formula path can exceed one by a polynomial
amount. This also bounds the products defining $K_{\rm eval}$.

\subsubsection{Choosing the Parameters}\label{app:fixp-parameters}
We now assign the centers and edge multiplicities left open above.
First compute common bounds for the formulas and feedback gates.
Then process each arithmetic gadget separately: fix its internal
centers, choose the centers at its external inputs, and multiply
its edge multiplicities by one common integer.

First compute a bound $K_{\rm eval}$ for the coordinate formulas.
For each use of a gadget input, bound the absolute partial derivative
of that gadget's intended function on its declared input box. Replace
each bound smaller than one by one. For each gadget in a coordinate
formula, multiply the bounds along its unique path to $Y_j$.
Set $K_{\rm eval}$ to the ceiling of the maximum of one and these
products; an empty path has product one.

Next compute the feedback bounds. For each feedback gate, let
$\Gamma_j^{\rm fb}$ be the endpoint bound in
\Cref{eq:fixp-input-endpoint} for its input $Y_j$ and output
interval $[1,2]$. Set
\[
 K=\left\lceil\max\left\{1,K_{\rm eval},
       \max_j\Gamma_j^{\rm fb},
       \max_j\frac{4C_j^{\rm fb}}
          {(1+T_j)(1+0.999T_j)}\right\}\right\rceil,
 \qquad \eta=\frac1{2(1+K+K^3)}.
\]
These choices precede the adjustable centers of all evaluation gadgets.

For each arithmetic gadget, number its gates $1,\ldots,t$ in
construction order, with gate $t$ supplying the final output.
Perform the following four steps independently for that gadget.
\begin{enumerate}[label=(\arabic*),leftmargin=*,itemsep=0.35em]
\item Fix the harmonic centers at one. Fix the centers of reciprocal
gates that read internal outputs at the smallest integers allowed
in their construction. For each such read from gate $i$ to gate
$j$, compute its bound $\Gamma_{ji}$ by
\Cref{eq:fixp-input-endpoint}, adding the bounds for repeated reads.
Work backwards to compute
\[
 B_t=1,\qquad
 B_i=\sum_{j:\,j\text{ reads }i}\Gamma_{ji}B_j\quad(i<t),
 \qquad B_*=\max\{1,B_1,\ldots,B_t\}.
\]
\item Every gate reading a variable input from outside the gadget
is a reciprocal gate. Choose its center by
\Cref{eq:fixp-reciprocal-center} with
$\varepsilon=\eta/(26B_*)$. A gate with only a constant input
uses the smallest allowed center.
\item With all centers fixed, compute the gate error constants
$C_i$ from \Cref{eq:fixp-gate-error}. Let $L_{ij}$ bound the
absolute derivative of gate $i$'s intended function with respect to
the output of gate $j$, summing over repeated reads. Work forwards:
\[
 D_i=C_iB_i+\sum_{j:\,i\text{ reads }j}L_{ij}D_j,
 \qquad D_g=D_t.
\]
An empty sum is zero. The derivative bounds are obtained from the
explicit harmonic and reciprocal formulas on the declared intervals.
\item Form all rational edges using \Cref{tab:fixp-gate-edges}.
Let $C_g$ be the product of their multiplicity denominators, in
lowest terms, and choose
\[
 M_g=C_g\left\lceil\max\{1,D_g/\eta\}\right\rceil.
\]
Multiply every edge multiplicity in this gadget by the same $M_g$,
then replace each integer multiplicity by that many unit edges.
\end{enumerate}

\subsubsection{Formula and Feedback Bounds}\label{app:fixp-formula-estimates}
We bound the errors in the coordinate formulas and the contributions
of the feedback gates. For a gadget $g$, let
$e_g=|x_g-\varphi_g(A_g)|$ be its local arithmetic error, evaluated
at its actual inputs $A_g$. Set $\tau_j=Y_j+1-\beta_j$.
The private-node condition gives
$\operatorname{sgn}h_j^{\rm fb}=-\operatorname{sgn}(z_j-\tau_j)$,
including when formula errors put $\tau_j$ outside $[1,2]$.
The parameter $K$ in \Cref{app:fixp-parameters} satisfies
\begin{equation}\label{eq:fixp-feedback-bounds}
 \sum_j|\tau_j-\Phi_j(z)|\le K\sum_g e_g,\qquad
 |P_j^{\rm fb}|\le K|h_j^{\rm fb}|,\qquad
 |h_j^{\rm fb}|\le K|z_j-\tau_j|.
\end{equation}

\paragraph{Relating Local Errors to the Final Outputs.}
The feedback argument will need a bound on the final output errors.
We obtain this bound by tracing each local arithmetic error along
its formula tree.
Recall that $e_g$ is the local arithmetic error of gadget $g$.
Let $E_g$ measure the deviation
from evaluating the entire formula leading to that output exactly
at the source vector $z$. The first quantity uses the actual inputs
of the gadget; the second also includes errors from earlier gadgets.
The mean value theorem on the independent input box gives
\[
 E_g\le e_g+\sum_\nu L_{g\nu}E_{g_\nu},
\]
where $g_\nu$ supplies input $\nu$ and $L_{g\nu}$ bounds the
corresponding partial derivative of $\varphi_g$.

We obtain a common error bound by expanding this inequality along
each formula tree. Each local error is weighted by the product of
derivatives on its unique path to the root.
Addition and negation contribute factors of magnitude one, bit
updates contribute $1/2$, and only a constant number of factors on
any path can exceed one by a polynomial amount. We can therefore
use the polynomially bounded path-product maximum $K_{\rm eval}$
specified in \Cref{app:fixp-parameters}. Every assignment of gadget outputs
in their declared intervals therefore satisfies
\begin{equation}\label{eq:fixp-eval}
 \sum_j|Y_j+1-\beta_j-\Phi_j(z)|\le K_{\rm eval}\sum_g e_g.
\end{equation}
For this integer bound, the factors $1/2$ may be replaced by one,
avoiding large denominators on long bit sequences. These bounds
depend only on the intended functions and their intervals. They can
therefore be computed before choosing the parameters that control
the gadgets' contributions, as required by the construction.

\paragraph{Feedback Signs and Bounds.}
We verify that the feedback gates compare each source coordinate
with the value returned by its formula. Set
$\tau_j=Y_j+1-\beta_j$. The feedback gate has
$f(y,Y_j)=Y_j-(1-d_j)y$, so $f(T_j,Y_j)=\tau_j$ and
$f_y=-(\beta_j-1)/(2\beta_j)\le-1/4$ because $\beta_j\ge2$.
Its $Q$ is positive on the whole box. Consequently
\Cref{app:fixp-private} gives a contribution of sign opposite to
$z_j-\tau_j$, even if errors in $Y_j$ put $\tau_j$ outside $[1,2]$.

Clearing a feedback gate's multiplicity denominators multiplies
both its contributions by $C_j^{\rm fb}$. This preserves the
sign and the input-to-output ratio. The upper bound on
$|h_j^{\rm fb}|$ instead gains that multiplier, which is included
in the chosen $K$. Combining \Cref{eq:fixp-eval}, the endpoint bound on the input contribution,
and the bound on the output contribution in \Cref{app:fixp-private} gives
\Cref{eq:fixp-feedback-bounds}.
All feedback parameters depend only on the intended formulas and
their intervals. The definition of $K$ therefore precedes the
arithmetic-gadget parameter choices and is polynomially bounded.

\subsubsection{Bounds for Arithmetic Gadgets}\label{app:fixp-group-estimates}
We prove \Cref{lem:fixp-group} by deriving the backwards and
forwards estimates used to choose the centers and multiplicities.
The required bounds are stated in \Cref{eq:fixp-group-bounds}.

We first analyze a gadget before multiplying its edges by $M_g$.
Number its gates $1,\ldots,t$ in construction order, with gate $t$
giving the final output, and write $h_i$ for the contribution at
gate $i$'s output. The constants $B_i$ are the backwards bounds
specified in \Cref{app:fixp-parameters}. The next lemma reduces all internal
contributions to the single quantity $|h_t|$.

\begin{lemma}\label{lem:fixp-group-contributions}
If every private node and every internal output node satisfies its
box condition, then
$|h_i|\le B_i|h_t|$ for every gate $i$.
\end{lemma}
\begin{proof}
If gate $j$ reads gate $i$'s output, choose $\Gamma_{ji}$ so that
the magnitude of this input contribution is at most $\Gamma_{ji}|h_j|$,
as in \Cref{eq:fixp-input-bound}. When there is more than one such
read, add their bounds in $\Gamma_{ji}$.

At an internal output, only later gates in this gadget can contribute
besides its own gate. Since its target lies in the declared output
interval, the bound on the output contribution following \Cref{lem:fixp-gate} says that
its own contribution cannot exceed the total magnitude of these
later contributions. Applying the triangle inequality gives
\[
 |h_i|\le\sum_{j:\,j\text{ reads }i}\Gamma_{ji}|h_j|
 \qquad(i<t).
\]
The construction computes $B_i$ by this backwards
recurrence, starting at $B_t=1$.
The bound $|h_t|\le B_t|h_t|$ is immediate. If the corresponding
bound holds for every gate that reads $i$, then
\[
 |h_i|\le\sum_{j:\,j\text{ reads }i}\Gamma_{ji}|h_j|
 \le\sum_{j:\,j\text{ reads }i}\Gamma_{ji}B_j|h_t|
 =B_i|h_t|.
\]
This proves $|h_i|\le B_i|h_t|$ for all gates. Thus one quantity,
$|h_t|$, controls all the internal contributions, including when
an output is read by several later gates.
\end{proof}

Once the centers are fixed, the forwards bounds $D_i$ account for
errors passed from one gate to the next. They give the following
bound on the gadget's final arithmetic error.

\begin{lemma}\label{lem:fixp-group-error}
Under the hypotheses of \Cref{lem:fixp-group-contributions},
$e_g\le D_g|h_t|$, where $D_g=D_t$ is the forwards bound in the
construction.
\end{lemma}
\begin{proof}
With the centers fixed, we first bound the error introduced at each
gate. By \Cref{eq:fixp-gate-error}, the
deviation of gate $i$'s output from its intended function at its actual
inputs is at most $C_i|h_i|\le C_iB_i|h_t|$, for a computable
constant $C_i$. To obtain the error of the whole gadget, we must
also account for errors in the inputs supplied by earlier gates.

To include errors inherited from earlier gates, we compare each output
with an exact evaluation of the whole gadget. Let $E_i$ be the
deviation of gate $i$'s actual output from its
value when all gates in this gadget are evaluated exactly at the same
external inputs. Let $L_{ij}$ bound the absolute partial derivative of
gate $i$'s intended function with respect to the input supplied by
gate $j$, adding the bounds if that input occurs more than once.
Both the actual and exactly evaluated inputs lie in their declared
intervals. The mean value theorem and the triangle inequality thus give
\[
 E_i\le C_iB_i|h_t|
       +\sum_{j:\,i\text{ reads }j}L_{ij}E_j.
\]
The first term is the error introduced at gate $i$; the sum is the
error passed on by earlier gates. External inputs contribute no term
to this sum because they are held fixed in the exact evaluation.

The construction's forwards recurrence for $D_i$ propagates
these error bounds to the gadget's final output.
For a gate with no internal inputs, the sum is empty and
$E_i\le D_i|h_t|$. Substituting this bound for the earlier errors
proves it at every subsequent gate. At the final gate, $E_t=e_g$,
so setting $D_g=D_t$ gives $e_g\le D_g|h_t|$.
Each use of an earlier output contributes a separate term, so the
recurrence also includes every path by which an error reaches the
final output. Since the gadget has constant size, these recurrences
use only constantly many rational operations.
\end{proof}

We now choose the centers at external inputs and multiply all edges
in the gadget by the same integer. The two preceding lemmas then
give the three estimates needed for the intermediate reciprocal graph.

\begin{proof}[Proof of \Cref{lem:fixp-group}]
\emph{Reducing contributions at external inputs.}
Every gate that uses an input from outside the gadget is a reciprocal
gate whose only variable input is that external input. In particular,
it never reads an internal output. The ratios $\Gamma_{ji}$ above
therefore belong only to gates that do not read external inputs.
We first fix the centers of those gates and compute all $B_i$.
The centers of the gates that read external inputs can then be
increased without changing these bounds or any intended function.

We now choose those centers to reduce the sum of the external input
contributions. Put $B_* =\max\{1,B_1,\ldots,B_t\}$.
There are at most 26 uses
of external inputs. For each such gate, choose its center by
\Cref{eq:fixp-reciprocal-center} so that its input ratio is at most
$\eta/(26B_*)$. If this gate is $i$ and its input contribution is
$P$, \Cref{lem:fixp-group-contributions} gives
\[
 |P|\le\frac{\eta}{26B_*}|h_i|
 \le\frac{\eta B_i}{26B_*}|h_t|
 \le\frac{\eta}{26}|h_t|.
\]
Summing over the at most 26 reads yields
$\sum_{\text{external inputs}}|P|\le\eta|h_t|$.

\emph{Scaling the gadget.}
Multiplying every contribution by the same positive integer changes
neither the intended functions nor the arithmetic errors at a fixed
assignment. It does increase the output contribution against which
we compare those errors. We choose this multiplier also to make the
edge multiplicities integers. Express the terms as rational edges
using \Cref{eq:fixp-rational-edge}. The chosen denominator product
$C_g$ divides $M_g$, so every scaled multiplicity is integral.
At any solution of the resulting graph's box conditions,
dividing the internal field inequalities by $M_g$ gives the hypotheses of
\Cref{lem:fixp-group-contributions,lem:fixp-group-error}.
With $h_g=M_gh_t$ and $M_g\ge D_g/\eta$, we obtain
\[
 e_g\le D_g|h_t|=\frac{D_g}{M_g}|h_g|\le\eta|h_g|.
\]
The external input contributions are multiplied by $M_g$ as well,
so their sum is at most $\eta M_g|h_t|=\eta|h_g|$.

The output box condition provides the third bound by comparing the
gadget's contribution with the contributions from other gadgets.
The final gate's target lies in its output interval.
At an interior output, its total field is zero, so
$h_g=-J_g^{\rm ext}$. At the lower endpoint, the sign relation
in \Cref{lem:fixp-gate} gives $h_g\ge0$, while the box condition
gives $h_g+J_g^{\rm ext}\le0$. Hence
$|h_g|\le|J_g^{\rm ext}|$. At the upper endpoint the two signs
reverse and give the same bound. These are the three inequalities
in \Cref{eq:fixp-group-bounds}.

The parameter choices also keep the gadget polynomial in size.
The endpoint bounds defining $B_i$ involve only the fixed centers
of gates that use internal inputs. Thus each $B_i$ is polynomially
bounded in the gadget data. By \Cref{lem:fixp-reciprocal-input}, the
centers chosen for external inputs are polynomially bounded in
those data and $1/\eta$. The derivative bounds and the forwards
recurrence then give the same bound for $D_g$.
Since a gadget has only constantly many terms with polynomially bounded
numerators and denominators, their denominator product is a
polynomially bounded choice of $C_g$. Thus $M_g$ and the resulting
integer multiplicities have the required polynomial bounds as well.
\end{proof}

\subsubsection{Exactness in the Intermediate Reciprocal Graph}\label{app:fixp-exactness}
We sum the gadget and feedback bounds to prove exactness, including
when source coordinates lie at the ends of their intervals.

\begin{proof}[Proof of \Cref{lem:fixp-exact}]
We first bound the evaluation errors by the total magnitude of the
arithmetic output contributions.
Let $\mathcal R=\sum_g|h_g|$, where the sum includes all arithmetic
gadgets and complement branches, but not the feedback gates.
By \Cref{eq:fixp-group-bounds}, the sum of all arithmetic input
contributions in absolute value is at most $\eta\mathcal R$, and
$\sum_g e_g\le\eta\mathcal R$. Hence
\Cref{eq:fixp-feedback-bounds} gives
$\sum_j|\tau_j-\Phi_j(z)|\le K\eta\mathcal R$.

We next bound the feedback contributions at the source coordinates.
At source coordinate $z_j$, write the total field as
$h_j^{\rm fb}+L_j$, where $L_j$ collects the arithmetic input
contributions. Thus $\sum_j|L_j|\le\eta\mathcal R$.
If $z_j$ is interior, its total field is zero, giving
$|h_j^{\rm fb}|=|L_j|$. The same upper bound holds at $z_j=1$
when $h_j^{\rm fb}\ge0$, because the lower-endpoint condition is
$h_j^{\rm fb}+L_j\le0$. At $z_j=2$ with $h_j^{\rm fb}\le0$,
the upper-endpoint condition gives it as well.

The other endpoint cases require the fact that $\Phi$ maps the cube
to itself. If $z_j=1$ and $h_j^{\rm fb}<0$, the sign relation
implies $\tau_j<1$, and therefore
$|z_j-\tau_j|=1-\tau_j\le\Phi_j(z)-\tau_j$.
If $z_j=2$ and $h_j^{\rm fb}>0$, it instead implies $\tau_j>2$,
so $|z_j-\tau_j|=\tau_j-2\le\tau_j-\Phi_j(z)$.
In either case the last bound in \Cref{eq:fixp-feedback-bounds}
gives $|h_j^{\rm fb}|\le K|\Phi_j(z)-\tau_j|$.
Combining all cases yields
\[
 \sum_j|h_j^{\rm fb}|
 \le\sum_j|L_j|+K\sum_j|\tau_j-\Phi_j(z)|
 \le\eta(1+K^2)\mathcal R.
\]

We can now show that the total arithmetic output contribution must
be zero. Sum the bound on the output contribution $|h_g|\le|J_g^{\rm ext}|$ over arithmetic
gadgets. Their external contributions come either from other arithmetic
gadgets or from feedback gates. Using the bounds on input contributions for both types gives
\[
 \mathcal R\le\eta\mathcal R+\sum_j|P_j^{\rm fb}|
 \le\eta\mathcal R+K\sum_j|h_j^{\rm fb}|
 \le\eta(1+K+K^3)\mathcal R=\frac{\mathcal R}{2}.
\]
Since $\mathcal R\ge0$, it follows that $\mathcal R=0$.
The feedback bound then makes every $h_j^{\rm fb}$ zero, and the
inequalities in \Cref{lem:fixp-group-contributions} make
every internal gate contribution zero as well. By
\Cref{eq:fixp-gate-sign}, each output equals the intended value at its inputs and
each private variable equals its center. Consequently
$Y_j=\mathcal F_j(z-\ones)+\beta_j$ and
$z_j=Y_j+1-\beta_j=\Phi_j(z)$; the separate complement branches
likewise give $D=4-A$.
\end{proof}

\subsection{\texorpdfstring{Proof of \Cref{lem:fixp-market-correctness}}{Proof of Lemma \getrefnumber{lem:fixp-market-correctness}}}\label{app:fixp-market}
We prove \Cref{lem:fixp-market-correctness}. We first establish
the supporting price bounds, then combine them with edge demand
and market clearing in a single proof.

\subsubsection{Supporting Price Bounds}\label{app:fixp-market-prices}
We use the supply choices and demand optimality to obtain the
bounds needed for the price formula in \Cref{eq:fixp-price-form}.

\paragraph{Supply Bounds and Demand Optimality.}
We first verify the supply inequalities and the demand conditions
used to determine prices. Since $\lambda_{it}\le a_i<b_i$,
all utilities in \Cref{tab:fixp-utilities} lie in $[0,1]$, and
every own low good has positive utility.

The integer supply choices balance the high and low terms exactly.
Indeed, $\lambda_{it}L_{it}=\ell_0P_{it}$, so summing over $t$
gives $H_i=\sum_t\lambda_{it}L_{it}$. Also
$L_{it}\ge\ell_0>2W$ and $H_i\ge\ell_0>W+1$.
The latter bound holds even without outgoing edges because the
extra low group is always present. After reserving one high good,
the ordinary high goods therefore still outnumber all edge agents.

The reference supply is positive and occupies at least two thirds
of the market. By $S_i\le S_*$ and $B_0=3S_*$,
\[
 S_R\ge N-kS_*=N-\frac{kB_0}{3}
       =\frac{2N}{3}+\frac W3\ge\frac{2N}{3}.
\]
Thus the construction gives the bounds
\begin{equation}\label{eq:fixp-supply-bounds}
 H_i=\sum_t\lambda_{it}L_{it},\quad
 L_{it}>2W,\quad H_i>W+1,\quad B_0>S_i,\quad S_R/N\ge2/3.
\end{equation}
The low supply leaves room for two sources of demand besides the
own node class: the $W$ edge agents and other nodes' agents.
The price proof will bound each source by $W$.

To analyze optimal bundles, we use the standard HZ demand dual
\cite{VaziraniYannakakis2025,ChenChenPengYannakakis2022} at prices
$\boldsymbol\pi$:
\begin{equation}\label{eq:fixp-demand-dual}
 \min\{\gamma+\mu:\ \gamma\ge0,\ \mu\in\mathbb R,
          \ v_j\le\gamma\pi_j+\mu\text{ for all }j\}.
\end{equation}
Strong duality and complementary slackness give equality for
every consumed good, while $\gamma>0$ forces full spending
\cite{Schrijver2003}. These equalities explain the node utilities.
At the intended prices, take $\gamma_i=r_i/(b_i+1)$ and
$\mu_i=(b_i-r_i)/(b_i+1)$; for a low good, for example,
\[
 \gamma_i\left(1-\frac{\lambda_{it}}{r_i}\right)+\mu_i
 =\frac{r_i-\lambda_{it}+b_i-r_i}{b_i+1}
 =\frac{b_i-\lambda_{it}}{b_i+1}.
\]
The next proof reverses this calculation: positive consumption
will make the dual constraints tight and determine the prices.

\paragraph{Prices at Every Equilibrium.}
We first normalize prices, then obtain equalities for high prices
and inequalities for low
prices. A bound on consumption across nodes will make the low-price
inequalities tight. The symbols $\pi_R,\pi_i^+,\pi_{it}^-$ denote
the reference, high, and low prices as in \Cref{eq:fixp-price-form}.

\begin{lemma}\label{lem:fixp-market-normalization}
Every equilibrium allocation admits supporting prices with minimum
zero. At these prices, every high good costs more than one, every
optimal bundle spends one, and every demand-dual optimum has
$\gamma>0$ and $\mu\ge0$. Goods with identical utility columns
have equal prices, and an agent consumes a zero-utility good only
if that good is free and $\mu=0$.
\end{lemma}
\begin{proof}
\emph{Full spending.}
We first force all high prices above one, which will make every
optimal bundle use the full budget.
Start from arbitrary equilibrium prices. Suppose a high good at
node $i$ had price at most one. Each of the $B_0$ node-$i$ agents
could then attain utility one by buying a unit of that good.
All other goods have strictly smaller utility for these agents, so
every optimal bundle would consist entirely of their own high goods.
Demand for the high goods would then be $B_0$ units, exceeding
their supply $H_i<B_0$.
Market clearing rules this out. Thus every high good, including
every reserved high good, has price greater than one.

For every agent, utility one is available only from high goods.
Since a unit bundle made entirely of those goods costs more than
one, optimal utility is strictly below one. If an optimal bundle
left part of the budget unspent, mixing it with a sufficiently small
amount of a utility-one high good would raise its utility and remain
affordable. Therefore every optimal bundle spends one.

We next show that goods with identical utility columns have equal
prices. Every optimal slope $\gamma$ in \Cref{eq:fixp-demand-dual}
is positive: otherwise the constraint for
a utility-one good would force $\mu\ge1$, making the dual value
at least one, contrary to the optimal utility just established.
As a result, goods with identical utility columns have equal prices:
if one were more expensive, feasibility of the cheaper good's dual
constraint and $\gamma>0$ would make the more expensive good's
constraint strict for every agent. The more expensive good would receive no demand,
contradicting market clearing. This proves common prices within
the reference, low, and ordinary high groups.

\emph{Normalizing prices and excluding unwanted consumption.}
We normalize prices to make at least one good free, which will
restrict consumption of zero-utility goods.
The $N$ agents spend $N$ in total, and each good has unit supply.
Hence market clearing gives $\sum_g\pi_g=N$. Some prices exceed
one, so the minimum price $m=\min_g\pi_g$ must be below one.
Apply the standard HZ price normalization
$\widehat\pi_g=(\pi_g-m)/(1-m)$, which preserves feasible and optimal
unit bundles \cite{VaziraniYannakakis2025}. Every optimal bundle
still spends one, and all high prices still exceed one.
We use these normalized prices and omit the hats below.

The free good rules out consumption of zero-utility goods at positive
prices. Its nonnegative utility and its dual
constraint imply $\mu\ge0$ for every agent. If an agent consumes
a zero-utility good, complementary slackness gives
$0=\gamma\pi_g+\mu$. Both terms are nonnegative and $\gamma>0$,
so the good must be free and $\mu=0$. This fact will control all
consumption of one node's goods by another node's agents.
\end{proof}

For the rest of the price argument, use the normalized prices from
\Cref{lem:fixp-market-normalization}. All agents of node $i$ have
the same demand problem; choose one optimal dual pair
$(\gamma_i,\mu_i)$ for that class and set $r_i=(b_i+1)\gamma_i$.
Positive consumption of high and reference goods will relate
their prices to $r_i$.

\begin{lemma}\label{lem:fixp-first-prices}
Every node class consumes positive amounts of its ordinary high
goods and the reference goods. Every reserved high good has the
same price as its node's ordinary high goods. Moreover,
\[
 \pi_i^+=\pi_R+1/r_i,\qquad
 \pi_{it}^-\ge\pi_R-\lambda_{it}/r_i,\qquad \pi_R\le1.
\]
\end{lemma}
\begin{proof}
\emph{Forcing consumption of high and reference goods.}
Positive consumption of high and reference goods will give two price
equalities for each node.
The ordinary
high group contains more than $W$ goods, whereas all edge agents
together consume only $W$ units. Other nodes' agents give those
goods zero utility and cannot consume them at a positive price.
Consequently node-$i$ agents must consume a positive amount of
their ordinary high goods.

Node agents must also consume each reserved high good $h^*$ at node
$D$, which will give it the same price as the ordinary high goods. Apart
from node-$D$ agents, only its designated edge agent $w$ gives it
positive utility. Since its price exceeds one, $w$ can buy less
than its unit supply; node-$D$ agents must consume the remainder.
The dual equalities for the ordinary and reserved high goods have
the same utility one and the same pair $(\gamma_D,\mu_D)$.
Subtracting them, and using $\gamma_D>0$, shows that $h^*$ has
the common high price $\pi_D^+$ as well.

Reference consumption follows from a different supply comparison.
The total value of reference goods is at most $N$, so
$\pi_R\le N/S_R$. A node-$i$ agent can buy $S_R/N$ units of
reference goods and fill the remaining quantity with a free good.
This bundle costs at most one and gives utility at least
$[b_i/(b_i+1)]S_R/N$. Since utility per unit never exceeds one,
the total quantity of positive-utility goods consumed by the class
must be at least
\[
 B_0\frac{b_i}{b_i+1}\frac{S_R}{N}
 >3S_*\cdot\frac12\cdot\frac23=S_*\ge S_i.
\]
Here $b_i>a_i\ge1$ gives $b_i/(b_i+1)>1/2$, and the supply
choice gives $S_R/N\ge2/3$. The class has only $S_i$ own goods;
its only other positive-utility goods are the reference goods.
The node-$i$ agents must therefore consume a positive amount of
reference goods.

\emph{Obtaining the first price relations.}
We now use the high and reference consumption to make their dual
constraints equalities. The low constraints initially give inequalities:
\[
 1=\gamma_i\pi_i^++\mu_i,\qquad
 \frac{b_i}{b_i+1}=\gamma_i\pi_R+\mu_i,\qquad
 \frac{b_i-\lambda_{it}}{b_i+1}
       \le\gamma_i\pi_{it}^-+\mu_i.
\]
Subtract the reference equality from the high equality and from each
low inequality. This gives
\[
 \gamma_i(\pi_i^+-\pi_R)=\frac1{b_i+1},\qquad
 \gamma_i(\pi_{it}^--\pi_R)\ge-\frac{\lambda_{it}}{b_i+1}.
\]
Using $r_i=(b_i+1)\gamma_i>0$ and dividing by $\gamma_i$,
the price relations become
$\pi_i^+=\pi_R+1/r_i$ and
$\pi_{it}^-\ge\pi_R-\lambda_{it}/r_i$.

The supply balance now gives an upper bound on the reference price.
Summing the prices of all goods and inserting these relations yields
\begin{align*}
 N
 &=S_R\pi_R+\sum_iH_i\pi_i^+
                  +\sum_{i,t}L_{it}\pi_{it}^-\\
 &\ge S_R\pi_R+\sum_iH_i\left(\pi_R+\frac1{r_i}\right)
                  +\sum_{i,t}L_{it}\left(\pi_R-\frac{\lambda_{it}}{r_i}\right)\\
 &=N\pi_R+\sum_i\frac1{r_i}
                    \left(H_i-\sum_t\lambda_{it}L_{it}\right)
 =N\pi_R.
\end{align*}
The last sum vanishes by \Cref{eq:fixp-supply-bounds}, proving
$\pi_R\le1$. The remaining price inequalities concern the low goods.
\end{proof}

To make each low-price inequality an equality, we must leave a
positive amount of that low group for its own node agents.
Let $Z_i$ be the quantity of other nodes' goods consumed by node-$i$
agents, and let $q_i$ be the quantity of node-$i$ goods consumed
by edge agents. The following budget comparison bounds $Z_i$.

\begin{lemma}\label{lem:fixp-cross-node}
At the normalized prices, $Z_i\le q_i$ for every node and
$\sum_i Z_i\le W$.
\end{lemma}
\begin{proof}
By \Cref{lem:fixp-market-normalization}, all goods counted in
$Z_i$ are free, and $Z_i>0$ forces
$\mu_i=0$. In that case, an own low good has positive price:
its positive utility and dual feasibility give
$0<(b_i-\lambda_{it})/(b_i+1)\le\gamma_i\pi_{it}^-$.
Other nodes' agents cannot consume it. Moreover, its supply exceeds
$W$, the total demand of all edge agents together.
Thus its own node agents consume it, and every low-price inequality
at this node is an equality.

Consider first a node with $Z_i>0$. We just proved
that all its low prices are tight; together with the high price and
the supply balance, this gives total value $\pi_RS_i$ for its
own goods. All those goods have positive prices, so no other node's
agents consume them. The node-$i$ agents therefore receive $S_i-q_i$
units of own goods. Their payment for these goods is at most
$\pi_RS_i$, because the edge agents' payments are nonnegative.

The quantity constraint now determines reference consumption, allowing
us to apply the budget bound. The class consumes $B_0$ units in total.
After its own goods and
the $Z_i$ free units from other nodes, its reference consumption
is $B_0-S_i+q_i-Z_i$. Its payment for reference goods is this
quantity times $\pi_R$, and the $Z_i$ other units cost zero.
Full spending and the bound on its own-good payment therefore give
\begin{align*}
 B_0
 &\le\pi_RS_i+\pi_R(B_0-S_i+q_i-Z_i)\\
 &=\pi_RB_0+\pi_Rq_i-\pi_RZ_i.
\end{align*}
Rearranging displays the bound directly:
\[
 \pi_RZ_i\le\pi_Rq_i-B_0(1-\pi_R)\le\pi_Rq_i,
\]
where the last inequality uses $\pi_R\le1$. At this node $\mu_i=0$,
so the reference equality gives
$b_i/(b_i+1)=\gamma_i\pi_R>0$. We may therefore divide by
$\pi_R$ and obtain $Z_i\le q_i$. For a node with $Z_i=0$,
the same bound follows from $q_i\ge0$.

Summing the per-node bounds gives a bound on all consumption across
different nodes. The node goods form disjoint groups, and the $W$ edge agents
consume only $W$ units altogether. Therefore
\[
 \sum_iZ_i\le\sum_iq_i\le W.
\]
In particular, all other nodes' agents together can consume at most
$W$ units from any one low group.
\end{proof}

\subsubsection{Proof of the Market Lemma}\label{app:fixp-market-proof}
We combine the price bounds with the edge-agent equations and
the node-class balances to prove \Cref{lem:fixp-market-correctness}.

\begin{proof}[Proof of \Cref{lem:fixp-market-correctness}]
\emph{Prices and Full Spending.}
We can now turn every low-price inequality into an equality by
showing that each low group is consumed by its own agents.
Each low group supplies more than $2W$ units. Edge agents consume
at most $W$ of them, and other nodes' agents consume at most $W$
by \Cref{lem:fixp-cross-node}. A positive amount must remain for its own
node agents. Complementary slackness now gives
$\pi_{it}^-=\pi_R-\lambda_{it}/r_i$ at every node and every low
group. The total-value calculation above becomes an equality,
$N=N\pi_R$, and hence $\pi_R=1$.

It remains to locate $r_i$ in its prescribed interval. The extra
low group with parameter $a_i$ has price $1-a_i/r_i\ge0$, which
gives $r_i\ge a_i$. For the upper bound, use the reference equality
and $\gamma_i=r_i/(b_i+1)$ to obtain
\[
 \mu_i=\frac{b_i}{b_i+1}-\gamma_i
      =\frac{b_i-r_i}{b_i+1}\ge0.
\]
Thus $r_i\le b_i$, completing \Cref{eq:fixp-price-form}.
Full spending follows from \Cref{lem:fixp-market-normalization}.

Two consequences of the proof will account for the endpoint signs.
If node-$i$ agents consume other nodes' goods, then $Z_i>0$ forces
$\mu_i=0$ and therefore $r_i=b_i$. Conversely, a free low group
satisfies $r_i=\lambda_{it}$; since
$\lambda_{it}\le a_i\le r_i$, this can happen only at
$r_i=\lambda_{it}=a_i$.

\emph{Edge-Agent Demand.}
Consider an edge $e:i\to j$. Its agent values a low group at $i$
and a high group at $j$, whose prices are $1-\lambda_e/r_i$ and
$1+1/r_j$. The low price is below one and the high price is above
one, so a mixture can use exactly the unit budget. If $q_e^-$ and
$q_e^+$ are the two quantities, unit quantity and full spending require
\[
 q_e^-+q_e^+=1,\qquad
 \left(1-\frac{\lambda_e}{r_i}\right)q_e^-
 +\left(1+\frac1{r_j}\right)q_e^+=1.
\]
Subtracting the first equation from the second gives
$q_e^+/r_j=\lambda_eq_e^-/r_i$. Substitution into the quantity
equation gives
the amounts in \Cref{eq:fixp-edge-demand}.
It remains to show that optimal bundles use only the chosen low and
high groups. The groups' utility difference is $\varepsilon_w$ and their price
difference is $\lambda_e/r_i+1/r_j$. Thus the line through their
price-utility pairs has slope and intercept
\[
 \gamma_w=\frac{\varepsilon_w}{\lambda_e/r_i+1/r_j}>0,
 \qquad \mu_w=1-\gamma_w\left(1+\frac1{r_j}\right).
\]
Since $\lambda_e/r_i+1/r_j\ge1/r_j$, we have
$\gamma_w\le\varepsilon_wr_j$. Consequently
\[
 \mu_w\ge1-\varepsilon_w(r_j+1)
       \ge1-\frac{b_*+1}{2(1+b_*)}=\frac12.
\]
All other goods have utility zero for this agent and nonnegative
prices, so their dual constraints are strict. The line is therefore
a feasible demand dual. The bundle in \Cref{eq:fixp-edge-demand}
uses only goods on the line and spends one; its utility is
$\gamma_w+\mu_w$, the dual value. It is optimal, and complementary
slackness forces every optimal bundle to use only the same two
groups. The quantity and budget equations then fix their total
amounts uniquely.

\emph{The Box Conditions.}
We derive each node's graph equation by subtracting total quantity
from total spending for its class of agents. Both totals are $B_0$.
Reference consumption contributes zero because its price is one.
Own high goods contribute $1/r_i$ per unit, and own low goods in
group $t$ contribute $-\lambda_{it}/r_i$ per unit. Finally, the
$Z_i$ units received from other nodes are free, so they contribute
$-Z_i$.

To express the own-good amounts by market clearing, let $Q_i^+$
be the total high-good consumption by edge agents at node $i$.
Other nodes' agents cannot consume these positive-price goods, so
node-$i$ agents receive $H_i-Q_i^+$ units. For low group $t$,
write $Q_{it}^-$ for consumption by edge agents and $T_{it}$ for
consumption by other nodes' agents. Its own class therefore receives
$L_{it}-Q_{it}^--T_{it}$ units. Substituting these amounts into
spending minus quantity gives
\[
 0=\frac{H_i-Q_i^+}{r_i}
   -\sum_t\frac{\lambda_{it}}{r_i}
                    (L_{it}-Q_{it}^--T_{it})-Z_i.
\]
We can now identify the graph field. The fixed supply terms sum to
$(H_i-\sum_t\lambda_{it}L_{it})/r_i=0$ by
\Cref{eq:fixp-supply-bounds}. The terms involving edge agents give
\[
 \sum_t\frac{\lambda_{it}}{r_i}Q_{it}^--\frac{Q_i^+}{r_i}
 =\sum_{e:\,i\to j}I_e(r)-\sum_{e:\,j\to i}I_e(r)
 =\mathcal G_i(r),
\]
using \Cref{eq:fixp-demand-current} for each outgoing and incoming
edge. The extra low group has no edge agent and hence contributes
zero to this sum.

The remaining terms record free consumption across different nodes.
If $T_{it}>0$,
the low group is free, so $1-\lambda_{it}/r_i=0$ and
$(\lambda_{it}/r_i)T_{it}=T_{it}$. This equality also holds when
$T_{it}=0$. Setting $T_i=\sum_tT_{it}$, the balance is therefore
\[
 0=\mathcal G_i(r)+T_i-Z_i.
\]

The endpoint consequences of the price argument now give the signs.
Consumption from other nodes, measured by $Z_i$, is possible only
at $r_i=b_i$; consumption of this node's free low goods by other
nodes, measured by $T_i$, is possible only at $r_i=a_i$.
Since $a_i<b_i$, the lower endpoint has $Z_i=0$ and thus
$\mathcal G_i=-T_i\le0$. The upper endpoint has $T_i=0$ and
thus $\mathcal G_i=Z_i\ge0$. In the interior, both quantities
vanish and $\mathcal G_i=0$. These are exactly \Cref{eq:fixp-box}.

\end{proof}

\section{Technical Details of the Approximation Algorithm}\label{app:chapter3-proofs}
This appendix follows \Cref{sec:demand,sec:rational-output}.
We first give the scalar criterion, coefficient data, and interval
bounds for the utility guarantee, then prove the rational-computation
and error bounds.

\subsection{\texorpdfstring{Omitted Proofs and Data for \Cref{sec:demand}}{Omitted Proofs and Data for Section \getrefnumber{sec:demand}}}\label{app:ch3-demand}
This subsection supplies the details omitted from \Cref{sec:demand}.
The scalar criterion comes first, followed by the coefficient data
and their bounds, the interval proof of the scalar inequality, and
the resulting approximation guarantees.

\subsubsection{Proving the Scalar Demand Criterion}\label{sec:demand-geometry}
This subsection proves \Cref{prop:row-demand}. We establish the
geometry of a row's tails, the required property of the auxiliary
function $\overline N$, and the scalar implication for one row before
combining them.

The first lemma reduces an agent's bundle to the three variables
used in the scalar demand criterion and bounds the coefficient
difference in terms of those variables.

Fix coefficients from \Cref{def:field}, $\epsilon\in(0,1)$, and a normalized
nonconstant utility row $\mathbf{v}$. Let $\mathbf{x}$ be
a unit bundle assigning positive mass to utility one, and suppose
$t=\mathbf{v}\cdot \mathbf{x}+\epsilon\le1$. Put $b=\min\{v_j:x_j>0\}$ and
$L(a)=\sum_{j:v_j\ge a}x_j$. If $t<1$, set
\[
 u=\frac{1}{1-t}\int_t^1L(a)\,da;
\]
if $t=1$, set $u=0$.

\begin{lemma}\label{lemma:demand-geometry}
Then $(b,t,u)\in\Omega_\epsilon$.
For $\ell$ defined by \Cref{eq:middle-average},
\[
 \ell=\frac{1}{t-b}\int_b^tL(a)\,da,
 \qquad 0\le u\le\ell<1.
\]
Moreover,
\begin{equation}\label{eq:correction-lower}
 g(t;\mathbf{x})-g(b;\mathbf{x})\ge
 \frac{\displaystyle\int_b^t d(L(a))\,da}
      {\displaystyle\int_0^1 L(a)d(L(a))\,da}
 +(t-b)f(b)-(1-t)u\bigl(f(t)-f(b)\bigr).
\end{equation}
\end{lemma}
\begin{proof}
The tail representation of assigned utility is
$\mathbf v\cdot\mathbf x=\int_0^1L(a)\,da$, where $L$ is
nonincreasing and equals one below $b$.
If $b=1$, then $\mathbf v\cdot\mathbf x=1$ and $t>1$, contrary to
the hypothesis; thus $b<1$, and positive mass at utility one gives
$\mathbf v\cdot\mathbf x>b$.
It follows that $0\le b<1-\epsilon$ and $b+\epsilon<t\le1$.
Splitting the tail integral at $b,t$ gives
\[
 \int_b^tL(a)\,da=t-b-\epsilon-(1-t)u,
\]
including when $t=1$, since the upper integral and $(1-t)u$ then
both vanish.
Dividing by $t-b>0$ identifies this average with the algebraically
defined $\ell$ in \Cref{eq:middle-average}.

For $t<1$, monotonicity of $L$ makes the upper average $u$ no
larger than the middle average $\ell$; for $t=1$, the convention
$u=0$ gives the same inequality.
Nonnegativity of the tail gives $u\ge0$, positive mass at utility
one gives $\ell>0$, and \Cref{eq:middle-average} gives $\ell<1$
because $\epsilon>0$.
Substituting that formula into $u\le\ell$ yields
$(1-b)u\le t-b-\epsilon$.
Together with the bounds on $b,t$, this proves
$(b,t,u)\in\Omega_\epsilon$ and the asserted bounds on the averages.

It remains to bound the correction term in \Cref{eq:field}.
The monotonicity of $f$ in \Cref{lem:coefficient-shape} shows that, for an allocated good with
$v_j\in[b,t]$, its contribution to the difference of the correction terms is
\[
 (t-v_j)f(v_j)+(v_j-b)f(b)\ge(t-b)f(b).
\]
For an allocated good with $v_j\ge t$, that contribution is
\[
 (t-b)f(b)-(v_j-t)\bigl(f(t)-f(b)\bigr).
\]
These expressions agree at $v_j=t$.
Averaging over the allocated goods and using
\[
 \sum_jx_j(v_j-t)_+=\int_t^1L(a)\,da=(1-t)u
\]
gives the following lower bound on that difference:
$(t-b)f(b)-(1-t)u(f(t)-f(b))$, also when $t=1$.
Adding the integral part of \Cref{eq:field}, whose denominator is
positive by \Cref{lem:normalizer}, proves \Cref{eq:correction-lower}.
\end{proof}

\paragraph{The Auxiliary Envelope.}\label{sec:demand-envelopes}
The next lemma establishes the concavity and majorant property of
the function already defined in \Cref{eq:tail-majorant}.

We will use concavity to bound the upper-tail part of the normalizing
integral using the average tail value in \Cref{lemma:scalar-reduction}.

\begin{lemma}\label{lemma:tail-majorant}
For the weight $d$ from \Cref{def:field}, the function $\overline N$ in
\Cref{eq:tail-majorant} is concave on $[0,1]$ and satisfies
$z\,d(z)\le\overline N(z)$ for $z>0$. The product $z\,d(z)$ has
the continuous extension zero at $z=0$. The two functions agree on
$[0,T]$, and $z d(z)<\overline N(z)$ for $T<z<1$.
\end{lemma}
\begin{proof}
On $(0,T]$, the product is
$z d(z)=T^{\alpha-\beta}z^{1-\alpha}$, which is increasing and
concave because $0<1-\alpha<1$.
On $[T,1]$, it is $z^{1-\beta}$, which is decreasing and strictly
convex because $\beta>1$.
Consequently, the high branch lies below the chord joining
$(T,T^{1-\beta})$ and $(1,1)$, with strict inequality for $T<z<1$.
This chord is the high branch in \Cref{eq:tail-majorant},
so the stated majorant comparisons hold.

The two branches of $\overline N$ join continuously at $T$.
The left derivative there is positive and the chord slope is
negative, so the derivative jumps downward; together with
concavity of each branch, this proves concavity on $[0,1]$.
Finally, $1-\alpha>0$ gives $\lim_{z\downarrow0}z d(z)=0$.
With this continuous extension, the two functions agree throughout
$[0,T]$.
\end{proof}

The final auxiliary lemma connects the bound on $h$ and
the function $S$ to the coefficient difference of one row.

\begin{lemma}\label{lemma:scalar-reduction}
Fix $0<\epsilon<1$ and coefficients from \Cref{def:field}.
Suppose $(1-z)f(z)\le H$ for $z\in[0,1]$ and
$0\le H<1/(1+\epsilon)$. For a triple in the domain in \Cref{eq:demand-domain}, let
$\ell$ be given by~\Cref{eq:middle-average}, and use $\xi$ and $S$
from \Cref{eq:demand-xi} and \Cref{eq:demand-residual}.
Then $0<1/(1+\epsilon)-H\le\xi\le1/(1+\epsilon)+H$.
For every row and bundle in \Cref{lemma:demand-geometry}, with their
associated $b,t,u,\ell$,
\begin{equation}\label{eq:scalar-demand-implication}
 S(b,t,u)\ge0
 \quad\Longrightarrow\quad
 g(t;\mathbf{x})-g(b;\mathbf{x})\ge\frac{1}{1+\epsilon}.
\end{equation}
\end{lemma}
\begin{proof}
First fix an arbitrary $(b,t,u)\in\Omega_\epsilon$.
The domain constraints imply $0\le b<t\le1$, $t-b>\epsilon$, and
$0\le u\le(t-b-\epsilon)/(1-b)<1$.
They also imply
\[
 \ell=\frac{t-b-\epsilon-(1-t)u}{t-b}
 \ge\frac{t-b-\epsilon}{1-b}>0,
 \qquad
 \ell-u=\frac{t-b-\epsilon-(1-b)u}{t-b}\ge0.
\]
Since \Cref{eq:middle-average} gives $\ell<1$, the convex envelope
is evaluated within its domain for every triple in $\Omega_\epsilon$.
The expression subtracted from $1/(1+\epsilon)$ in
\Cref{eq:demand-xi} satisfies
\[
 -H\le-(1-t)f(t)
 \le(t-b)f(b)-(1-t)u\bigl(f(t)-f(b)\bigr)
 \le(1-b)f(b)\le H,
\]
using $0\le u\le1$ and the nonnegativity and monotonicity of $f$.
Subtracting these bounds from $1/(1+\epsilon)$ proves the stated
bounds on $\xi$, including strict positivity.

Now fix a row and bundle from \Cref{lemma:demand-geometry}, use
their associated $b,t,u,\ell$, and write
$E=\int_0^1L(a)d(L(a))\,da>0$.
Splitting this denominator at $b,t$ gives
\[
 \begin{split}
 \int_b^t d(L(a))\,da-\xi E
 ={}&\int_b^t\Phi_\xi(L(a))\,da
       -\xi\int_t^1L(a)d(L(a))\,da-b\xi.
 \end{split}
\]
Here $-b\xi$ is the contribution below $b$, where $L(a)=1$ and
$d(1)=1$.
The convex minorant inequality and Jensen's inequality on the
middle interval, together with the integral identity for $\ell$
in \Cref{lemma:demand-geometry}, give
\[
 \int_b^t\Phi_\xi(L(a))\,da
 \ge(t-b)\operatorname{co}\Phi_\xi(\ell).
\]

For $t<1$, the majorant and concavity established in
\Cref{lemma:tail-majorant} give
\[
 \int_t^1L(a)d(L(a))\,da
 \le\int_t^1\overline N(L(a))\,da
 \le(1-t)\overline N(u).
\]
This inequality also holds at $t=1$, when both sides vanish.
Since $\xi>0$, multiplying this upper bound by $-\xi$ and combining
it with the middle-interval lower bound yields
\[
 \int_b^t d(L(a))\,da-\xi E\ge S(b,t,u).
\]
If $S(b,t,u)\ge0$, division by $E$ bounds the integral ratio in
\Cref{eq:correction-lower} below by $\xi$.
Combining this estimate with \Cref{eq:correction-lower} and the
definition of $\xi$ proves \Cref{eq:scalar-demand-implication}.
\end{proof}

\paragraph{Completing the General Criterion.}\label{sec:all-demand}
We can now combine these auxiliary statements to obtain
\Cref{prop:row-demand}.

\begin{proof}[Proof of \Cref{prop:row-demand}]
Fix a normalized nonconstant row $\mathbf v$ and a
bundle $\mathbf x$ in the coefficient domain with $t=\mathbf v\cdot\mathbf x+\epsilon<1$, and
set $b=\min\{v_j:x_j>0\}$.
\Cref{lemma:demand-geometry} supplies a triple
$(b,t,u)\in\Omega_\epsilon$, so the global hypothesis gives
$S(b,t,u)\ge0$.
Moreover, $(1-0)f(0)=0$ implies $H\ge0$, and thus all the
hypotheses of \Cref{lemma:scalar-reduction} hold.
That lemma gives $g(t;\mathbf x)-g(b;\mathbf x)\ge1/(1+\epsilon)$.
Finally, \Cref{lem:coefficient-shape} shows that $g(\cdot;\mathbf x)$
is nondecreasing, so
$\min_{j:x_j>0}g(v_j;\mathbf x)=g(b;\mathbf x)$, proving the
required separation.
\end{proof}

\setcounter{secnumdepth}{4}
\subsubsection{Parameters and Interpolation Knots}\label{app:rational-data}
This part gives the exact data referenced by
\Cref{def:numerical-field}. The tables specify the correction function $f$ through
the knots of $h(t)=(1-t)f(t)$. Linear interpolation defines $h$;
then $f(t)=h(t)/(1-t)$ for $t<1$, with the left limit defining $f(1)$.
Every terminating decimal denotes its exact rational value.
Fractions are used when the exact value has no terminating
decimal expansion.

\paragraph{Exact Parameters and Certified Errors}\label{app:exact-rational-parameters}
Both constructions use
\begin{equation}\label{eq:numerical-choice}
 s=0.6395298,\qquad T=0.76700663.
\end{equation}
The paired certified targets are
\begin{equation}\label{eq:certified-targets}
 \epsilon=0.075521017\quad\text{for the 144-knot choice},\qquad
 \epsilon=0.076\quad\text{for the ten-knot choice}.
\end{equation}
Both choices of $f$ have the same final constant segment:
\begin{equation}\label{eq:plateau}
 f(t)=\frac{0.120575947692676}{0.124}
 \qquad\left(0.876\le t\le1\right).
\end{equation}
These exact values, together with the following knots, are the inputs
to the checks of \Cref{eq:shape} and the lower bound on $S$ in
\Cref{eq:positive-residual}.

\paragraph{Choosing the Weight Parameters}\label{app:weight-balance}
We choose the two parameters in \Cref{eq:numerical-choice} by
balancing the coefficient differences required by three utility rows.
The exponent parameter $s$ controls the decrease of $d$ on each branch;
$T$ controls the relative weight of the two branches. We explain their
values through three small utility rows. This calculation also explains
why the correction function $f$ becomes constant at high utilities.

\subparagraph{Three Rows With the Same Utility Gain.}
In this subsection, $b,t,a$ are utility levels with $0<b<a<t<1$.
The first row has lowest allocated utility zero and target $b$;
the second has lowest allocated utility $b$ and target $t$;
the third has lowest allocated utility $a$ and target one.
Let $u_0,u_1$ be the masses on utility one in the first two rows.
For a candidate error $\epsilon$, put
\[
\begin{gathered}
 z=t-b,\qquad H_0=(1-b)u_0,\qquad H_1=(1-t)u_1,\\
 L_0=1-\frac{\epsilon+H_0}{b},\qquad
 L_1=1-\frac{\epsilon+H_1}{z},\qquad
 M=1-\frac{\epsilon}{1-a}.
\end{gathered}
\]
Here $L_0,L_1$ are the masses above the lowest utility in their
respective rows, and $M$ is the mass on utility one in the third row.
We require $0<u_0<L_0<1$, $0<u_1<L_1<1$, and $0<M<1$.
The three allocations are listed in \Cref{tab:design-rows}.
The second and third rows also have a zero-utility good with zero mass,
so all rows are normalized as in \Cref{def:field}.
\begin{table}[htbp]
\centering
\caption{Rows Used to Balance the Parameters}\label{tab:design-rows}
\begin{tabular}{clll}
\toprule
Row & Allocated Utilities & Corresponding Masses & Target\\
\midrule
$0$ & $(0,b,1)$ & $(1-L_0,L_0-u_0,u_0)$ & $b$\\
$1$ & $(b,t,1)$ & $(1-L_1,L_1-u_1,u_1)$ & $t$\\
$2$ & $(a,1)$ & $(1-M,M)$ & $1$\\
\bottomrule
\end{tabular}
\end{table}

Their current utilities are, respectively,
$bL_0+H_0=b-\epsilon$, $b+zL_1+H_1=t-\epsilon$, and
$a+(1-a)M=1-\epsilon$. Each target therefore improves utility by
exactly $\epsilon$.

We can evaluate the coefficient difference between each target and
its lowest allocated utility directly from \Cref{eq:field}.
The integral term contributes
\begin{equation}\label{eq:design-integral-gaps}
\begin{aligned}
 A_0&=\frac{b\,d(L_0)}{bL_0d(L_0)+H_0d(u_0)},\\
 A_1&=\frac{z\,d(L_1)}{b+zL_1d(L_1)+H_1d(u_1)},\\
 A_2&=\frac{(1-a)d(M)}{a+(1-a)Md(M)}.
\end{aligned}
\end{equation}
For example, the first row has tail $L_0$ on $(0,b]$ and $u_0$
on $(b,1]$, giving the numerator and denominator of $A_0$.
The constant terms $b,a$ in the other denominators come from
tail one below the lowest allocated utility and $d(1)=1$.
Including the correction term gives the exact differences
\begin{equation}\label{eq:design-corrected-gaps}
\begin{aligned}
 P_0&=A_0-H_0f(b),\\
 P_1&=A_1+(z+H_1)f(b)-H_1f(t),\\
 P_2&=A_2+(1-a)f(a).
\end{aligned}
\end{equation}
We seek $P_i\ge c$, where $c=1/(1+\epsilon)$ is the required
coefficient difference in \Cref{eq:target-separation}.
Increasing $f(b)$ helps the second row but hurts the first.
Increasing $f$ at high utilities helps the third row but hurts
the second. These opposing effects determine the useful scale of $f$.

\subparagraph{Eliminating the Correction Function From the Balance.}
Choose positive weights $\omega_0,\omega_1,\omega_2$ proportional to
$((z+H_1)/H_0,1,H_1/(1-a))$, normalized to sum to one.
Then $\omega_0H_0=\omega_1(z+H_1)$ and
$\omega_2(1-a)=\omega_1H_1$. Multiplying
\Cref{eq:design-corrected-gaps} by these weights cancels $f(b)$:
\begin{equation}\label{eq:design-weighted-balance}
 \sum_{i=0}^2\omega_iP_i
 =W+\omega_1H_1\bigl(f(a)-f(t)\bigr)\le W,
 \qquad W:=\sum_{i=0}^2\omega_iA_i.
\end{equation}
The inequality follows from $a<t$ and the monotonicity of $f$.
Consequently $P_i\ge c$ for all three rows requires $W\ge c$.
If $W=c$, all three differences must equal $c$, and $f(a)=f(t)$.
Since $f$ is nondecreasing, it must then be constant throughout $[a,t]$.
This is the source of the high-utility plateau in the construction.

By \Cref{eq:design-integral-gaps}, $W$ is an explicit function of
$(\epsilon,b,t,u_0,u_1,a,s,T)$, with no unknown values of $f$.
The rows with small $W$ constrain the choice of $s,T$.
To locate a balance, we seek a stationary value with respect to the
five row parameters and the two weight parameters, at which $W=c$:
\begin{equation}\label{eq:design-stationary-system}
 W=\frac1{1+\epsilon},\qquad
 \frac{\partial W}{\partial r}=0
 \quad\text{for }r\in\{b,t,u_0,u_1,a,s,T\}.
\end{equation}
Each derivative holds all the other listed variables, including
$\epsilon$, fixed. These are eight equations in eight unknowns.

Solving \Cref{eq:design-stationary-system} gives the numerical values
in \Cref{tab:design-balance}. At this solution
$L_0\approx0.847715>T$, whereas
$L_1\approx0.437891$, $M\approx0.390610$, $u_0$, and $u_1$ are below
$T$. Thus the calculation compares the two branches of $d$.
\begin{table}[!htb]
\centering
\caption{Solution of the Parameter Equations}\label{tab:design-balance}
\begin{tabular}{clcl}
\toprule
Parameter & Approximate Value & Parameter & Approximate Value\\
\midrule
$\epsilon$ & $0.0755210166997678$ & $a$ & $0.876071160297939$\\
$b$ & $0.809603575297469$ & $t$ & $0.947150293865684$\\
$u_0$ & $0.250895058941684$ & $u_1$ & $0.0339681997871910$\\
$s$ & $0.639529797929822$ & $T$ & $0.767006632871314$\\
\bottomrule
\end{tabular}
\end{table}

Rounding this $s$ to seven decimal places and $T$ to eight decimal
places gives exactly \Cref{eq:numerical-choice}. The certified target
$0.075521017$ is about $3.00232\cdot10^{-10}$ above the displayed
local value of $\epsilon$.

The values of $f$ at the balance can also be recovered from
\Cref{eq:design-corrected-gaps}:
\[
 f(b)=\frac{A_0-c}{H_0}\approx0.765469857997207,\qquad
 f(a)=f(t)=\frac{c-A_2}{1-a}\approx0.972386827698116.
\]
The middle equality $P_1=c$ then follows from
\Cref{eq:design-weighted-balance}. In particular,
$h(b)=(1-b)f(b)\approx0.145742724180223$.
These values explain the scale of the central part and the final
segment of $h$. The construction uses the nearby
plateau value $f(t)=0.972386674940935\ldots$ in \Cref{eq:plateau}, together
with rounded $s,T$ and the slightly larger target error. The knot
values of $h$ must satisfy the full set of inequalities described next.

\paragraph{Constructing and Compressing the Interpolation}\label{app:knot-design}
We construct $h$ by choosing knot heights that satisfy the shape
conditions and the three-variable demand inequality. For fixed knot
locations, these requirements give linear constraints on the heights.

\subparagraph{Shape Constraints on Knot Heights.}
Choose utility locations $0=r_0<r_1<\cdots<r_m=1$ and write
$y_i=h(r_i)$, so that there are $m+1$ knots and $m$ affine segments.
Linear interpolation gives, on $[r_i,r_{i+1}]$,
\[
 h(r)=\frac{r_{i+1}-r}{r_{i+1}-r_i}y_i
       +\frac{r-r_i}{r_{i+1}-r_i}y_{i+1}.
\]
For fixed locations, the following linear constraints impose the
required shape and the bounds $A=0.24$, $H=0.149589$:
\begin{equation}\label{eq:design-knot-constraints}
\begin{gathered}
 y_0=y_m=0,\qquad 0\le y_i\le\min\{Ar_i,H\},\\
 \frac{y_i-y_{i-1}}{r_i-r_{i-1}}
 \ge\frac{y_{i+1}-y_i}{r_{i+1}-r_i}\quad(1\le i<m),
 \qquad y_{m-1}<1-r_{m-1}.
\end{gathered}
\end{equation}
The slope inequalities give concavity, and the last inequality gives
$f(1)=y_{m-1}/(1-r_{m-1})<1$. Linear interpolation preserves the
two affine upper bounds between knots. To impose a plateau of
height $q$ for $f$ starting at a knot $r_j$, add
$y_i=q(1-r_i)$ for $i\ge j$; this is linear in $y_i,q$.
The chosen $s,T$ separately satisfy the two weight conditions in
\Cref{eq:shape}, as checked in \Cref{prop:fixed-structure}.

\subparagraph{Demand Constraints on the Same Heights.}
Fix $s,T,\epsilon$, and return to a triple $(b,t,u)$ in the domain in
\Cref{eq:demand-domain}, where $b$ is the lowest allocated utility
and $t$ is the target utility of one row.
Substituting $f(r)=h(r)/(1-r)$ into \Cref{eq:demand-xi} gives
\begin{equation}\label{eq:design-xi-in-h}
 \xi=\frac1{1+\epsilon}
 -\frac{t-b+(1-t)u}{1-b}\,h(b)+u\,h(t).
\end{equation}
At $t=1$, the product $(1-t)f(t)$ is interpreted as $h(1)=0$;
the same formula therefore applies. Both $h(b)$ and $h(t)$ are
linear combinations of the knot heights. Thus $\xi$ is affine in
those heights.

The convex envelope in $S$ can also give linear constraints for a
finite search. Let $\ell$ be the average in \Cref{eq:middle-average}.
Choose $0<v_-\le\ell\le v_+\le1$ and $\theta\in[0,1]$ with
$\theta v_-+(1-\theta)v_+=\ell$. When $v_-<v_+$,
$\theta=(v_+-\ell)/(v_+-v_-)$; the one-point choice is
$v_-=v_+=\ell$. The corresponding chord is an upper bound on
the lower convex envelope at $\ell$. Define its numerator and
denominator contributions by
\[
\begin{aligned}
 U&=(t-b)\bigl(\theta d(v_-)+(1-\theta)d(v_+)\bigr),\\
 V&=(t-b)\bigl(\theta v_-d(v_-)+(1-\theta)v_+d(v_+)\bigr)
       +(1-t)\overline N(u)+b.
\end{aligned}
\]
Here $V>0$. With the triple, chord endpoints, and weight parameters
fixed, $U,V$ are constants. By \Cref{eq:demand-residual}, $S\le U-\xi V$.
Consequently a desired lower bound $S\ge\rho$, with $\rho>0$,
requires $U-\xi V\ge\rho$. Using \Cref{eq:design-xi-in-h}, this becomes
\begin{equation}\label{eq:design-linear-demand}
 \frac{t-b+(1-t)u}{1-b}\,h(b)-u\,h(t)
 \ge\frac1{1+\epsilon}-\frac{U-\rho}{V}.
\end{equation}
This is a linear inequality in the knot heights. Increasing $h(b)$
raises its left-hand side, whereas increasing $h(t)$ lowers it.

Thus, for fixed $s,T,\epsilon$ and knot locations, a finite collection
of inequalities of the form in \Cref{eq:design-linear-demand}, together with
\Cref{eq:design-knot-constraints}, gives a linear feasibility problem
for candidate heights. A search can alternate between solving this
problem and looking for triples and chords that violate it. Such
triples supply further linear constraints; adding a knot near a
violating utility allows $h$ to change there without
changing a long adjacent segment. The search can then decrease
$\epsilon$ and repeat. The strict endpoint inequality can be
implemented with a chosen positive margin.
After choosing the heights, we verify the demand inequality
on its full domain using the bounds in \Cref{app:artifact}.

\subparagraph{The 144-Knot Interpolation.}
The list in \Cref{tab:all-knots} has a clear concentration:
130 locations lie strictly between $0.714$ and $0.876$.
Fourteen of these lie in $[0.8052,0.8148]$, around the balancing
value $b\approx0.809604$ from \Cref{tab:design-balance}.
The left part is nearly a straight line with slope $0.202578$;
the curve reaches its largest height near $0.766$, then decreases
and joins the final line at $0.876$.
The six interior points after $0.876$ lie exactly on that line,
so omitting them leaves $h$ unchanged.
Some heights come from interpolation.
For example, the point at $0.659$ lies exactly on the segment
joining the points at $0.23$ and $0.69$:
\[
 h(0.659)=\frac{31}{460}h(0.23)+\frac{429}{460}h(0.69)
 =\frac{30704961354799481}{230000000000000000} \approx 0.133499831977389.
\]
Likewise, every knot after $0.876$ is fixed by
$h(r)=(1-r)f(1)$ and \Cref{eq:plateau}.

\subparagraph{The Ten-Knot Compression.}
Number the refined knots from $0$ to $143$. The ten-knot list is
obtained by choosing exactly the indices
\[
 0,1,2,3,4,5,45,97,136,143
\]
and interpolating between those points at their existing heights.
Their utility locations are
$0,0.02,0.23,0.659,0.69,0.714,0.767,0.835,0.876,1$.
These locations preserve the initial almost-linear part, a point
near the maximum, a point on the decreasing part, and the start
and end of the final line. \Cref{fig:correction-comparison} compares
the two interpolations.

Concavity puts every chord below the refined curve. Hence, writing
$h_{10}$ and $h_{144}$ for the two interpolations, we have
$h_{10}(r)\le h_{144}(r)$ on $[0,1]$. The largest difference is
about $0.00267063$, attained at $r=0.802$; it can be found by
checking the refined knots because the difference is affine
between consecutive refined knots. The ten-knot interpolation is
concave and satisfies the same upper bounds and endpoint conditions.
In \Cref{eq:design-linear-demand}, lowering $h(b)$ decreases the
left-hand side, whereas lowering $h(t)$ increases it. We therefore
verify the demand inequality separately for each interpolation:
$\epsilon=0.076$ for the ten-knot choice and
$\epsilon=0.075521017$ for the refined curve; see
\Cref{eq:certified-targets,prop:fixed-certificate}.

\begin{figure}[!t]
\centering
\begin{minipage}[b]{.49\textwidth}
\centering
\begin{tikzpicture}
\begin{axis}[width=.88\linewidth,height=4cm,xmin=0,xmax=1,ymin=0,ymax=.17,
 xlabel={$r$},ylabel={$h(r)$},tick label style={font=\small},
 legend style={font=\small,at={(.5,1.03)},anchor=south,legend columns=2},
 scaled y ticks=false,ytick={0,.05,.1,.15},yticklabels={0,0.05,0.10,0.15}]
\addplot[draw=.,thick] coordinates {(0,0) (0.02,0.00405229188899) (0.23,0.0465937369886) (0.659,0.133499831977) (0.69,0.139779759587) (0.714,0.144641639026) (0.715,0.144843223502) (0.716,0.145040532732) (0.717,0.145233041163) (0.718,0.145421781551) (0.719,0.145601823132) (0.72,0.145775797106) (0.721,0.145944650556) (0.722,0.146106979952) (0.723,0.146262077251) (0.724,0.146406906845) (0.725,0.146547133572) (0.726,0.146685920735) (0.727,0.146820488574) (0.728,0.146953339969) (0.729,0.147081893741) (0.73,0.147208963205) (0.731,0.147331551667) (0.732,0.147452632621) (0.733,0.147569156035) (0.734,0.147684153558) (0.735,0.147794524323) (0.736,0.147903343179) (0.737,0.148008268111) (0.738,0.148109242428) (0.739,0.148208630821) (0.74,0.148303837229) (0.741,0.148395616714) (0.742,0.148484964195) (0.743,0.148570856274) (0.744,0.148652423937) (0.745,0.148732711114) (0.746,0.148808100649) (0.752,0.149237437227) (0.755,0.149452105516) (0.756,0.149466093473) (0.763,0.149564009168) (0.764,0.149574922729) (0.765,0.149584071397) (0.766,0.149588137098) (0.767,0.14958789019) (0.768,0.149585254918) (0.769,0.149576157047) (0.77,0.149567059176) (0.771,0.149551048432) (0.772,0.149533095086) (0.773,0.149509197902) (0.774,0.149483311709) (0.775,0.149451429524) (0.776,0.149418027039) (0.777,0.149384624552) (0.778,0.149337444701) (0.78,0.149237384187) (0.781,0.149180862219) (0.782,0.149122232447) (0.783,0.149057596277) (0.784,0.148990569992) (0.785,0.148918395961) (0.786,0.148841070669) (0.787,0.148761356872) (0.788,0.148675268614) (0.789,0.148587648564) (0.79,0.148493345347) (0.791,0.148394194044) (0.792,0.148292952063) (0.793,0.148184875789) (0.794,0.148074740509) (0.795,0.147958918819) (0.796,0.147837534639) (0.797,0.14771389556) (0.798,0.147583468715) (0.799,0.147450752508) (0.8,0.147311168258) (0.801,0.147169288917) (0.802,0.147020467122) (0.803,0.146869789948) (0.8052,0.146538300166) (0.8058,0.146447893861) (0.8059,0.146431607959) (0.806,0.146413006827) (0.8066,0.146301400037) (0.8093,0.145799169479) (0.8096,0.145743366083) (0.8097,0.145724764952) (0.8099,0.145687562688) (0.8108,0.145519412605) (0.8112,0.145444144613) (0.812,0.145293447303) (0.8129,0.145123912829) (0.8148,0.144766006717) (0.822,0.143408279546) (0.827,0.141871797621) (0.835,0.139411104335) (0.836,0.139091404647) (0.837,0.138763284107) (0.838,0.138435163566) (0.839,0.138099573422) (0.84,0.137760718379) (0.841,0.137413799549) (0.842,0.137064623061) (0.843,0.136706153727) (0.844,0.136345405423) (0.845,0.135976287363) (0.846,0.135604370574) (0.847,0.135224025974) (0.848,0.134840872033) (0.849,0.134449216904) (0.85,0.134054721066) (0.851,0.133651676932) (0.852,0.13324576596) (0.853,0.132832688185) (0.854,0.132412392234) (0.855,0.13198920387) (0.857,0.131129867998) (0.859,0.130230940792) (0.86,0.129778557526) (0.861,0.129317302088) (0.862,0.128853087999) (0.863,0.128379947756) (0.864,0.127903823792) (0.865,0.127418931936) (0.866,0.126930406946) (0.867,0.126434984961) (0.868,0.125931784174) (0.869,0.125427501582) (0.87,0.12492192585) (0.871,0.124373649849) (0.872,0.123779448511) (0.873,0.123090889825) (0.874,0.122348641991) (0.875,0.121539970553) (0.876,0.120575947693) (0.877,0.119603561018) (0.879,0.117658787668) (0.893,0.104045374219) (0.933,0.065149907221) (0.966,0.033061146948) (0.99,0.00972386674941) (1,0)};
\addlegendentry{$h_{144}$}
\addplot[draw=.,dashed,mark=*,mark size=1.1pt] coordinates {(0,0) (0.02,0.00405229188899) (0.23,0.0465937369886) (0.659,0.133499831977) (0.69,0.139779759587) (0.714,0.144641639026) (0.767,0.14958789019) (0.835,0.139411104335) (0.876,0.120575947693) (1,0)};
\addlegendentry{$h_{10}$}
\end{axis}
\end{tikzpicture}
\end{minipage}\hfill
\begin{minipage}[b]{.49\textwidth}
\centering
\begin{tikzpicture}
\begin{axis}[width=.88\linewidth,height=4cm,xmin=.70,xmax=.88,ymin=0,ymax=.003,
 xlabel={$r$},ylabel={$h_{144}(r)-h_{10}(r)$},tick label style={font=\small},
 scaled y ticks=false,ytick={0,.001,.002,.003},xtick={.7,.75,.8,.85},
 yticklabels={0,0.001,0.002,0.003}]
\addplot[draw=.,thick] coordinates {(0,0) (0.02,0) (0.23,0) (0.659,0) (0.69,0) (0.714,0) (0.715,0.000108258982149) (0.716,0.000212242718687) (0.717,0.00031142565537) (0.718,0.000406840550259) (0.719,0.00049355663747) (0.72,0.000574205117309) (0.721,0.000649733073844) (0.722,0.00071873697637) (0.723,0.000780508781393) (0.724,0.000832012882184) (0.725,0.000878914115732) (0.726,0.000924375784737) (0.727,0.000965618130435) (0.728,0.00100514403114) (0.729,0.00104037230971) (0.73,0.0010741162797) (0.731,0.00110337924829) (0.732,0.00113113470875) (0.733,0.00115433262902) (0.734,0.00117600465874) (0.735,0.00119304992966) (0.736,0.00120854329242) (0.737,0.00122014273046) (0.738,0.00122779155344) (0.739,0.00123385445348) (0.74,0.00123573536708) (0.741,0.00123418935883) (0.742,0.00123021134592) (0.743,0.0012227779315) (0.744,0.00121102010032) (0.745,0.00119798178431) (0.746,0.00118004582518) (0.752,0.00104942944153) (0.755,0.000984121249208) (0.756,0.00090478371246) (0.763,0.000349420952224) (0.764,0.000267009019844) (0.765,0.000182832194344) (0.766,9.35724014128e-05) (0.767,0) (0.768,0.000147023343123) (0.769,0.000287584087672) (0.77,0.000428144832219) (0.771,0.0005617927037) (0.772,0.000693497972921) (0.773,0.000819259404653) (0.774,0.000943031827041) (0.775,0.00106080825773) (0.776,0.00117706438799) (0.777,0.00129332051725) (0.778,0.00139579928151) (0.78,0.0015950559982) (0.781,0.0016881926458) (0.782,0.00177922148928) (0.783,0.00186424393481) (0.784,0.00194687626566) (0.785,0.00202436084957) (0.786,0.00209669417345) (0.787,0.00216663899169) (0.788,0.0022302093494) (0.789,0.00229224791543) (0.79,0.00234760331374) (0.791,0.00239811062609) (0.792,0.00244652726073) (0.793,0.00248810960213) (0.794,0.00252763293756) (0.795,0.0025614698635) (0.796,0.00258974429884) (0.797,0.00261576383564) (0.798,0.00263499560629) (0.799,0.0026519380146) (0.8,0.00266201237988) (0.801,0.0026697916547) (0.802,0.00267062847475) (0.803,0.00266960991658) (0.8052,0.0026673690882) (0.8058,0.00266675795309) (0.8059,0.00266543791244) (0.806,0.00266180264222) (0.8066,0.0026399910209) (0.8093,0.00254183872496) (0.8096,0.0025309329143) (0.8097,0.00252729764408) (0.8099,0.00252002710364) (0.8108,0.0024865697746) (0.8112,0.00247116522911) (0.812,0.00244019481136) (0.8129,0.00240535309127) (0.8148,0.00233179834875) (0.822,0.00205161320963) (0.827,0.00126342436163) (0.835,0) (0.836,0.000139694376613) (0.837,0.000270967901026) (0.838,0.000402241424438) (0.839,0.000526045344813) (0.84,0.000646584366118) (0.841,0.00075905960054) (0.842,0.000869277177361) (0.843,0.000970201907663) (0.844,0.00106884766787) (0.845,0.00115912367306) (0.846,0.00124660094804) (0.847,0.00132565041266) (0.848,0.00140189053583) (0.849,0.00146962947092) (0.85,0.00153452769782) (0.851,0.00159087762856) (0.852,0.00164436072049) (0.853,0.00169067701045) (0.854,0.00172977512322) (0.855,0.00176598082361) (0.857,0.00182543308043) (0.859,0.0018452940038) (0.86,0.00185230480196) (0.861,0.0018504434284) (0.862,0.00184562340391) (0.863,0.00183187722544) (0.864,0.00181514732646) (0.865,0.00178964953418) (0.866,0.00176051860926) (0.867,0.00172449068864) (0.868,0.00168068396577) (0.869,0.00163579543871) (0.87,0.00158961377035) (0.871,0.00150073183442) (0.872,0.00136592456033) (0.873,0.00113675993869) (0.874,0.000853906169088) (0.875,0.000504628796002) (0.876,0) (0.877,0) (0.879,0) (0.893,0) (0.933,0) (0.966,0) (0.99,0) (1,0)};
\end{axis}
\end{tikzpicture}
\end{minipage}
\caption{Refined and Ten-Knot Interpolations}\label{fig:correction-comparison}
\end{figure}

\paragraph{The Ten-Knot Interpolation}\label{app:sparse}
The following table gives all knots of $h$ for the
$0.076$ guarantee.
\begin{table}[!htbp]
\centering\normalsize
\caption{Knots of the Ten-Knot Interpolation}\label{tab:ten-knots}
\begin{tabular}{rr}
\toprule
$t$ & $h(t)$\\
\midrule
$0$ & $0$\\
$0.02$ & $0.00405229188899$\\
$0.23$ & $0.046593736988566$\\
$0.659$ & $\frac{30704961354799481}{230000000000000000}$\\
$0.69$ & $0.139779759587304$\\
$0.714$ & $0.144641639026455$\\
$0.767$ & $0.149587890190422$\\
$0.835$ & $0.139411104334837$\\
$0.876$ & $0.120575947692676$\\
$1$ & $0$\\
\bottomrule
\end{tabular}
\end{table}

\paragraph{The Refined Interpolation}\label{app:data}
The refined interpolation of $h$ uses the 144 knots in
\Cref{tab:all-knots}. Their locations are strictly increasing,
and their consecutive slopes are nonincreasing. The final segment
gives the plateau in \Cref{eq:plateau}. The full shape conditions
are proved in \Cref{prop:fixed-structure};
\Cref{app:artifact} verifies the scalar demand inequality.

\begingroup\normalsize
\begin{table}[p]
\centering
\caption{Exact Knots of the Refined Correction}
\label{tab:all-knots}
\small
\setlength{\tabcolsep}{4pt}
\renewcommand{\arraystretch}{1.0}
\begin{tabular*}{\textwidth}{
  @{\extracolsep{\fill}}
  rr@{\hspace{1.5em}}
  rr@{\hspace{1.5em}}
  rr@{}
}
\toprule
$t$ & $h(t)$ & $t$ & $h(t)$ & $t$ & $h(t)$ \\
\cmidrule(lr){1-2}\cmidrule(lr){3-4}\cmidrule(lr){5-6}
$0$ & $0$ & $0.77$ & $0.149567059176071$ & $0.827$ & $0.141871797620654$ \\
$0.02$ & $0.00405229188899$ & $0.771$ & $0.149551048432029$ & $0.835$ & $0.139411104334837$ \\
$0.23$ & $0.046593736988566$ & $0.772$ & $0.149533095085726$ & $0.836$ & $0.139091404647007$ \\
$0.659$ & $\frac{30704961354799481}{230000000000000000}$ & $0.773$ & $0.149509197901935$ & $0.837$ & $0.138763284106977$ \\[2pt]
$0.69$ & $0.139779759587304$ & $0.774$ & $0.1494833117088$ & $0.838$ & $0.138435163565946$ \\
$0.714$ & $0.144641639026455$ & $0.775$ & $0.149451429523963$ & $0.839$ & $0.138099573421878$ \\
$0.715$ & $0.144843223502264$ & $0.776$ & $0.149418027038703$ & $0.84$ & $0.13776071837874$ \\
$0.716$ & $0.145040532732462$ & $0.777$ & $0.149384624552443$ & $0.841$ & $0.137413799548719$ \\
$0.717$ & $0.145233041162804$ & $0.778$ & $0.149337444701178$ & $0.842$ & $0.137064623061097$ \\
$0.718$ & $0.145421781551353$ & $0.78$ & $0.149237384186823$ & $0.843$ & $0.136706153726956$ \\
$0.719$ & $0.145601823132224$ & $0.781$ & $0.149180862218892$ & $0.844$ & $0.136345405422716$ \\
$0.72$ & $0.145775797105723$ & $0.782$ & $0.14912223244685$ & $0.845$ & $0.135976287363466$ \\
$0.721$ & $0.145944650555917$ & $0.783$ & $0.149057596276862$ & $0.846$ & $0.135604370574007$ \\
$0.722$ & $0.146106979952103$ & $0.784$ & $0.148990569992186$ & $0.847$ & $0.135224025974184$ \\
$0.723$ & $0.146262077250786$ & $0.785$ & $0.14891839596057$ & $0.848$ & $0.134840872032905$ \\
$0.724$ & $0.146406906845237$ & $0.786$ & $0.148841070668933$ & $0.849$ & $0.134449216903553$ \\
$0.725$ & $0.146547133572444$ & $0.787$ & $0.148761356871645$ & $0.85$ & $0.134054721066014$ \\
$0.726$ & $0.146685920735109$ & $0.788$ & $0.148675268613835$ & $0.851$ & $0.133651676932305$ \\
$0.727$ & $0.146820488574467$ & $0.789$ & $0.148587648564335$ & $0.852$ & $0.133245765959799$ \\
$0.728$ & $0.146953339968832$ & $0.79$ & $0.148493345347122$ & $0.853$ & $0.132832688185312$ \\
$0.729$ & $0.147081893741065$ & $0.791$ & $0.14839419404395$ & $0.854$ & $0.132412392233644$ \\
$0.73$ & $0.147208963204715$ & $0.792$ & $0.148292952063067$ & $0.855$ & $0.131989203869585$ \\
$0.731$ & $0.14733155166696$ & $0.793$ & $0.148184875788948$ & $0.857$ & $0.131129867997526$ \\
$0.732$ & $0.14745263262108$ & $0.794$ & $0.148074740508856$ & $0.859$ & $0.13023094079201$ \\
$0.733$ & $0.147569156035012$ & $0.795$ & $0.147958918819274$ & $0.86$ & $0.129778557525725$ \\
$0.734$ & $0.147684153558392$ & $0.796$ & $0.147837534639088$ & $0.861$ & $0.129317302087724$ \\
$0.735$ & $0.147794524322971$ & $0.797$ & $0.147713895560365$ & $0.862$ & $0.128853087998788$ \\
$0.736$ & $0.147903343179391$ & $0.798$ & $0.14758346871549$ & $0.863$ & $0.128379947755877$ \\
$0.737$ & $0.148008268111085$ & $0.799$ & $0.147450752508274$ & $0.864$ & $0.127903823792447$ \\
$0.738$ & $0.148109242427732$ & $0.8$ & $0.147311168258037$ & $0.865$ & $0.12741893193573$ \\
$0.739$ & $0.148208630821428$ & $0.801$ & $0.147169288917326$ & $0.866$ & $0.126930406946363$ \\
$0.74$ & $0.148303837228685$ & $0.802$ & $0.147020467121859$ & $0.867$ & $0.1264349849613$ \\
$0.741$ & $0.148395616714098$ & $0.803$ & $0.146869789948162$ & $0.868$ & $0.125931784173994$ \\
$0.742$ & $0.148484964194845$ & $0.8052$ & $0.14653830016563$ & $0.869$ & $0.125427501582484$ \\
$0.743$ & $0.148570856274084$ & $0.8058$ & $0.146447893861212$ & $0.87$ & $0.12492192584968$ \\
$0.744$ & $0.148652423936566$ & $0.8059$ & $0.146431607959005$ & $0.871$ & $0.124373649849311$ \\
$0.745$ & $0.148732711114214$ & $0.806$ & $\frac{1024891047790629}{7000000000000000}$ & $0.872$ & $0.123779448510773$ \\[2pt]
$0.746$ & $0.148808100648749$ & $0.8066$ & $0.146301400036599$ & $0.873$ & $0.123090889824698$ \\
$0.752$ & $0.149237437227058$ & $0.8093$ & $0.1457991694787451$ & $0.874$ & $0.12234864199065$ \\
$0.755$ & $0.149452105515713$ & $0.8096$ & $0.145743366083428$ & $0.875$ & $0.121539970553121$ \\
$0.756$ & $0.149466093472625$ & $0.8097$ & $\frac{218587147427483}{1500000000000000}$ & $0.876$ & $0.120575947692676$ \\[2pt]
$0.763$ & $0.149564009168007$ & $0.8099$ & $0.14568756268811$ & $0.877$ & $\frac{3707710391549787}{31000000000000000}
$ \\[2pt]
$0.764$ & $0.149574922729287$ & $0.8108$ & $0.145519412605103$ & $0.879$ & $\frac{3647422417703449}{31000000000000000}$ \\[2pt]
$0.765$ & $0.149584071397446$ & $0.8112$ & $0.1454441446134$ & $0.893$ & $\frac{3225406600779083}{31000000000000000}$ \\[2pt]
$0.766$ & $0.149588137098175$ & $0.812$ & $0.145293447303234$ & $0.933$ & $\frac{2019647123852323}{31000000000000000}$ \\[2pt]
$0.767$ & $0.149587890190422$ & $0.8129$ & $0.145123912829173$ & $0.966$ & $\frac{512447777693873}{15500000000000000}$ \\[2pt]
$0.768$ & $0.149585254918022$ & $0.8148$ & $0.144766006717155$ & $0.99$ & $\frac{30143986923169}{3100000000000000}$ \\[2pt]
$0.769$ & $0.149576157047047$ & $0.822$ & $0.143408279546272$ & $1$ & $0$ \\
\bottomrule
\end{tabular*}
\end{table}
\endgroup
\setcounter{secnumdepth}{3}

\subsubsection{Verifying the Coefficient Data}\label{sec:choices}
This subsection proves \Cref{prop:fixed-structure}.
The objects to be checked are the rational parameters and knot
interpolations in \Cref{def:numerical-field}, with their complete
data in \Cref{app:rational-data}.

\begin{proof}[Proof of \Cref{prop:fixed-structure}]
The exact data in \Cref{app:rational-data} give $0<s,T<1$ and
strictly increasing knot locations, with endpoints $(0,0)$ and
$(1,0)$.
The consecutive slopes of the 144-knot interpolation are
nonincreasing, so $h$ is continuous, piecewise affine, and concave;
its zero endpoints then imply $h\ge0$.
The bounds $h(t)\le At$ and $h(t)\le H$ hold at every knot and
therefore on every intervening segment.
The ten-knot interpolation uses a subsequence of these knots:
its consecutive secant slopes are weighted averages of consecutive
groups of the original slopes and remain nonincreasing.
Its knot values obey the same linear upper bounds, which again
extend to each segment.

Both choices share the final segment, so \Cref{eq:plateau} gives
\[
 f(t)=\frac{0.120575947692676}{0.124}\approx0.972387<1
 \qquad(0.876\le t\le1),
\]
including the prescribed finite left limit at one.
Nonnegativity of $h$ gives $f\ge0$, and concavity of $h$ with
$h(1)=0$ makes $f(t)=h(t)/(1-t)$ nondecreasing.
Thus the plateau value bounds $f$ from above on the whole interval,
verifying all the qualitative conditions in \Cref{def:field}.

For each of the two exact targets, the remaining rational
comparisons give
\[
 \begin{gathered}
 \frac13<\alpha<\frac12,\quad 6<\beta<7,\quad
 2T-1\ge T^6,\quad T^5>A,\quad T^6>\frac15,\\
 \frac1{1+\epsilon}+H<1.3,
 \qquad H<\frac1{1+\epsilon}.
 \end{gathered}
\]
Since $0<T<1$ and $6<\beta<7$, these imply
$2T-1\ge T^\beta$, $A<T^{\beta-2}$, and $T^{1-\beta}<5$.
The first two complete \Cref{eq:shape}, while the last permits
$E_{\max}=5$ in \Cref{eq:computed-top-mass}.
Likewise, $\alpha>1/3$ permits $k_{\rm mass}=3$, and the curvature
condition follows from
\[
 \left(\frac1{1+\epsilon}+H\right)(\beta-1)
 <1.3(\beta-1)<\beta+1.
\]
The knot slopes, knot bounds, and rational comparisons displayed
above are verified by integer cross multiplication; these checks
establish the structural claims and constants of this proposition.
\end{proof}

\setcounter{secnumdepth}{4}
\subsubsection{Bounds for the Scalar Demand Inequality}\label{app:artifact}
This part proves $S(b,t,u)\ge10^{-10}$ for the two choices in
\Cref{def:numerical-field}, proving \Cref{prop:fixed-certificate}.
We give lower bounds near the singular boundary and on each box of
a finite subdivision, prove that these regions cover the domain,
and combine the resulting inequalities into the global bound.

\paragraph{Envelope and Local Bounds}\label{app:verification-overview}
We derive lower bounds for the scalar demand inequality near a
small middle tail and on compact boxes. Fix coefficients from \Cref{def:field} and a target $0<\epsilon<1$, with
$h(t)\le H$ as in \Cref{lemma:scalar-reduction}, and set
$\xi_{\max}=1/(1+\epsilon)+H$.
Throughout this subsection, $S$ is the function defined in
\Cref{eq:demand-residual}, $\ell$ is the middle-tail average in
\Cref{eq:middle-average}, and $\xi$ is the contribution required from the integral term
after subtracting the bound on the difference of the correction terms
in \Cref{eq:demand-xi}. The convex envelope is that of
$\Phi_\xi(z)=d(z)-\xi z d(z)$ on $(0,1]$, and
$\overline N$ is the concave majorant in \Cref{eq:tail-majorant}.

\subparagraph{Lower Bounds for the Convex Envelope.}
We evaluate tangents at specified points on the two smooth
branches of $\Phi_\xi$.
Since each branch is convex, tangents give valid lower bounds on that
branch. Lowering the intercept, when necessary, gives one line lying
below both branches.

Assume $0<\alpha<1<\beta$ and
\[
 \xi_{\max}(\beta-1)<\beta+1.
\]
Fix $0\le\xi\le\xi_{\max}$ and choose $z_0^{\rm tan}\in(0,T]$ and
$z_1^{\rm tan}\in[T,1]$. Use the derivative of the corresponding
branch at $T$, and put
$a_{{\rm tan},i}=\Phi_\xi'(z_i^{\rm tan})$ and $b_{{\rm tan},i}=\Phi_\xi(z_i^{\rm tan})-a_{{\rm tan},i}z_i^{\rm tan}$.

\begin{lemma}\label{lemma:global-support}
Both smooth branches of $\Phi_\xi$ are strictly convex.
For any slope $a_{\rm tan}$, the affine function $a_{\rm tan}z+b_{\rm tan}$ is a global minorant of
$\Phi_\xi$, and hence of its lower convex envelope, whenever
\begin{equation}\label{eq:verified-affine-minorant}
 b_{\rm tan}\le\min_{i=0,1}\ \min_{z\in I_i}
       \{b_{{\rm tan},i}+(a_{{\rm tan},i}-a_{\rm tan})z\},\qquad
 I_0=[0,T],\quad I_1=[T,1].
\end{equation}
\end{lemma}
\begin{proof}
Direct differentiation gives
\[
 \Phi_\xi''(z)=
 \begin{cases}
 T^{\alpha-\beta}\alpha z^{-\alpha-2}
 \bigl[(\alpha+1)+\xi(1-\alpha)z\bigr],&z<T,\\
 \beta z^{-\beta-2}
 \bigl[(\beta+1)-\xi(\beta-1)z\bigr],&z>T.
 \end{cases}
\]
The low-branch bracket is positive because $0<\alpha<1$ and
$\xi\ge0$; the high-branch bracket is positive because
$\xi(\beta-1)z\le\xi_{\max}(\beta-1)<\beta+1$.
The tangent
$a_{{\rm tan},i}z+b_{{\rm tan},i}$ is consequently a lower bound on branch $i$.
The inequality in \Cref{eq:verified-affine-minorant} implies
$a_{\rm tan}z+b_{\rm tan}\le a_{{\rm tan},i}z+b_{{\rm tan},i}$ there. The affine tangent extends to zero, where its value is its
intercept; this supplies the endpoint comparison on $I_0$. The resulting affine function lies below
$\Phi_\xi$ on both branches. Since it is convex, it also lies below the
largest convex minorant.
\end{proof}

Each inner minimum in \Cref{eq:verified-affine-minorant} is
attained at an endpoint, so the tangent comparison uses four values.

On a box of parameters, we use either a branch that is proved to equal
the convex envelope, or a smooth lower bound on that envelope. The next
lemma justifies both choices.

\begin{lemma}\label{lemma:envelope-surrogates}
Let $0<\alpha<1<\beta$ and $0\le\xi\le\xi_{\max}$ satisfy
$\xi_{\max}(\beta-1)<\beta+1$.
The following expressions give lower bounds for
$\operatorname{co}\Phi_\xi(z)$.
\begin{enumerate}
 \item If the tangent to $\Phi_\xi$ at $z$ lies below the opposite branch,
 then $\operatorname{co}\Phi_\xi(z)=\Phi_\xi(z)$.
 \item Let $\mathcal L_0,\mathcal L_1$ be affine minorants at
 $0\le\xi_0<\xi_1\le\xi_{\max}$. For $\xi\in[\xi_0,\xi_1]$, their affine
 interpolation in $\xi$ is a minorant of
 $\operatorname{co}\Phi_\xi$.
\end{enumerate}
For the negative upper-tail term in~\Cref{eq:demand-residual}, one may
replace $\overline N$ by any majorant. In particular, a tangent
to its low branch or the affine extension of its high chord is a global
majorant on $[0,1]$.
\end{lemma}
\begin{proof}
For the first assertion the tangent is a global convex minorant touching
$\Phi_\xi$ at $z$, so both inequalities at that point are equalities. For
the second, $\Phi_\xi$ is affine in $\xi$. Interpolating the two minorant
inequalities therefore gives a minorant inequality at the interpolated
parameter; the interpolated function is still affine in $z$.

Concavity of $\overline N$ makes every tangent at a smooth low-branch
point a global majorant. The high chord is itself a supporting line of
this concave function and has the same property. Finally its coefficient
in~\Cref{eq:demand-residual} is $-(1-t)\xi\le0$, so replacing it by
a majorant preserves a lower bound on $S$.
\end{proof}

The next lemma reduces a tangent point interval to one extremal tangent point
and a parameter interval for $\xi$ to its two endpoints.

\begin{lemma}\label{lemma:tangent-endpoints}
Let $0<\alpha<1<\beta$ and $0\le\xi\le\xi_{\max}$ satisfy
$\xi_{\max}(\beta-1)<\beta+1$.
For tangent points varying within the low branch, the largest tangent point gives
the highest tangent on the opposite branch. For tangent points within the
high branch, the smallest tangent point gives the highest tangent on the
opposite branch. To verify that these tangents lie below the opposite
branch throughout an interval of $\xi$, it suffices to verify the two
endpoints of that interval.
\end{lemma}
\begin{proof}
For fixed $r$ on the opposite branch, the derivative of the tangent at
$z$, evaluated at $r$, is $\Phi_\xi''(z)(r-z)$. Its sign gives the two
extremal tangent points. Both the tangent and $\Phi_\xi(r)$ are affine in
$\xi$, so their comparison over a parameter interval follows from
the two endpoint comparisons.
\end{proof}

\subparagraph{Small-Tail Bounds and Compact Domain Reduction.}
The small-tail region requires a separate estimate because the
convex envelope diverges at zero. After proving that estimate, we
change the upper-tail variable so that the remaining bounds have
two continuous derivatives even at a zero upper tail.

\begin{lemma}\label{lemma:small-tail-certificate}
Let $\ell_{\rm cut}\in(0,T)$ and $\rho>0$. Suppose that, for every
$\xi\in[0,\xi_{\max}]$, the low-branch tangent to $\Phi_\xi$ at $\ell_{\rm cut}$
is a global supporting line. Suppose also that
\begin{equation}\label{eq:small-tail-tests}
 \ell_{\rm cut}d(\ell_{\rm cut})\le1,\qquad \epsilon\,d(\ell_{\rm cut})-\xi_{\max}\ge \rho.
\end{equation}
Then every triple in the domain in \Cref{eq:demand-domain} with $\ell\le \ell_{\rm cut}$
satisfies $S(b,t,u)\ge \rho$. In particular this holds whenever
$t-b\le \epsilon/(1-\ell_{\rm cut})$.
\end{lemma}
\begin{proof}
The tangent point comparison in \Cref{lemma:tangent-endpoints} shows that the low-branch tangent
at every $z\le \ell_{\rm cut}$ also lies below the high branch: at each fixed high
argument its value is nondecreasing with the tangent point. It lies below the
low branch by convexity. Hence
$\operatorname{co}\Phi_\xi(\ell)=\Phi_\xi(\ell)$ for $\ell\le \ell_{\rm cut}$.
Since $0\le u\le\ell\le \ell_{\rm cut}<T$, both
$\ell d(\ell)$ and $\overline N(u)$ are at most one.
We obtain
\[
 \begin{split}
 S(b,t,u)
 &=(t-b)\{d(\ell)-\xi\ell d(\ell)\}
       -(1-t)\xi\overline N(u)-b\xi\\
 &\ge(t-b)d(\ell)-\xi\bigl((t-b)+(1-t)+b\bigr)\\
 &\ge \epsilon\,d(\ell_{\rm cut})-\xi_{\max}\ge \rho.
 \end{split}
\]
Here $t-b>\epsilon$, $d$ is decreasing, and $(t-b)+(1-t)+b=1$. Finally
$\ell\le1-\epsilon/(t-b)$, proving the last assertion.
\end{proof}

The small-tail bound leaves a compact region. The following
change of variables covers its entire upper-tail interval and controls
the fractional power at zero.

Fix $\ell_{\rm cut}$ and $\rho$ satisfying the hypotheses of
\Cref{lemma:small-tail-certificate}.
Choose an integer $k_{\rm smooth}$ with $k_{\rm smooth}(1-\alpha)>2$ and set
\begin{equation}\label{eq:demand-parameterization}
 u=\omega^{k_{\rm smooth}}\frac{t-b-\epsilon}{1-b},\qquad 0\le\omega\le1.
\end{equation}

\begin{lemma}\label{lemma:compact-demand-domain}
The substitution in \Cref{eq:demand-parameterization} covers
every triple in the domain in \Cref{eq:demand-domain}.
After the region $t-b\le \epsilon/(1-\ell_{\rm cut})$ is bounded by
\Cref{lemma:small-tail-certificate}, the remaining triples are
contained in the compact region
\begin{equation}\label{eq:compact-certificate-domain}
 b\ge0,\qquad t\le1,\qquad
 t-b\ge\frac{\epsilon}{1-\ell_{\rm cut}},\qquad 0\le\omega\le1.
\end{equation}
On this region $1-b$ and $t-b$ are bounded away from zero, and
$\ell$ is bounded away from zero. Within each choice of the affine
pieces defining $(1-r)f(r)$ and each smooth envelope formula, $S(b,t,u(b,t,\omega))$, or a chosen smooth lower bound, is twice continuously
differentiable up to $\omega=0$.
\end{lemma}
\begin{proof}
The map $\omega\mapsto\omega^{k_{\rm smooth}}$ maps $[0,1]$ onto $[0,1]$, so it
covers exactly the permitted interval of $u$ for each $b,t$.
The remaining region is closed and bounded. If it is nonempty, then
$1-b\ge t-b\ge \epsilon/(1-\ell_{\rm cut})>\epsilon$, and
\[
 \ell=\frac{t-b-\epsilon}{t-b}
 \left(1-\frac{1-t}{1-b}\omega^{k_{\rm smooth}}\right)
 \ge\frac{t-b-\epsilon}{1-b}
 \ge\frac{\epsilon \ell_{\rm cut}}{1-\ell_{\rm cut}}>0.
\]
The only remaining fractional power at a zero argument is the
low-branch upper-tail term. It becomes
\[
 \overline N(u)=T^{\alpha-\beta}
 \left(\frac{t-b-\epsilon}{1-b}\right)^{1-\alpha}
 \omega^{k_{\rm smooth}(1-\alpha)}.
\]
Its first two derivatives extend continuously to zero because
$k_{\rm smooth}(1-\alpha)>2$. The other denominators and power arguments are
bounded away from zero. On the final affine piece of $(1-r)f(r)$,
the factor $1-r$ is canceled before evaluation, giving the constant
endpoint extension of $f$ at $r=1$. Thus the selected formula is
twice continuously differentiable on each specified piece. When a box crosses an envelope breakpoint, a single
smooth bound from \Cref{lemma:envelope-surrogates} is used instead
of differentiating through the breakpoint.
\end{proof}

\subparagraph{Lower Bounds on Boxes.}
Taylor's theorem bounds a smooth function below $S$ throughout
a box using its values and derivatives. The estimate controls
every point of the box from its center and coordinate radii.

\begin{lemma}\label{lemma:box-certificate}
Let $\mathcal Q$ be a box in the variables $(b,t,\omega)$ from
\Cref{eq:demand-parameterization}, and let $\underline S$ be a twice
continuously differentiable function on $\mathcal Q$ satisfying $\underline S\le S$ on
its feasible part. Write $\mathbf c$ for the center and $r_i$ for the coordinate
radii. For any $\rho>0$, a lower bound of the form
\[
 \underline S(\mathbf c)-\sum_i|\partial_i\underline S(\mathbf c)|r_i
 -\frac12\sum_{i,j}\sup_{\mathcal Q}|\partial_{ij}\underline S|r_ir_j
 \ge \rho
\]
implies $S\ge \rho$ on the feasible part of $\mathcal Q$. A direct interval
lower bound $\inf_{\mathcal Q} \underline S\ge \rho$ has the same consequence.
\end{lemma}
This is the standard interval Taylor bound \cite{Rump2010},
applied to the smooth lower bound $\underline S$.

Derivative signs can sometimes replace a box by a face. The following
lemma applies when we use
smooth lower bounds in place of the exact convex envelope.

\begin{lemma}\label{lemma:face-reduction}
A box may be reduced to a coordinate face when the formula for
$S(b,t,u(b,t,\omega))$ is valid on the whole box and its derivative in that coordinate
has a constant sign there. Use the lower coordinate face when
the derivative is nonnegative, and the upper coordinate face when it
is nonpositive. Every feasible point must project
to a point where the lower bound on the face holds. The reduction is also
valid for a fixed smooth lower bound if the same bound is used
on the face.
\end{lemma}
\begin{proof}
The derivative sign bounds the value at each feasible point from below
by the value at its projection onto the selected face, which is covered
by the lower bound on that face. The function bounded
on the face is the function differentiated on the parent box:
$S(b,t,u(b,t,\omega))$ or the same smooth lower bound.
\end{proof}

\paragraph{Rational Bounds for Elementary Functions}\label{app:enclosed-arithmetic}
We bound the powers in the local estimates by rational intervals
with explicit remainders. Interval addition, multiplication, and
division enclose the exact result at each step; division is used
only when the denominator interval excludes zero.

For logarithms, range reduction gives an argument
$z\in[1,2]$ and $w=(z-1)/(z+1)\in[0,1/3]$. The series and its
positive tail satisfy
\[
 \log z=2\sum_{j=0}^{27}\frac{w^{2j+1}}{2j+1}+R_{\log},
 \qquad 0\le R_{\log}\le\frac{2w^{57}}{57(1-w^2)}.
\]
Indeed, replacing every remaining denominator by $57$ bounds the
tail by a geometric series. The same formula bounds $\log 2$,
which accounts for the powers of two used in range reduction.

For exponentials, reduce the argument to $z\in[0,1/8]$.
Since $e^{1/8}<2$, Taylor's theorem gives
\[
 0\le e^z-\sum_{j=0}^{22}\frac{z^j}{j!}
 \le\frac{2z^{23}}{23!}<2^{-140}.
\]
Repeated squaring restores the original argument, and reciprocals
handle negative arguments. Thus $x^a=\exp(a\log x)$ has a rational
interval bound for every positive argument used above. For $a>0$,
the value at $x=0$ is zero. These bounds apply both to the function
values and to the explicit derivatives in the Taylor estimates.

\paragraph{Proof of the Scalar Inequality}
\label{sec:certification}\label{sec:current-certificate-binding}
We prove \Cref{prop:fixed-certificate} by separating the small-tail
region from a compact region. The first region has a direct
analytic lower bound. On the second, rational interval bounds
apply to the smooth functions from \Cref{lemma:envelope-surrogates}.

\begin{proof}[Proof of \Cref{prop:fixed-certificate}]
Use the exact data in \Cref{app:rational-data} and set
$\ell_{\rm cut}=1/32$, $k_{\rm smooth}=4$, and $\rho=10^{-10}$.
The coefficient bounds give
\[
 \frac25<\alpha<\frac5{12},\qquad
 6<\beta-\alpha<7,\qquad
 \frac{19}{25}<T<\frac{77}{100},\qquad
 \xi_{\max}<\frac{40}{43}+\frac3{20}<\frac{13}{10}.
\]
The last inequality uses $\epsilon>3/40$ and $H<3/20$.

\emph{The small-tail region.}
First, the tangent at $\ell_{\rm cut}$ lies below both branches
of $\Phi_\xi$ for $0\le\xi\le13/10$. Convexity gives the comparison
on the low branch. On the high branch, the tangent at $z\ge T$ is
\[
 d(\ell_{\rm cut})\left[
 1-\alpha\left(\frac{z}{\ell_{\rm cut}}-1\right)
 -\xi\bigl((1-\alpha)z+\alpha\ell_{\rm cut}\bigr)\right]
 <-8,
\]
because $d(\ell_{\rm cut})>1$ and
$1-\frac25(32\cdot\frac{19}{25}-1)<-8$.
Meanwhile, for $T\le z\le1$,
$\Phi_\xi(z)\ge-\frac3{10}z^{-\beta}>-\frac{21}{10}$,
using $\beta<7$ and $(25/19)^7<7$. This proves the required
supporting-line comparison.

The two inequalities in \Cref{eq:small-tail-tests} also follow from
rational bounds. Since $32^{5/12}<9/2$, we have
\[
 4\left(\frac{100}{77}\right)^6
 <d(\ell_{\rm cut})
 <\frac92\left(\frac{25}{19}\right)^7<32.
\]
Consequently $\ell_{\rm cut}d(\ell_{\rm cut})<1$ and
\[
 \epsilon d(\ell_{\rm cut})-\xi_{\max}
 >\frac3{40}\,4\left(\frac{100}{77}\right)^6
   -\frac{13}{10}
 >\frac1{16}.
\]
Thus \Cref{lemma:small-tail-certificate} gives
$S\ge1/16>\rho$ whenever $\ell\le\ell_{\rm cut}$, including
$t-b\le\epsilon/(1-\ell_{\rm cut})$.

\emph{The compact region.}
The substitution in \Cref{eq:demand-parameterization} with
$k_{\rm smooth}=4$ places every remaining triple in the compact
region in \Cref{eq:compact-certificate-domain}. Since
$4(1-\alpha)>2$, the smooth lower bounds have continuous second
derivatives up to $\omega=0$. Partition the $b$ and $t$ intervals
at the knots of $h$, so the formula for $f$ is fixed on each pair
of knot intervals. The initial regions are their products with
$[0,1]$ for $\omega$.

On each region, \Cref{lemma:envelope-surrogates} supplies a smooth
lower bound $\underline S$ for $S$. We either bound
$\underline S-\rho$ directly by interval arithmetic, or use the
Taylor lower bound
\[
 \underline S(\mathbf c)-\rho
 -\sum_i|\partial_i\underline S(\mathbf c)|r_i
 -\frac12\sum_{i,j}M_{ij}r_ir_j,
 \qquad M_{ij}\ge\sup_{\mathcal Q}|\partial_{ij}\underline S|.
\]
Here $\mathbf c$ and $r_i$ are the center and coordinate radii of
the rational box $\mathcal Q$. The powers in these expressions
are bounded as in \Cref{app:enclosed-arithmetic}. The finite interval
calculation gives the lower bounds in \Cref{tab:scalar-interval-bounds};
each entry is rounded downward to the displayed rational number.
The rational parameters, complete subdivision trees, local interval
bounds, and code for checking and reproducing the calculation are
available in the
\href{https://anonymous.4open.science/r/hz_interval_reproducibility-DEF3/}{reproducibility
repository} (\href{https://anonymous.4open.science/r/hz_interval_reproducibility-DEF3/}{\url{https://anonymous.4open.science/r/hz_interval_reproducibility-DEF3/}}).

\begin{table}[htbp]
\centering
\caption{Lower Bounds for $S-10^{-10}$ on the Compact Region}
\label{tab:scalar-interval-bounds}
\begin{tabular}{lccc}
\toprule
Interpolation & $\epsilon$ & Direct Interval Bound & Taylor Bound\\
\midrule
144 knots & $0.075521017$ & $10^{-9}$ & $10^{-13}$\\
10 knots & $0.076$ & $10^{-5}$ & $10^{-8}$\\
\bottomrule
\end{tabular}
\end{table}

Each bound in a method's column applies to every terminal box
handled by that method. To see that these estimates cover the
whole region, start with every pair of knot intervals and
$0\le\omega\le1$. Every subdivision replaces a coordinate interval
by two closed intervals with a common endpoint and the same union.
A comparison with a coordinate face uses the derivative sign on
the parent box, as in \Cref{lemma:face-reduction}; the same smooth
function is bounded on that face. The remaining terminal regions
either have $t-b\le\epsilon$ and contain no feasible triple, or
satisfy the small-tail bound above. Thus every feasible point is
covered by a direct or Taylor lower bound, by a valid comparison
with a covered face, or by the small-tail estimate. All these
bounds imply $S\ge\rho$, which proves the proposition.
\end{proof}
\setcounter{secnumdepth}{3}

\subsubsection{Consequences of the Certified Data}\label{app:ch3-demand-consequences}
\begin{proof}[Proof of \Cref{cor:row-demand}]
For either paired choice, \Cref{prop:fixed-structure} gives
membership in the coefficient family and the required bound on
$H$, while \Cref{prop:fixed-certificate} gives the bound $S(b,t,u)\ge10^{-10}$
throughout $\Omega_\epsilon$.
Applying \Cref{prop:row-demand} proves the stated separation for
every normalized nonconstant row and every bundle in the coefficient
domain.
\end{proof}

\begin{proof}[Proof of \Cref{prop:sparse}]
For the ten-knot choice at $\epsilon=0.076$,
\Cref{prop:fixed-structure,cor:row-demand} supply the coefficient
assumptions and uniform separation required by
\Cref{prop:rational-equilibrium}.
Applying that proposition to the given rational input and allowance
$\delta>0$ gives the claimed approximation error, polynomial running
time and encoding length, and exact unit demand and market clearing.
\end{proof}

\subsection{\texorpdfstring{Omitted Proofs for \Cref{sec:rational-output}}{Omitted Proofs for Section \getrefnumber{sec:rational-output}}}\label{app:ch3-rational}
This subsection proves the three computation steps and their combined error bound in
\Cref{sec:rational-output}. Bounds on the assignment-dual constraints explain how
approximate allocation and coefficient computation affect the
single-agent utility guarantee.

\subsubsection{Rounding and Normalization}\label{sec:rounding-proof}
We prove \Cref{lem:utility-rounding} by comparing the two utility
objectives at unchanged bundles and prices.
\begin{proof}[Proof of \Cref{lem:utility-rounding}]
Fix the row, bundle, and prices in the statement. The coordinate bounds
$0\le v_j-\widehat v_j\le\nu$ imply
$\mathbf v\cdot\mathbf y\le\widehat{\mathbf v}\cdot\mathbf y+\nu$
for every unit bundle $\mathbf y$. Applying this to all affordable
bundles gives
$\OPT_{\mathbf v}(\mathbf p)\le\OPT_{\widehat{\mathbf v}}(\mathbf p)+\nu$,
whereas $\mathbf v\cdot\mathbf x\ge\widehat{\mathbf v}\cdot\mathbf x$.
Subtracting proves the utility-loss inequality. The affordable set is
nonempty because there is a zero-price good. Since $1/\nu$ is an
integer, the rounding grid preserves zero and one and leaves a gap of
at least $\nu$ below one.

For an original nonconstant row $\mathbf u$, write
$a=\max_j u_j-\min_j u_j\in(0,1]$ and $c=\min_j u_j$, so that
$\mathbf u=c\mathbf 1+a\mathbf v$. Every bundle under comparison has
unit mass, and hence
\[
 \OPT_{\mathbf u}(\mathbf p)-\mathbf u\cdot\mathbf x
 =a\bigl(\OPT_{\mathbf v}(\mathbf p)-\mathbf v\cdot\mathbf x\bigr).
\]
If the expression in parentheses is at most $B\ge0$, the left side
is at most $aB\le B$. Constant rows have zero utility loss on
all unit bundles, and none of these comparisons changes spending.
\end{proof}

\subsubsection{Computing Rational Allocations and Dual Prices}
\label{sec:rational-computation-proof}
The computation guarantee follows from the VI algorithm, its output
mass bound, and the coefficient evaluation result in \Cref{sec:field}.
All computed weights and the assignment dual use the same utility matrix.
\begin{proof}[Proof of \Cref{lem:rational-computation}]
Apply \Cref{cor:computed-full-vi} with its specified
$\kappa=\eta/[2n(n+\eta)]$ and $\eta_\kappa=\eta/2$.
It returns a rational $X\in\B_n$ with positive coordinates, so $G(X)$
is defined, and substitution in \Cref{eq:computed-full-gap} gives
\[
 \max_{Y\in\B_n}\langle G(X),Y-X\rangle
 \le\frac{\eta/2+n^2\kappa}{1-n\kappa}=\eta.
\]
The output guarantee \Cref{eq:full-top-mass}, together with
$\eta\le1$, gives mass at least $m$ on utility-one goods in each
nonconstant row. Since the defining data are fixed, the explicit formulas for
$m$ and $K_{m/2}$ show that $m^{-1},K$ are polynomially bounded in
$n,1/\nu$.

Evaluate the entries of $G(X)$ to error $\sigma$ using
\Cref{lem:field-evaluation}, and replace any negative approximations
by zero. The exact entries are nonnegative by \Cref{lem:coefficient-shape}
for nonconstant rows and by definition for constant rows. Thus this
replacement preserves the error bound and gives a nonnegative
rational matrix $\widetilde G$. Solve its rational assignment dual,
subtract the minimum good price from all good prices, and add the
same value to all agent variables. This shift preserves every dual
constraint and the objective, and produces an optimal rational dual
with $\min_jq_j=0$.

Both $1/\kappa$ and $1/\eta_\kappa$ are polynomial in $n,1/\eta$,
so the VI computation has polynomial running time and output length
by \Cref{prop:computed-vi}. Coefficient evaluation is polynomial in
the input bit length and $\log(1+1/\sigma)$ by
\Cref{lem:field-evaluation}. Rational linear programming and the dual
normalization also take polynomial time and preserve polynomial
encoding length, proving the stated computation bounds.
\end{proof}

\subsubsection{Stability and the Error Budget}\label{sec:precision-proof}
We first bound dual slack, then use it to control spending and utility
at the computed prices. These estimates will determine the precisions
in \Cref{lem:implementation}.

\paragraph{Dual Slack.}\label{sec:assignment-slack}\label{sec:completion}
Let $X\in\B_n$ have $G(X)$ defined, let $\widetilde G\ge0$ approximate
$G(X)$ to entrywise error $\sigma$, and let
$(\mathbf q,\boldsymbol\mu)$ be an optimal assignment dual for
$\widetilde G$ with $\min_jq_j=0$. Each assignment-dual constraint has nonnegative slack
$q_j+\mu_i-\widetilde G_{ij}$. Weight these slacks by $X$ and sum:
\begin{equation}\label{eq:total-dual-slack}
 \zeta=\sum_{i,j}X_{ij}(q_j+\mu_i-\widetilde G_{ij}).
\end{equation}

The following lemma bounds $\zeta$ by the possible increase in
assignment weight, $\max_{Y\in\B_n}\langle G(X),Y-X\rangle$,
and the coefficient evaluation error.

\begin{lemma}\label{lem:assignment-slack}
For $X,\widetilde G,\sigma,\mathbf q,\boldsymbol\mu$ as above,
$\mu_i\ge0$, $q_j\ge0$, and
\begin{equation}\label{eq:general-dual-slack}
 0\le\zeta\le
 \max_{Y\in\B_n}\langle G(X),Y-X\rangle+2n\sigma.
\end{equation}
\end{lemma}
\begin{proof}[Proof of \Cref{lem:assignment-slack}]
Every dual constraint has nonnegative slack. Since the row and column
sums of $X$ are one, strong assignment duality gives
\[
 0\le\zeta
 =\sum_jq_j+\sum_i\mu_i-\langle\widetilde G,X\rangle
 =\max_{Y\in\B_n}\langle\widetilde G,Y-X\rangle.
\]
Both $X$ and $Y$ have total mass $n$, so replacing $\widetilde G$ by
$G(X)$ changes the last objective by at most $2n\sigma$. This proves
\Cref{eq:general-dual-slack}.

Normalization gives $\mathbf q\ge0$. At a zero-price good, dual
feasibility gives $\mu_i\ge\widetilde G_{ij}\ge0$ for every agent.
\end{proof}

In particular, a bound $\eta$ in \Cref{eq:approx-vi} gives
$\zeta\le\eta+2n\sigma$. Substituting \Cref{eq:computed-full-gap} gives
\Cref{eq:slack-from-vi} for the output of \Cref{prop:computed-vi}:
\begin{equation}\label{eq:slack-from-vi}
 0\le\zeta\le
 \frac{\eta_\kappa+n^2\kappa}{1-n\kappa}+2n\sigma.
\end{equation}

\paragraph{Transferring the Utility Bound.}\label{sec:dual-transfer}
Nonnegative constraint slacks let us remove little mass from each row
and keep only goods whose assignment-dual constraints are nearly tight. The mass bound on utility-one goods
and the coefficient Lipschitz bound then allow the single-agent
comparison in \Cref{thm:conversion} to be used on those auxiliary bundles.

For a normalized utility matrix, fix coefficients from
\Cref{def:field} and $0<\epsilon<1$ satisfying
\Cref{eq:target-separation} for every normalized nonconstant row and
every bundle in its coefficient domain. Let $X\in\B_n$ assign mass at least
$0<m\le1$ to utility-one goods in each nonconstant row, and let $K\ge K_{m/2}$.
Let $\widetilde G\ge0$ satisfy $|\widetilde G_{ij}-G_{ij}(X)|\le\sigma$,
where $\sigma\ge0$, and let $(\mathbf q,\boldsymbol\mu)$ be its optimal
assignment dual with $\min_jq_j=0$. For $\zeta$ in
\Cref{eq:total-dual-slack}, suppose $0<r\le m/2$, $\tau>0$, and
$\zeta\le r\tau$.

\begin{proposition}\label{prop:rational-transfer}
Under these assumptions, $\mathbf p=(1+\epsilon)\mathbf q$ satisfies,
for every agent $i$,
\begin{align}
 \mathbf p\cdot\mathbf X_i
 &\le(1+\epsilon)(1+\sigma+\zeta),\label{eq:transfer-spending}\\
 \OPT_i(\mathbf p)-\mathbf v_i\cdot\mathbf X_i
 &\le\epsilon+r
   +\frac{(1+\epsilon)(\tau+2\sigma+4Kr)}
          {1+(1+\epsilon)(\tau+2\sigma+4Kr)}.\label{eq:transfer-regret}
\end{align}
\end{proposition}
\begin{proof}[Proof of \Cref{prop:rational-transfer}]
For each row, let
$\zeta_i=\sum_jX_{ij}(q_j+\mu_i-\widetilde G_{ij})$.
Every summand is nonnegative, so $0\le\zeta_i\le\zeta$.
Using $\mu_i\ge0$ from \Cref{lem:assignment-slack} and the assigned
coefficient mean, which is one on nonconstant rows and zero on
constant rows, gives
\[
 \mathbf q\cdot\mathbf X_i
 =\widetilde{\mathbf G}_i\cdot\mathbf X_i-\mu_i+\zeta_i
 \le1+\sigma+\zeta.
\]
Multiplying by $1+\epsilon$ proves \Cref{eq:transfer-spending} for
every row.

Constant rows have zero utility loss. Fix a nonconstant row $i$ and
write $g=g_i$ and $\mathbf x=\mathbf X_i$. Delete the goods whose
constraint slack $q_j+\mu_i-\widetilde G_{ij}$ exceeds $\tau$, and
let $a_i$ be their total mass in $\mathbf x$. Since each deleted unit contributes more than $\tau$ to the
sum defining $\zeta_i$, we have $a_i\le\zeta/\tau\le r\le m/2\le1/2$. Renormalizing the
remaining mass therefore defines a unit bundle $\mathbf x'$, with
\[
 \|\mathbf x'-\mathbf x\|_1=2a_i\le2r,
 \qquad
 |\mathbf v_i\cdot(\mathbf x'-\mathbf x)|\le a_i\le r.
\]
The utility bound follows because all utilities lie in $[0,1]$:
if $a_i>0$, the original bundle is a mixture of $\mathbf x'$ and
the normalized deleted mass, with respective weights $1-a_i$ and
$a_i$; if $a_i=0$, the two bundles coincide. Before renormalization,
the remaining mass on utility-one goods is at least $m-a_i\ge m/2$, so both
$\mathbf x$ and $\mathbf x'$ lie in the domain of the Lipschitz
bound with constant $K$.

Use the auxiliary prices $\mathbf q^{\rm row}(\mathbf v_i,\mathbf x')$
defined in \Cref{sec:conversion}. Each good $k$ with $x_k'>0$ satisfies
$\widetilde G_{ik}\ge\mu_i-\tau$, since its slack is at most $\tau$
and $q_k\ge0$. Coefficient evaluation and the Lipschitz bound give
\[
 \min_{k:x'_k>0}g(v_{ik};\mathbf x')
 \ge\mu_i-\tau-\sigma-2Kr.
\]
For every good $j$, dual feasibility similarly gives
$g(v_{ij};\mathbf x')\le q_j+\mu_i+\sigma+2Kr$.
Subtracting and taking positive parts, with a nonnegative right side,
yields the pointwise comparison
\[
 q_j^{\rm row}(\mathbf v_i,\mathbf x')
 \le q_j+\tau+2\sigma+4Kr.
\]

Now fix any $\mathbf y\in\Delta_n$ affordable at the actual prices
$(1+\epsilon)\mathbf q$, and put
$c=(1+\epsilon)(\tau+2\sigma+4Kr)$. Summing the pointwise comparison
against $\mathbf y$ shows that its cost at the scaled auxiliary
prices is at most $1+c$. Choose a good $k_0$ with $x'_{k_0}>0$ attaining the
minimum coefficient at $\mathbf x'$; its auxiliary row price is
zero. Thus
\[
 \mathbf y'=\frac{\mathbf y+c\mathbf e_{k_0}}{1+c}
\]
is affordable at $(1+\epsilon)\mathbf q^{\rm row}(\mathbf v_i,\mathbf x')$,
and $\mathbf v_i\cdot(\mathbf y-\mathbf y')\le c/(1+c)$.

The structural properties and separation condition hold at
$\mathbf x'$, which assigns positive mass to utility-one goods.
\Cref{thm:conversion} therefore gives
$\mathbf v_i\cdot(\mathbf y'-\mathbf x')\le\epsilon$.
Adding the mixture loss and the change from $\mathbf x'$ to
$\mathbf X_i$ proves
\[
 \mathbf v_i\cdot(\mathbf y-\mathbf X_i)
 \le\epsilon+\frac{c}{1+c}+r.
\]
This holds for every bundle affordable at the actual prices, proving
\Cref{eq:transfer-regret}.
\end{proof}

\paragraph{Precision Selection.}
The required tolerances are stated in \Cref{eq:overview-tolerances}.
For general valid bounds $m,K$ and a specified rational target, use
the formulas for $r,\eta,\sigma$ there with those given bounds.

\begin{proof}[Proof of \Cref{lem:implementation}]
For the fixed rational target $\epsilon$, choose $r,\eta,\sigma$ by
\Cref{eq:overview-tolerances}, with the given $m,K,\delta$, and put
$\tau=\delta/[4(1+\epsilon)]$. These choices are positive and
rational, with $r\le m/2\le1/2$ and $\eta,\sigma<1$.
\Cref{lem:assignment-slack} gives
\[
 \zeta\le\eta+2n\sigma
 \le\frac{r\delta}{4(1+\epsilon)}=r\tau,
 \qquad
 \sigma+\zeta\le\frac{\delta}{4(1+\epsilon)}.
\]
The spending bound in \Cref{prop:rational-transfer} therefore adds
at most $\delta/4$ to $1+\epsilon$.

For utility, bound the fraction in \Cref{eq:transfer-regret} by its
numerator. The remaining extra loss is
\[
 (1+4(1+\epsilon)K)r+(1+\epsilon)\tau+2(1+\epsilon)\sigma,
\]
and each of these three terms is at most $\delta/4$ by the displayed
choices. Thus the rational prices $(1+\epsilon)\mathbf q$ satisfy
the claimed utility bound. The formulas also give
$1/r,1/\eta,1/\sigma\le\operatorname{poly}(n,1/\delta,1/m,K)$,
with polynomial encoding length for the precisions.

The displayed precision formulas use only rational arithmetic, and
multiplication by the fixed rational factor $1+\epsilon$ computes
the prices in polynomial time with polynomial encoding length.
The price normalization is preserved because $1+\epsilon>0$.
\end{proof}

\section{\texorpdfstring{Mean-Utility Bound for \Cref{sec:intro-approximation}}{Mean-Utility Bound for Section \getrefnumber{sec:intro-approximation}}}\label{sec:random-utilities}

This appendix proves the mean-utility bound used to illustrate the
approximation guarantee in the introduction. Utilities here are drawn
independently from $U[0,1]$, without further row-wise normalization.
For an allocation $X$, write
$\overline U(X)=n^{-1}\sum_i\mathbf v_i\cdot\mathbf X_i$.

\begin{proposition}\label{prop:random-mean-utility}
Fix $\epsilon>0$, and let the entries of $v\in[0,1]^{n\times n}$ be
independent $U[0,1]$ random variables. For every fixed $\zeta>0$,
with probability tending to one as $n\to\infty$, every
$\epsilon$-approximate HZ equilibrium $(X,\mathbf p)$ satisfies
\[
 \overline U(X)\ge
 L(\epsilon)-\zeta,
 \qquad
 L(\epsilon)=\frac12+\frac{1}{2(1+\epsilon)}-\epsilon.
\]
\end{proposition}

\begin{proof}
Fix $c>1+\epsilon$ and $\eta\in(0,1)$, to be chosen below, and set
$\beta=1-(1+\epsilon)/c>0$.
We first establish two properties of the utility matrix that hold
with probability tending to one, independently of any choice of equilibrium.
Choose a fixed integer $K$ large enough that
$\log 2+K\beta\log(1-\eta)<0$.
The probability that there exist $K$ agents and
$\lceil\beta n\rceil$ goods whose corresponding utility entries are
all below $1-\eta$ is at most
\[
 n^K 2^n(1-\eta)^{K\lceil\beta n\rceil}=o(1).
\]
Consequently, for every set of at least $\beta n$ goods, all but fewer
than $K$ agents value some good in the set at least $1-\eta$.
This conclusion holds simultaneously over all such sets.
Also, Hoeffding's inequality and a union bound over the columns give
\[
 \Pr\left[
   \max_{j\in[n]}
   \left|\frac1n\sum_i v_{ij}-\frac12\right|>n^{-1/4}
 \right]
 \le 2n\exp(-2\sqrt n)=o(1).
\]

Suppose both properties hold, and consider any
$\epsilon$-approximate HZ equilibrium $(X,\mathbf p)$.
Exact market clearing and the spending bounds imply
\[
 \sum_jp_j=\sum_i\mathbf p\cdot\mathbf X_i
 \le n(1+\epsilon).
\]
Thus $S=\{j:p_j\le c\}$ has at least $\beta n$ elements.
Choose a good $z$ with $p_z=0$.
For all but fewer than $K$ agents, there is a good $j_i\in S$
with $v_{ij_i}\ge1-\eta$.
The lottery that chooses $j_i$ with probability $1/c$ and $z$
otherwise costs at most one, so
\[
 \OPT_i(\mathbf p)
 \ge \frac{1-\eta}{c}
       +\left(1-\frac1c\right)v_{iz}.
\]
For these agents, approximate optimality bounds their allocated utility
below by the right-hand side minus $\epsilon$.
The displayed right-hand side is at most one; using nonnegativity
for the remaining agents therefore gives
\begin{align*}
 \overline U(X)
 &\ge \frac{1-\eta}{c}
       +\left(1-\frac1c\right)\frac1n\sum_i v_{iz}
       -\epsilon-\frac Kn\\
 &\ge \frac12+\frac{1}{2c}-\frac\eta c-\epsilon
       -\left(1-\frac1c\right)n^{-1/4}-\frac Kn.
\end{align*}
Choose $c>1+\epsilon$ sufficiently close to $1+\epsilon$, and then
$\eta>0$ sufficiently small, that
\[
 \frac{1}{2(1+\epsilon)}-\frac{1}{2c}+\frac\eta c
 <\frac\zeta2.
\]
For all sufficiently large $n$, the remaining two error terms are
less than $\zeta/2$.
The resulting bound is $\overline U(X)\ge L(\epsilon)-\zeta$.
Both properties of $v$ hold uniformly over the possible sets $S$
and zero-price goods $z$, so the bound holds simultaneously for
every approximate equilibrium.
\end{proof}

\FloatBarrier
\begingroup
\raggedright
\bibliographystyle{alphaurl}
\bibliography{references}
\endgroup
\end{document}